\documentclass[envcountsame,11pt]{llncs} 

\usepackage[margin=1in]{geometry}

\usepackage[disable]{todonotes}

\usepackage{amsfonts}
\usepackage{amsmath}
\usepackage{amssymb}
\usepackage{xspace}

\usepackage{tikz}

\usetikzlibrary{shapes,positioning,calc,snakes}
\tikzstyle{state}=[minimum size=1.5cm,inner sep=0cm,draw,regular polygon, regular polygon sides=4]
\tikzstyle{dummy_state}=[minimum size=1.8cm,inner sep=0cm,regular polygon, regular polygon sides=4]
\tikzstyle{clock_state}=[minimum size=1.8cm,inner sep=0cm,draw,regular polygon, regular polygon sides=4]
\tikzstyle{avg_state}=[minimum size=1.3cm,inner sep=0cm,draw,circle]
\tikzstyle{arg}=[->,style=dashed,semithick]
\tikzset{every loop/.style={min distance=5mm,in=0,out=60,looseness=1}}

\newcommand{\PSPACE}{\ensuremath{\textsc{PSPACE}}\xspace}
\newcommand{\org}{\ensuremath{\textsc{Or}}\xspace}
\newcommand{\notg}{\ensuremath{\textsc{Not}}\xspace}
\newcommand{\gadget}{\ensuremath{\textsc{Gadget}}\xspace}
\newcommand{\eat}[1]{}
\newcommand{\sink}{\ensuremath{si}}
\newcommand{\half}{\ensuremath{\frac{1}{2}}}
\newcommand{\bit}{\ensuremath{\textsc{Bit}}}
\newcommand{\shift}{\ensuremath{\textsc{Shift}}\xspace}
\newcommand{\lsz}{\ensuremath{\textsc{Lsz}}}
\usepackage{xspace}
\newcommand{\Index}{\textsc{Iter}\xspace}
\newcommand{\sigmagj}{\ensuremath{{\sigma^{g(j)}}}\xspace}
\DeclareMathOperator{\valgjnobar}{Val^{\sigma^{g(j)}}\xspace}
\DeclareMathOperator{\valgj}{\overline{Val}^{\sigma^{g(j)}}\xspace}
\newcommand{\clockappeal}{\ensuremath{\bigl(\frac{1}{2} - \frac{1}{4i}\bigr)}\xspace}

\newcommand{\lonej}{1.6}
\newcommand{\lj}{3.4}
\newcommand{\rj}{3.2}
\newcommand{\rjprime}{3.3} 
\newcommand{\ro}{0.8} 
\newcommand{\rol}{0.76} 
\newcommand{\bl}{3.1} 
\newcommand{\dnc}{1.01}  
\newcommand{\dncu}{1.1} 
\newcommand{\magic}{0.06} 
\newcommand{\xj}{0.9}
\newcommand{\aj}{0.95}
\newcommand{\ajprime}{3.5 + \frac{1}{2 d(i)}}

\DeclareMathOperator{\val}{Val}
\DeclareMathOperator{\appeal}{Appeal}

\DeclareMathOperator{\inp}{I}
\DeclareMathOperator{\poly}{poly}
\DeclareMathOperator{\const}{Constr}

\newcommand{\bitswitch}{\textsc{BitSwitch}\xspace}
\newcommand{\circuitvalue}{\textsc{CircuitValue}\xspace}
\newcommand{\actionswitch}{\textsc{ActionSwitch}\xspace}
\newcommand{\dms}{\textsc{DantzigMdpSol}\xspace}
\newcommand{\basisentry}{\textsc{BasisEntry}\xspace}
\newcommand{\dls}{\textsc{DantzigLpSol}\xspace}

\title{\LARGE The Complexity of the Simplex Method}
\author{\large John Fearnley \and Rahul Savani}
\institute{\normalsize University of Liverpool}

\begin{document}

\maketitle

\begin{abstract}
The simplex method is a well-studied and widely-used pivoting method for solving linear programs. When Dantzig originally formulated the simplex method, he gave a natural pivot rule that pivots into the basis a variable with the most violated reduced cost. In their seminal work, Klee and Minty showed that this pivot rule takes exponential time in the worst case. We prove two main results on the simplex method. Firstly, we show that it is PSPACE-complete to find the solution that is computed by the simplex method using Dantzig's pivot rule. Secondly, we prove that deciding whether Dantzig's rule ever chooses a specific variable to enter the basis is PSPACE-complete. We use the known connection between Markov decision processes (MDPs) and linear programming, and an equivalence between Dantzig's pivot rule and a natural variant of policy iteration for average-reward MDPs. We construct MDPs and show PSPACE-completeness results for single-switch policy iteration, which in turn imply our main results for the simplex method.
\end{abstract}

\section{Introduction}

Linear programming is a fundamental technique in computer science, and the
simplex method is a widely used technique for solving linear programs. The
simplex method requires a \emph{pivot rule} that determines which variable is
pivoted into the basis in each step. Dantzig's original formulation of the
simplex method used a particularly natural pivot rule: in each step, the
non-basic variable with the most negative reduced cost is chosen to enter the
basis~\cite{D65}. We will call this \emph{Dantzig's pivot rule}. Klee and Minty
have shown that Dantzig's pivot rule takes exponential time in the worst
case~\cite{KM72}. 

The simplex method is a member of a of much wider class of \emph{local search}
algorithms. The complexity class PLS, which was introduced by Johnson,
Papadimitriou, and Yannakakis, captures problems where a locally optimal
solution can be found by a local search algorithm~\cite{JPY88}. PLS has matured
into a robust class, and there is now a wide range of problems that are known to
be PLS-complete~\cite{Joh07,MDT10}. It is widely believed that PLS-complete
problems do not admit polynomial-time algorithms.

To show that a problem lies in PLS, we must provide three polynomial-time 
functions: a function
$A$ that produces a candidate solution, a function $B$ that assigns an value to
each candidate solution, and a function~$C$ that for each candidate solution,
either produces a neighbouring candidate solution with higher value, or reports that no
such candidate solution exists. Thus, each PLS problem comes equipped with a
\emph{natural algorithm} that executes $A$ to find an initial candidate
solution, and then iterates $C$ until a local optimum is found. 
For some problems in PLS, it is known that it is PSPACE-complete to find the
solution that is computed by the natural algorithm~\cite{PSY90,SY91}.
So far, this is
only known to hold
for problems that
are \emph{tight PLS-complete}~(see, e.g., \cite{Yan97}), which is a stronger form of
PLS-completeness, or for problems that are suspected to be tight PLS-complete, such as
the local max-cut problem on graphs of degree four~\cite{MT10}.

Obviously, since linear programming is in P, it cannot be PLS-complete unless PLS=P.
Despite this fact, in the first main theorem of this paper, we show that it is
PSPACE-complete to compute the solution found by the simplex method
equipped with Dantzig's pivot rule. Given a linear
program $\mathcal{L}$, an initial basic feasible solution $b$, and a variable
$v$, the problem $\dls(\mathcal{L}, b, v)$ asks the following question: if
Dantzig's pivot rule is started at basis $b$, and finds an optimal solution $s$
of $\mathcal{L}$, is variable $v$ in the basis of $s$? We show the following
theorem, which holds regardless of the degeneracy resolution
rule\footnote{Degeneracy resolution is used to break ties when there is more
than one possible entering and/or leaving variable, e.g., when there are two
variables with the most negative reduced cost.} used by Dantzig's pivot rule.
\begin{theorem}
\label{thm:dls}
\dls is PSPACE-complete.
\end{theorem}


We can also understand Theorem~\ref{thm:dls} in the context of the complexity
class PPAD.
Linear programming is in the intersection of PLS and PPAD.
The Lemke-Howson algorithm is a complementary pivoting algorithm that finds an 
equilibrium of a bimatrix game, and it was the inspiration for the complexity class
PPAD~\cite{P94}.
In fact, a variant of the simplex method can be seen as a special case of Lemke's
algorithm for linear complementarity problems, which is closely related to the
Lemke-Howson algorithm.
The canonical
PPAD-complete problem is \emph{end of the line} which asks: given a succinct
encoding of an exponentially-large graph, where every vertex has in-degree and
out-degree at most $1$, and an initial vertex $v$ with in-degree $0$, find
another vertex $u \ne v$ that either has in-degree or out-degree $0$.

A natural algorithm for a PPAD problem simply follows the line that starts at $v$ until it
finds a vertex with out-degree $0$. The \emph{other end of the line} problem
asks us to find this vertex, and it is known to be PSPACE-complete~\cite{P94}.
This fact has been used to show that is PSPACE-complete to find any of the
equilibria of a bimatrix game that can be computed by the Lemke-Howson
algorithm~\cite{GPS13}, 
even though the problem of finding a Nash equilibrium
of a bimatrix game is only PPAD-complete~\cite{CDT09}.
In Theorem~\ref{thm:dls}, we show an even bigger gap between the complexity of
finding any solution and a specific one, namely, we show that it is
PSPACE-complete to compute the solution found by Dantzig's pivot rule, even
though the problem of finding a solution of a linear program is in P.

A potential criticism of Theorem~\ref{thm:dls} is that it requires that the
linear program has more than one optimal solution.
It is possible that one could turn a linear program with multiple solutions
into one with a unique solution using perturbations.
If $\mathcal{L}$ is a linear
program with a unique solution, then problem $\dls(\mathcal{L}, b, v)$ is
trivially in P.
However, the fact that Dantzig's pivot rule can solve PSPACE-complete problems
does not depend on a linear program with multiple solutions, as we show in our
second main theorem. Given a linear program $\mathcal{L}$, a variable $v$, and a
initial basic feasible solution $b$ in which $v$ is not basic, the problem
$\basisentry(\mathcal{L}, b, v)$ asks the following question: if Dantzig's pivot
rule is started at $b$, will it ever choose variable $v$ to enter the basis? 
The following theorem states our second
main result, which holds regardless of the degeneracy resolution rule used by
Dantzig's pivot rule.
\begin{theorem}
\label{thm:basisentry}
\basisentry is PSPACE-complete.
\end{theorem}

Theorem~\ref{thm:basisentry} continues a line of work that was recently
initiated by Disser and Skutella~\cite{DS13}. They defined an algorithm to be
\emph{NP-mighty} if it implicitly solves every problem in NP, and they showed
that the (network) simplex algorithm with Dantzig's pivot rule is NP-mighty.
Using this terminology, Theorem~\ref{thm:basisentry} shows that Dantzig's pivot
rule is in fact PSPACE-mighty.
This is also related to the interesting recent work of Adler, Papadimitriou, and
Rubinstein~\cite{APR14}, which explicitly conjectures that
Theorem~\ref{thm:basisentry} may be true, but states that it is a challenging
open problem.

Given its good performance in practice, we find it interesting that the simplex
method can actually solve PSPACE-hard problems. As alluded to by Disser and
Skutella, it seems somewhat counter-intuitive that, while an exponential-time
worst-case example for an algorithm is often considered to show that the
algorithm is ``stupid'', a PSPACE-hardness results suggests it is in some sense
``clever''.

\paragraph{\bf Techniques.} 

In order to prove Theorems~\ref{thm:dls} and~\ref{thm:basisentry}, we make use
of a known connection between the simplex method for linear programming and
\emph{policy iteration} algorithms for Markov decision processes (MDPs),
which are discrete-time stochastic control processes~\cite{Put94}.
The problem of finding an optimal policy in an MDP can be solved in
polynomial time by a reduction to linear programming. However, policy iteration is
a local search technique that is often used as an alternative. Policy iteration
starts at an arbitrary policy. In each policy it assigns each action an
\emph{appeal}, and if an action has positive appeal, then switching this action
creates a strictly better policy. Thus, policy iteration proceeds by repeatedly
switching a subset of switchable actions, until it finds a policy with no
switchable actions. The resulting policy is guaranteed to be optimal.

We use the following connection: If a policy iteration algorithm for an MDP
makes only a single switch in each iteration, then it corresponds directly to
the simplex method for the corresponding linear program. In particular,
Dantzig's pivot rule corresponds to the natural switching rule for
policy iteration that always switches the action with highest appeal. We 
call this \emph{Dantzig's rule}.
This connection is well known, and has been applied in other contexts.
Friedmann, Hansen, and Zwick used this connection in the expected total-reward
setting, to show sub-exponential lower bounds for some randomized pivot
rules~\cite{FHZ11}.
Post and Ye have shown that Dantzig's pivot rule is strongly
polynomial for \emph{deterministic discounted} MDPs~\cite{PY13}, while
Hansen, Kaplan, and Zwick went on to prove various further bounds
for this setting~\cite{HKZ14}.

We define two problems for Dantzig's rule. Let $\mathcal{M}$ be an MDP, $\sigma$
be a starting policy, and~$a$ be an action. The problem $\dms(\mathcal{M},
\sigma, a)$ asks: if $\sigma^*$ is the optimal policy that is found when
Dantzig's rule is started at $\sigma$, does $\sigma^*$ use action $a$? The
problem $\actionswitch(\mathcal{M}, \sigma, a)$ asks: if Dantzig's rule is
started at some policy $\sigma$ that does not use $a$, will it ever switch
action $a$? We prove the following two theorems:

\begin{theorem}
\label{thm:mdpend}
$\dms$ is PSPACE-complete.
\end{theorem}

\begin{theorem}
\label{thm:actionswitch}
$\actionswitch$ is PSPACE-complete.
\end{theorem}

Policy iteration is a well-studied and frequently-used method. Thus, these
theorems are of interest in their own right. Additionally, given the
connections we show between policy iteration for our construction and Dantzig's
pivot rule on a certain linear program, these two theorems immediately imply
that $\dls$ and $\basisentry$ are PSPACE-complete (Theorems~\ref{thm:dls}
and~\ref{thm:basisentry}, respectively). The majority of the paper is dedicated
to proving Theorem~\ref{thm:actionswitch}; we then add one extra gadget to
prove Theorem~\ref{thm:mdpend}.

Our PSPACE-completeness results are shown by reductions from two slightly
different \emph{circuit iteration} problems. For example, the
PSPACE-completeness result for $\dls$ is a reduction from $\circuitvalue$, which
asks: given a function $F : \{0, 1\}^n \rightarrow \{0, 1\}^n$ implemented by a
boolean circuit $C$, an input bit-string $B$, and an integer $z$,  is the $z$-th
bit of $F^{2^n}(B)$ a $0$? 

We build an MDP that forces Dantzig's rule to compute $F^{i}(B)$ for all $i \le
2^n$.
Melekopoglou and Condon have shown an exponential-time lower bound for Dantzig's
rule~\cite{MC94}.
As a base for our construction, we use a \emph{clock}, which is a modification
of Condon and Melekopoglou's construction.
We show how a
boolean circuit can be implemented by a system of gadgets. Our construction
contains two full copies of a circuit for $F$. Both circuits have input bits
and output bits. For both circuits, the output bits of that circuit are connected 
to the input bits of the other circuit.
The circuits take turns in
computing $F$: the first circuit computes $F$ using the values stored in its input
bits; once that computation is complete, the second circuit copies the
resulting output into its input bits.
Each time the clock ticks, the two circuits swap roles, so the second circuit 
computes and the first circuit copies. The clock ticks $2^n-1$ times, so in the
end we will have computed $F^{2^n}(B)$, and this 
can be used to show PSPACE-completeness of $\actionswitch$ and $\dls$.

\paragraph{\bf Related work.} There has been a recent explosion of interest in
the complexity of pivot rules for the simplex method, and of switching rules for
policy iteration. The original spark for this line of work was a result of
Friedmann, which showed an exponential lower bound for the \emph{all-switches}
variant of strategy improvement for two-player \emph{parity
games}~\cite{F09,F11}. Fearnley then showed that the second player in
Friedmann's construction can be simulated by a probabilistic action, and used
this to show an exponential lower bound for the all-switches variant of policy
iteration of average-reward MDPs~\cite{F10}. Friedmann, Hansen, and Zwick then
showed a sub-exponential lower bound for the \emph{random facet} strategy
improvement algorithm for parity games~\cite{FHZ11b}, and then utilised
Fearnley's construction to extend the bound to the random facet pivot rule for
the simplex method~\cite{FHZ11}. Friedmann also gave a sub-exponential lower bound for
Zadeh's pivot rule for the simplex method~\cite{F11b}.

It is generally accepted that the simplex algorithm performs well in practice.
Our strong worst-case negative result should be understood in the context of a
long line of work that has attempted to explain the good behaviour of the
simplex algorithm. This started with probabilistic analyses of the expected
running time of variants of the simplex method by Adler and Megiddo~\cite{AM85},
Borgwardt~\cite{Bor86}, and Smale~\cite{Sma83}. Later, in seminal work,
Spielman and Teng~\cite{ST04} defined the concept of smoothed analysis and
showed that the simplex algorithm has smoothed polynomial complexity.

\bigskip

\paragraph{\bf Roadmap.} 

In Section~\ref{sec:definitions} we formalize policy iteration for MDPs and
explain the connection with the simplex method for linear programming. In
Section~\ref{sec:circuit_iter} we define the PSPACE-complete circuit iteration
problems, which are the starting point for our reductions. In
Section~\ref{sec:overview} we define our construction and give a high-level
overview of how it works. Finally, in Section~\ref{sec:proof} we give an
overview of the proof of correctness of our construction for our main two
theorems about MDPs, which directly imply our two main theorems for the simplex
method. The full proof is unfortunately long and technical and is contained in
its entirety in Appendices~\ref{app:construction} to~\ref{app:mdpend}.

\section{Preliminaries}
\label{sec:definitions}

\paragraph{\bf Markov decision processes.}
A Markov decision process (MDP) is defined by a tuple $\mathcal{M} = (S,
(A_s)_{s \in S}, p, r)$, where $S$ gives the set of states in the MDP. For each
state $s \in S$, the set $A_s$ gives the actions available at $s$. We also
define $A = \bigcup_s A_s$ to be the set of all actions in $\mathcal{M}$. For
each action $a \in A_s$, the function $p(s', a)$ gives the probability of moving
from $s$ to $s'$ when using action $a$. Obviously, we must have $\sum_{s' \in S}
p(s', a) = 1$, for every action $a \in A$. Finally, for each action $a \in A$,
the function $r(a)$ gives a rational \emph{reward} that is obtained when using
action $a$.

A \emph{deterministic memoryless policy} is a function $\sigma : S \rightarrow
A$, which for each state $s$ selects some action from $A_s$. All of the policies
that we consider in this paper will be deterministic and memoryless, so we will
henceforth refer to deterministic memoryless policies as \emph{policies}. We
define $\Sigma$ to be the set of all deterministic memoryless policies. We say
that an action $a$ is \emph{deterministic} if there exists an $s'$ such that
$p(s', a) = 1$. If $a$ is a deterministic action at some state $s$, and $s'$ is
the state such that $p(s', a) = 1$, then we use the shorthand $\sigma(s) = s'$
to denote that policy $\sigma$ selects action $a$ at state $s$.

In this paper, we use the \emph{expected average reward}
optimality criterion. It has been shown that maximizing expected
average reward is equivalent to solving the following system of \emph{optimality
equations}~\cite{Put94}. For each state $s \in S$ we have a \emph{gain equation}:
\begin{equation}
\label{eqn:gain}
G(s) = \max_{a \in A_s} \left(\sum_{s' \in S} p(s, a)
\cdot G(s') \right).\\
\end{equation}
Secondly, for each state $s$ we have a \emph{bias equation}. If $M_s = \{a \in
A_s \; : \; G(s) = \sum_{s' \in S} p(s'|s, a) \cdot G(s') \}$ is the set of
actions that satisfy the gain equation at the state $s$, then the bias equation
for $s$ is:
\begin{equation}
\label{eqn:bias}
B(s) = \max_{a \in M_s} \left( r(a) - G(s) + \sum_{s' \in S} p(s', a)
\cdot B(s') \right)
\end{equation}
It has been shown that these equations have a unique solution, and that for each
state $s$, the value of $G(s)$ is the largest expected average reward that can
be obtained from $s$~\cite{Put94}.

\paragraph{\bf Policy iteration.}
Policy iteration is an algorithm for finding solutions to the optimality
equations. For each policy $\sigma \in \Sigma$, we define the following system
of linear equations:
\begin{align}
\nonumber
G^{\sigma}(s) &= \sum_{s' \in S} p(s', \sigma(s)) \cdot G^{\sigma}(s) \\
\label{eqn:biassigma}
B^{\sigma}(s) &= r(\sigma(s)) - G^{\sigma}(s) + \sum_{s' \in S} p(s,
\sigma(s)) \cdot B^{\sigma}(s') 
\end{align}
In a solution of this system, $G^{\sigma}(s)$ gives the expected average reward
for obtained by following $\sigma$. We say that an action $a \in A_s$ is
\emph{switchable} in $\sigma$ if either $\sum_{s' \in S} p(s', a) \cdot
G^\sigma(s') > G^{\sigma}(s)$ or if $\sum_{s' \in S} p(s', a) \cdot G^\sigma(s')
= G^{\sigma}(s)$ and $r(a) - G^{\sigma}(s) + \sum_{s' \in S} p(s', a) >
B^{\sigma}(s)$. Switching an action $a$ at a state $s$ in a policy $\sigma$
creates a new policy $\sigma'$ such that $\sigma'(s) = a$, and $\sigma'(s') =
\sigma(s')$ for all states $s' \ne s$. 

We define an ordering over policies using gain and bias. If $\sigma, \sigma' \in
\Sigma$, then we say that $\sigma \prec \sigma'$ if and only if one of the two
following conditions hold:
\begin{itemize}
\item $G^{\sigma'}(s) \ge G^{\sigma}(s)$ for every state $s$, and there exists a
state $s'$ for which $G^{\sigma'}(s') > G^{\sigma}(s')$.
\item $G^{\sigma'}(s) = G^{\sigma}(s)$ and $B^{\sigma'}(s) \ge B^{\sigma}(s)$
for every state $s$, and there exists a state $s'$ for which $B^{\sigma'}(s') >
B^{\sigma}(s')$.
\end{itemize}
The following theorem states that when we switch a switchable action, then we
obtain a better policy in this ordering.

\begin{theorem}[\cite{Put94}]
\label{thm:profincrease}
If~$\sigma$ is a policy and~$\sigma'$ is a policy that is obtained by switching
a switchable action in~$\sigma$ then we have $\sigma \prec \sigma'$.
\end{theorem}
Policy iteration starts at an arbitrary policy $\sigma$. In each iteration, it
switches a switchable action in~$\sigma$ to create $\sigma'$, which is then
considered in the next iteration. Since there are finitely many policies in
$\Sigma$, Theorem~\ref{thm:profincrease}, implies that we must eventually arrive
at a policy $\sigma^*$ with no switchable actions. By definition, a policy with
no switchable actions is a solution to Equations~\eqref{eqn:gain}
and~\eqref{eqn:bias}, so $\sigma^*$ is an optimal policy, and the algorithm
terminates.

\paragraph{\bf Simplification.} The construction that we give in this paper has
a special structure, which will allow us to simplify policy iteration.
Specifically, our construction ensures that under an optimal policy $\sigma^*$,
we have $G^{\sigma^*} = 0$ for every state $s$. 
Moreover, we will start policy iteration from a
policy~$\sigma$ with $G^{\sigma}(s) = 0$ for every state $s$. So, by
Theorem~\ref{thm:profincrease}, we have that $G^{\sigma'}(s) = 0$ for every
policy $\sigma'$ considered during policy iteration. If we substitute $0$ into
the gain equation, then we obtain the following simplification of
Equation~\eqref{eqn:bias}, which we will refer to as the \emph{value} equation:
\todo[inline]{introduce val first}
\begin{equation}
\label{eqn:optval}
\val(s) = \max_{a \in A_s} \left( r(a) + \sum_{s' \in S} p(s', a)
\cdot \val(s') \right)
\end{equation}
Additionally, Equation~\eqref{eqn:biassigma} simplifies to:
\begin{equation}
\label{eq:val}
\val^{\sigma}(s) = r(\sigma(s)) + \sum_{s' \in S} p(s', \sigma(s)) \cdot \val^{\sigma}(s') 
\end{equation}
The definition of a switchable action is also simplified.
For each policy $\sigma \in \Sigma$, each state $s$, and each action $a \in A_s$
we define:
\begin{equation}
\label{eq:appeal}
\appeal^{\sigma}(a) = 
\left(r(a) + \sum_{s' \in S} p(s', a) \cdot \val^{\sigma}(s') \right) - 
\val^{\sigma}(s)
\end{equation}
Thus, $a$ is switchable in $\sigma$ if and only if $\appeal^{\sigma}(a) > 0$.
This is the formulation that we will use during our proofs. 

We remark that these equations essentially define the \emph{expected total
reward} optimality criterion~\cite{Put94}. However, the detour through expected
average reward was necessary, because our construction does not fall into one of
the classes of MDPs for which total-reward policy iteration is known to work. In
particular, in our construction, there exist policies that obtain negative
infinite expected total reward, which precludes direct application of
total-reward policy iteration. As we have seen, we can get around this
restriction by forcing average-reward policy iteration to start at a policy with
average reward $0$.

\paragraph{\bf Dantzig's rule.}

Recall that policy iteration specifies that some switchable action
should be switched in each iteration. However, if there is more than one
switchable action, it does not specify \emph{which} switchable action should be
chosen. This decision is delegated to a \emph{switching rule}. In this paper, we
concentrate on one particular switching rule, which we call \emph{Dantzig's
switching rule}.

Dantzig's switching rule always selects an action with maximal appeal. More
formally, if $\sigma$ is a policy that has at least one switchable action, then
Dantzig's rule selects some action $a$ that satisfies:
\begin{equation}
\label{eqn:dantzig}
\appeal^{\sigma}(a) = \max\{ \appeal^\sigma(a') \; : \; a' \in A \text{ and }
\appeal^{\sigma}(a') > 0\},
\end{equation}
If more than one action $a$ satisfies this equation, then Dantzig's switching
rule selects one arbitrarily, and our PSPACE-completeness results will hold no
matter how ties are broken. We will refer to the policy iteration algorithm that
always follows Dantzig's switching rule as \emph{Dantzig's rule}.

We are interested in two slightly different problems regarding Dantzig's rule.
Let $\mathcal{M}$ be an MDP, let $\sigma$ be an initial policy, and let $a$ be
an action. The first problem is $\actionswitch(\mathcal{M}, \sigma, a)$, which
requires us to answer the following question: When Dantzig's rule is started at
policy $\sigma$, will it ever switch to a policy $\sigma'$ with $\sigma'(s) =
a$, for some state $s$? The second problem is $\dms(\mathcal{M}, \sigma, a)$,
which requires us to answer the following question: Suppose that Dantzig's rule
is started at $\sigma$, and that it finds an optimal policy $\sigma^*$. Does
there exist a state $s$ such that $\sigma^*(s) = a$? Note that this second
problem is non-trivial, because although Equations~\eqref{eqn:gain}
and~\eqref{eqn:bias} have a unique solution, there can be multiple optimal
policies that satisfy these equations.

\paragraph{\bf The connection with linear programming.}

There is a strong connection between policy
iteration for Markov decision processes, and the simplex method for linear
programming. In particular, for a number of classes of MDPs there is a well-known reduction 
to linear programming, which essentially encodes the optimality equations
in Equations~\eqref{eqn:gain} and~\eqref{eqn:bias} as a linear
program~\cite{Put94} and implies a correspondence between single-switch policy 
iteration and the simplex method applied to the dual linear program.
Technically, our construction is a multi-chain average-reward MDP. This 
class does have a linear programming formulation, but it is the most complex
case, and that formulation is more complex than we need.
In particular, the correspondence between the simplex method and policy
iteration is less clear.
We use special properties of our construction to define simpler primal-dual
pair of linear programs using~\eqref{eqn:optval}.
We provide a full exposition of this correspondence and its correctness in
Appendix~\ref{app:dantzig}.
In particular, we show that for our construction, Dantzig's pivot rule
corresponds to Dantzig's switching rule.

Note that, when Dantzig's pivot rule is applied in linear programming, a
\emph{degeneracy resolution} rule is required to prevent the algorithm from
cycling. This rule picks the entering variable and leaving variable in the case
of ties.
In our formulation, the leaving variable is always unique.
The entering variable is determined according to Equation~\eqref{eqn:dantzig}. Since our
PSPACE-completeness results for MDPs hold no matter how ties are broken in
Equation~\eqref{eqn:dantzig}, our PSPACE-completeness results for Dantzig's
pivot rule will also hold, no matter which degeneracy resolution rule is used.

As we mentioned in the introduction, the connection between policy iteration and
the simplex method has been exploited in previous related work. However, we have
been unable to find an explicit formalisation of this connection for the case of
expected average reward. So in Appendix~\ref{app:dantzig} we provide our own
formalisation, and we show that Dantzig's switching rule for an average-reward
MDP corresponds to applying Dantzig's pivot rule to the standard resulting
linear program. Consequently, if we can show that $\actionswitch$ and $\dms$ are
PSPACE-complete problems, then we will have proved Theorems~\ref{thm:dls}
and~\ref{thm:basisentry}.

\paragraph{\bf Appeal reduction gadget}

We now describe a gadget that will be used frequently in our construction, which
we call the \emph{appeal reduction gadget}. Similar gadgets were used by
Melekopoglou and Condon to show an exponential-time lower bound for Dantzig's
rule~\cite{MC94}, and by Fearnley to show an exponential-time lower bound
against the all-switches rule~\cite{F10}. 

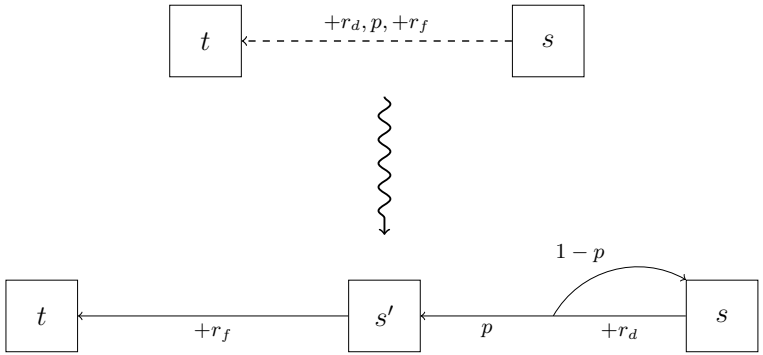
\begin{figure}
\begin{center}
\resizebox{!}{0.28\textwidth}{
\begin{tikzpicture}[node distance=2.5cm,auto]
\node[state,font=\large] (s) [] {$s$};
\node[left of=s] (half) {};
\node[state,font=\large,left of=half] (s2) [] {$s'$};
\node[state,font=\large,left=4cm of s2] (s3) [] {$t$};
\path
	(s) edge node [] {$+r_d$} (half.center)
	(half.center) edge [->] node [] {$p$} (s2)
	(half.center) edge [->,bend left=45] node [] {$1-p$} (s)
	(s2) edge [->] node [] {$+r_f$} (s3)
	;
\node[state,font=\large,above left=3cm and 1.5cm of s] (ss) [] {$s$};
\node[state,font=\large,left=4cm of ss] (ss2) [] {$t$};
\path[arg]
	(ss) edge node [swap] {$+r_d,p,+r_f$} (ss2)
	;
\path[snake=snake,draw,thick,line after snake=0.2cm,->]
	($(s2) + (0, 3.25)$) -- ($(s2) + (0, 1.2)$) 
	;
\end{tikzpicture} 
}
\end{center}
\caption{The appeal reduction gadget with reward $r$ and probability $p$. The
top diagram shows our shorthand, while the bottom diagram shows the gadget.
}
\label{fig:arg}
\end{figure}

The gadget is shown in Figure \ref{fig:arg}. Throughout the paper, we will use
the following diagramming notation for MDPs. States are represented as boxes,
and the name of the state is displayed in the center of the box. Deterministic
actions are represented as arrows that are annotated by rewards. Probabilistic
actions are represented as arrows that split. For these actions, the reward is
displayed before the split, while the transition probabilities are displayed
after the split.

The lower half of Figure~\ref{fig:arg} diagram shows the gadget itself, and the
upper half shows our diagramming notation for the gadget: whenever we use this
shorthand in our of our diagrams, we intend it to be replaced with the gadget in
the bottom half of Figure~\ref{fig:arg}. The parameters for the gadget are two
states $s$ and $t$, two rewards $r_d$ and $r_f$, and a probability $p$. In order
to simplify notation, if $\sigma$ is a policy such that $\sigma(s)$ chooses the
action towards $s'$, then in future we will simply  $\sigma(s) = t$. This is
because, when we use the shorthand notation, the state $s'$ does not appear in
our diagrams. 
The following lemma, which is
proved in Appendix~\ref{app:argappeal}, shows the two key properties of the
gadget.

\begin{lemma}
\label{lem:argappeal}
Let $\sigma$ be a policy, and let $a$ be the action between $s$ and $s'$.
\begin{itemize}
\item 
If $\sigma(s) = a$, then we have $\val^{\sigma}(s) = \val^{\sigma}(t) + r_f + \frac{r_d}{p}$.
\item If $\sigma(s) \ne a$, and if $\val^{\sigma}(t) = \val^{\sigma}(s) + b$,
for some constant $b$, then we have $\appeal^{\sigma}(a) = p \cdot (b + r_f) +
r_d$.
\end{itemize}
\end{lemma}

The first claim of Lemma~\ref{lem:argappeal} describes the outcome when a policy
uses the appeal reduction gadget. In particular if $r_d = 0$, which will
frequently be the case in our construction, then the appeal reduction gadget
acts like an action from $s$ to $t$ with reward $r_f$. The second claim
describes what happens when a policy does not use the appeal reduction gadget.
In this case, the appeal of moving to $t$ is scaled down by the probability $p$.
This property allows us to control when an action is switched by Dantzig's rule,
which will be crucial for our construction.

\section{Circuit Iteration Problems}
\label{sec:circuit_iter}

In order to prove PSPACE-completeness of $\actionswitch$ and $\dms$, we will
provide a reduction from two \emph{circuit iteration} problems, which we define
in this section.

\subsection{Circuits}

Let $C$ be a boolean circuit with $n$ input bits and $n$ output bits.
We represent $C$ as an list of gates indexed $1$ through $n+k$. 
The indices $1$ through $n$ represent the $n$ \emph{input bits}.
Then, for each $i > n$, we have:
\begin{itemize}
\item If gate $i$ is an \org gate, then we define $\inp_1(i)$ and $\inp_2(i)$ to
give the indices of its two inputs.
\item If gate $i$ is a \notg gate, then we define $\inp(i)$ to give the index of
its input.
\end{itemize}
The gates $(n + k) - n + 1 = k+1$ through $k + n$ correspond to the $n$
\emph{output bits} of the circuit, respectively. For the sake of convenience, we
also define, for each input bit $i$, we define $\inp(i) = k+i$, which indicates
that input bit $i$ should copy from output bit $\inp(i)$. Moreover, we assume
that the gate ordering is topological. That is, for each \org gate $i$ we assume
that $i > \inp_1(i)$ and $i > \inp_2(i)$, and we assume that for each \notg gate
$i$ we have $i > \inp(i)$.

For each gate $i$, let $d(i)$ denote the \emph{depth} of gate $i$, which is the
length of the longest path from $i$ to an input bit. 
Observe that we can increase the depth of a gate by inserting dummy \org gates:
given a gate~$i$, we can add an \org gate $j$ with $\inp_1(j) = i$ and $\inp_2(j)
= i$, so that $d(j) = i+1$. We use this fact in order to make the following
assumptions about our circuits:
\begin{itemize}
\item For each \org gate $i$, we have $d(\inp_1(i)) = d(\inp_2(i))$.
\item For each \notg gate $i$, we have $d(i) \ge 2$.
\item There is a constant $j$ such that, for every output bit $i \in \{k+1,
k+n\}$, we have $d(i) = j$.
\end{itemize}
From now on, we assume that all circuits that we consider satisfy these
properties. Note that, since all outputs gates have the same depth, we can
define $d(C) = d(k+1)$, which is the depth of all the output bits of the
circuit.

Given an input $B \in \{0, 1\}^n$, the truth values of each of the gates in $C$
are fixed. We define $C(B, i) = 1$ if gate $i$ is true for input $B$, and $C(B,
i) = 0$ if gate $i$ is false for input $B$.

Given a circuit $C'$, we define the \emph{negated form} of $C'$ to be a
transformation of $C'$ in which each output bit is negated. More formally, we
transform $C'$ into a circuit $C$
using the following operation: for each output bit $n+i$ in $C'$, we add a
\notg gate $n+k+i$ with $\inp(n+k+i) = n+i$.  In other words, we have have that
the $i$-th bit of $F(B)$ is $1$ if and only if the $C(B, i) = 0$.

\subsection{Circuit iteration problems}

A \emph{circuit iteration} instance is a triple $(F, B, z)$, where:
\begin{itemize}
\item $F : \{0, 1\}^n \rightarrow \{0, 1\}^n$ is a function represented as a
boolean circuit $C$,
\item $B \in \{0, 1\}^n$ is an initial bit-string, and
\item $z$ is an integer such that $1 \le z \le n$. 
\end{itemize}
We use standard notation for function iteration: given a bit-string $B \in
\{0,1\}^n$, we recursively define $F^{1}(B) = F(B)$, and $F^{i}(B) =
F(F^{i-1}(B))$ for all $i > 1$. We define two different circuit iteration problems,
which correspond to the two different theorems that we prove for Dantzig's
rule. Both are decision problems that take as input a circuit iteration instance $(F, B, z)$.
\begin{itemize}
\itemsep2mm
\item $\bitswitch(F, B, z)$: if the $z$-th bit of $B$ is $1$, then decide
whether there exists an even $i \le 2^n$ such that the $z$-th bit of $F^{i}(B)$
is $0$.
\item $\circuitvalue(F, B, z)$: decide whether the $z$-th bit of $F^{2^n}(B)$ a $0$.
\end{itemize}
The requirement for $i$ to be even in $\bitswitch$ is a technical requirement
that is necessary in order to make our reduction work. The fact that both of
these problems are PSPACE-complete should not be too surprising, because we can
use the circuit~$F$ to simulate a single step of a space-bounded Turing machine,
so when $F$ is iterated, it simulates the space-bounded Turing machine. The
following lemma is shown in Appendix~\ref{app:pspace}. 

\begin{lemma}
\label{lem:pspace}
Both $\bitswitch$ and $\circuitvalue$ are \PSPACE-complete.
\end{lemma}

\section{Overview}
\label{sec:overview}

Let $(F, B, z)$ be a circuit iteration instance, where $F$ is an $n$-bit
function, and where $C'$ is the circuit that implements $F$. Let $C$ be the
negated form of $C'$. Our goal is to reduce the problem $\bitswitch$ to the
problem $\actionswitch$. To do so, we construct an MDP, which will be called
$\const(C)$.

The core part of the construction is the \emph{clock}. For this, we use a
modified version of the exponential-time examples of Melekopoglou and
Condon~\cite{MC94}. The clock has two output states $c_0$ and $c_1$, and the
difference in value between these two states is what drives our construction. In
particular, the clock alternates between two output \emph{phases}: if we fix
$T = 3^{d(C) + 6}$ then:
\begin{itemize}
\item In phase $0$ we have $\val^{\sigma}(c_1) = \val^{\sigma}(c_0) + T$.
\item In phase $1$ we have $\val^{\sigma}(c_0) = \val^{\sigma}(c_1) + T$.
\end{itemize}
The clock alternates between phase $0$ and phase $1$. 
The clock goes through exactly $2^n$ many phases, and therefore alternates
between phase $0$ and phase $1$ exactly $2^{n}-1$ many times.

Our goal is to compute one iteration of $F$ in each clock phase. We will
maintain two copies of the circuit $C$, which will be numbered $0$ and $1$. At
the start of the computation, circuit $0$ will hold the initial bit-string $B$
in its input bits. In phase $0$, circuit $0$ will compute $F(B)$,  and then
circuit~$1$ will copy $F(B)$ into its input bits. When we change to phase $1$,
circuit $1$ will compute $F(F(B))$, and then circuit $0$ will copy $F(F(B))$
into its input bits. This pattern is then repeated until $2^n$ phases have
been encountered, and therefore when the clock terminates we will have computed
$F^{2^n}$, as required.

Each copy of the circuit is built from gadgets. We will design gadgets to model
the input bits, \org gates, and \notg gates. In particular, for each gate
$i$ in the circuit, and for each $j \in \{0, 1\}$, we will have a state $o^j_i$,
which represents the output of the gate $i$ in copy $j$ of the circuit.
The value of this state will indicate whether the gate is true or
false, in a way that we now formally define.

Let $j \in \{0, 1\}$. Recall that in phase $j$, circuit $j$ will compute the
function $F$, and circuit $1-j$ will copy the output of circuit $j$. Firstly,
for each $k$ with $0 \le k \le d(C)$, we define the following constants:
\begin{align*}
b_k &= 3^{d(C) - j + 2},  &
L_k &= \sum_{m=0}^{k-1} b_m, &
H_k &= \sum_{m=0}^{k} b_m.
\end{align*}
The constant $H_k$ gives a \emph{high} value, and will be used by gates of depth
$k$ to indicate that they are true. The constant $L_k$ gives a \emph{low} value,
and will be used by gates of depth $k$ to indicate that they are false. Note
that $H_k = L_k + b_k$ and that $H_k = L_{k+1}$ for all $k$. Moreover, note that
$H_{d(C)} \le 2 \cdot 3^{d(C) + 2}$ and that $H_{d(C)} < \frac{T}{2}$.

We use these constants to define the truth value of our gates. Let $\sigma$ be a
policy in phase $j$. Recall that in phase $j$, we have $\val^{\sigma}(c_{1-j}) =
\val^{\sigma}(c_j) + T$. The truth values in circuit $j$ will be given relative
to the value of $c_j$. More precisely, we have:
\begin{itemize}
\item Gate $i$ is false in $\sigma$ in phase $j$ if 
$\val^{\sigma}(o^j_i) = \val^{\sigma}(c_j) + L_{d(i)}$.
\item Gate $i$ is true in $\sigma$ in phase $j$ if
$\val^{\sigma}(o^j_i) = \val^{\sigma}(c_j) + H_{d(i)}$.
\end{itemize}
Note, in particular, that these values depend on the depth of the gate. 

In the rest of this section, we define each component of the construction, and
we give high level descriptions of how each component operates. We will give a
full description of $\const(C)$ in terms of diagrams, and a full formal
description of the $\const(C)$ can be found in Appendix~\ref{app:construction}.
A formal proof that our construction works will be presented in
Section~\ref{sec:proof}.

\subsection{The clock}

\begin{figure}[h]
\begin{center}
\resizebox{\textwidth}{!}{
\begin{tikzpicture}[auto,scale=2.25]

\draw (0,0) node[font=\large,clock_state] (c0) []   {$c_0$};
\draw (1,0) node[clock_state] (snm1) [] {$n$};
\draw (2,0) node[clock_state] (snm2) [] {$n-1$};
\draw (3,0) node[clock_state] (snm3) [] {$n-2$};
\draw (4,0) node[dummy_state] (snm4) [] {};
\draw (5,0) node[clock_state] (s1) []   {$1$};
\draw (6,0) node[clock_state] (s0) []   {$0$};
\draw (7,0) node[clock_state] (s0s)[] {$\sink$};
\draw (7,-1)node[clock_state] (s1s) [] {$\sink'$};

\draw (0,-1) node[font=\large,avg_state] (snp)[] {$c_1$};
\draw (1,-1) node[avg_state] (snm1p)[] {$n'$};
\draw (2,-1) node[avg_state] (snm2p)[] {$(n-1)'$};
\draw (3,-1) node[avg_state] (snm3p) [] {$(n-2)'$};
\draw (4,-1) node[dummy_state] (snm4p) [] {};
\draw (5,-1) node[avg_state] (s1p)[] {$1'$};

\path[->]  (c0) edge node [] {$0$} (snm1) ;

\path[arg] (snm1) edge node [] {$\alpha_n$} (snm2) ;
\path[arg] (snm2) edge node [] {$\alpha_{n-1}$} (snm3) ;
\path[arg] (snm4) edge node [] {$\alpha_2$} (s1) ;
\path[arg] (s1) edge node [] {$\alpha_1$} (s0) ;
\path[->] (s0) edge node [] {$0$} (s0s) ;

\draw[ultra thick,dotted] (3.5,-0.5) -- (4.25,-0.5); 

\path[arg] (snm1) edge node [left,yshift=6] {$\alpha_n$} (snm1p) ;
\path[arg] (snm2) edge node [left,yshift=6] {$\alpha_{n-1}$} (snm2p) ;
\path[arg] (snm3) edge node [left,yshift=6] {$\alpha_{n-2}$} (snm3p) ;
\path[arg] (s1) edge node   [left,yshift=6] {$\alpha_1$} (s1p) ;

\path[->] (snp) edge node [] {} (snm1p) ;
\path[->] (snm1p) edge node [] {} (snm2p) ;
\path[->] (snm2p) edge node [] {} (snm3p) ;
\path[->] (snm4p) edge node [] {} (s1p) ;
\path[->] (s1p) edge node [] {} (s1s) ;

\path[->] (snp) edge node [] {} (snm2) ;
\path[->] (snm1p) edge node [] {} (snm3) ;
\path[->] (snm2p) edge node [] {} (snm4) ;
\path[->] (snm4p) edge node [] {} (s0) ;
\path[->] (s1p) edge node [] {} (s0s) ;

\path[->] (s0s) edge[loop] node [] {$0$} (s0s) ;
\path[->] (s1s) edge node [right] {$T \cdot 2^{n+1}$} (s0s) ;

\end{tikzpicture} 
}
\end{center}
\caption{The clock construction. 
}
\label{fig:clock}
\end{figure}

Figure~\ref{fig:clock} shows the clock. For each $i$ with $1 \le i \le n$, the
probability $\alpha_i$ is defined to be:
\begin{equation}
\label{eq:arg_probs}
\alpha_i = \clockappeal \cdot T^{-1} \cdot 2^{-(f(i)-1)},
\end{equation}
We have used special notation to simplify this diagram. The circle states have a
single outgoing action $a$. Each circle state has two outgoing arrows, which are
each both have a $0.5$ probability of being taken when  $a$ is used. That is, if
$s$ is a circle state, and $t$ and $u$ are the two states with incoming arrows
from $s$, then we have $p(t, a) = 0.5$ and $p(u, a) = 0.5$. Also, to save
space, we only display the probability parameter for the appeal reduction
gadgets. This is because the $r_d$ and $r_f$ parameters are both $0$ in all of
the appeal reduction gadgets in the clock.

The clock is an adaptation of the exponential-time lower bound of Melekopoglou
and Condon~\cite{MC94}, but with several modifications. Firstly, they consider
minimising the reachability probability of state $0$, whereas we consider
maximizing the expected total reward. In their original construction, the action
at state $0$ went to $n$ with probability $0.5$, and to $\sink$ with probability
$0.5$. We have replaced this with a deterministic action from $0$ to $\sink$. 

Note that the two states $c_0$ and $c_1$ are the two clock states that we
described earlier. In the initial policy for the clock, all states select the
action going right. In the optimal policy, all states select the action going
right, except state $1$, which selects its downward action. However, to
move from the initial to final policy, Dantzig's rule makes an
exponential number of switches. In each step, the probability of reaching
$\sink'$
increases, and this is what allows us to generate an exponential sequence of
clock phases.

One final thing to note about the clock is that the state $\sink$ has
a self-loop with reward $0$. We call this the \emph{sink}. In the
optimal policy for our construction, for every state $s$, the probability of
eventually moving from $s$ to $\sink$ is $1$. This property will also
hold for the initial policy. This is what allows us to guarantee that the
expected average-reward is always $0$.

\subsection{Input bits}

For every input bit $i$, and every $j \in \{0, 1\}$, our construction will
contain a copy of the gadget shown in Figure~\ref{fig:input}.
The probabilities shown in Figure~\ref{fig:input} are defined as follows.
\begin{align*}
p_{3} &= \frac{\bl}{\frac{3T}{2} + H_0}. &
p_{4} &= \frac{\lj}{\frac{3T}{2} + H_0 - \frac{H_{d(C)} + L_{d(C)}}{2}}. \\
p_{5} &= \frac{\lonej}{\frac{T}{2} + \frac{H_{d(C)} + L_{d(C)}}{2} - H_0} &
p_{6} &= \frac{\rj}{\frac{3T}{2} + L_0 - H_{d(C)}}. \\
p_{7} &= \frac{\ro}{\frac{T}{2} + H_{d(C)} - L_0}. 
\end{align*}
Note that since all of these probabilities contain at least $\frac{T}{2}$ in the
denominator, they must all be strictly less than $1$.

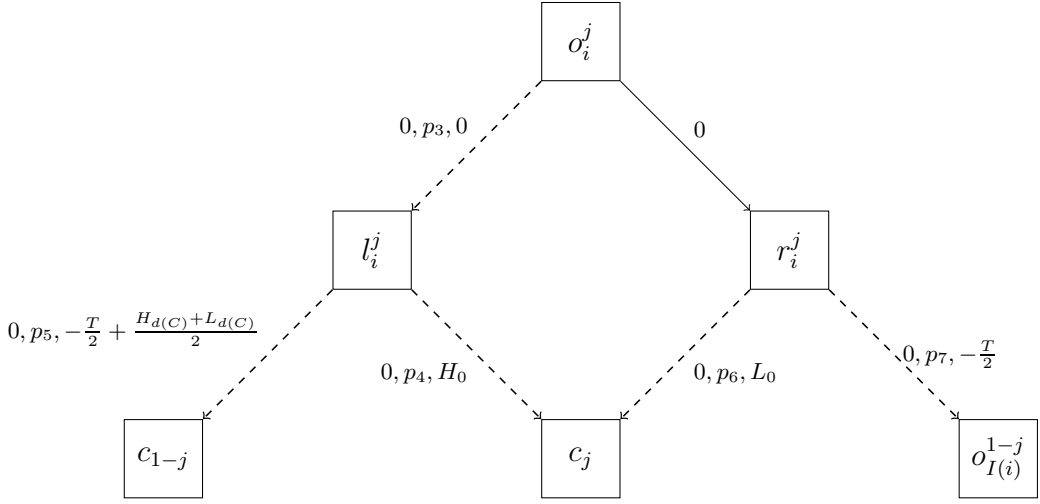
\begin{figure}[htb]
\begin{center}
\resizebox{!}{0.4\textwidth}{
\begin{tikzpicture}[node distance=4cm,auto]
\node[state,font=\large] (b) [] {$o^j_i$};
\node[state,font=\large] (l) [below left of=b] {$l^j_i$};
\node[state,font=\large] (r) [below right of=b] {$r^j_i$};
\node[state,font=\large] (o) [below right of=r] {$o^{1-j}_{I(i)}$};
\node[state,font=\large] (c_0) [below left of=l] {$c_{1-j}$};
\node[state,font=\large] (c_1) [below right of=l] {$c_j$};
\path[arg] 
	(b) edge node [swap] {$0,p_{3},0$} (l)
	(l) edge node [swap] {$0,p_{5},-\frac{T}{2} + \frac{H_{d(C)} + L_{d(C)}}{2}$} (c_0)
	(l) edge node [swap] {$0,p_{4},H_0$} (c_1)
	(r) edge node [] {$0,p_{6},L_0$} (c_1)
	(r) edge node [anchor=west] {$0, p_{7},-\frac{T}{2}$} (o)
	;
\path[->]
	(b) edge node [] {$0$} (r)
	;
\end{tikzpicture} 
}
\end{center}
\caption{The gadget for input bit $i$ in circuit $j$.}
\label{fig:input}
\end{figure}

The input bit gadget has two distinct modes: in phase $j$ the input bits of
circuit $j$ must output the values that they are holding, and once the output
bits of circuit $j$ have been computed, the input bits of circuit $1-j$ must
copy these outputs. Correspondingly, our input bits can either be in \emph{output
mode}, or in \emph{copy mode}.

In phase $j$, we say that input bit $i$ in circuit $j$ is in output mode in a
policy $\sigma$ if $\sigma(l^j_i) = c_j$ and $\sigma(r^j_i) = c_j$. In this case
we have:
\begin{itemize}
\item If $\sigma(o^j_i) = l^j_i$, then $\val^{\sigma}(o^j_i) =
\val^{\sigma}(c_j) + H_0$, so input bit $i$ is true in $\sigma$.
\item If $\sigma(o^j_i) = r^j_i$, then $\val^{\sigma}(o^j_i) =
\val^{\sigma}(c_j) + L_0$, so input bit $i$ is false in $\sigma$.
\end{itemize}
Thus, when gadget $i$ is in output mode, it always outputs either true or false,
and the choice made at $o^j_i$ determines which is the case.

In phase $j$, we say that input bit $i$ in circuit $1-j$ is in copy mode  in a
policy $\sigma$ if $\sigma(l^{1-j}_i) = c_{j}$, $\sigma(r^{1-j}_i) =
o^{1-j}_{\inp(i)}$, and $\sigma(o^{1-j}_i) = l^{1-j}_i$. Note that the two clock
states in Figure~\ref{fig:input} are indexed by $j$, so when we have
$\sigma(l^{1-j}_i) = c_{j}$, we have that $l^{1-j}_i$ is taking the \emph{left}
action shown in Figure~\ref{fig:input}.  If all gates in circuit $j$ have been
evaluated, then since we have assumed that all outputs gates of $C$ have the
same depth, there are two possible values that $o^{j}_{\inp(i)}$ can take:
\begin{itemize}
\item If $\val^{\sigma}(o^{j}_{\inp(i)}) = L_{d(C)}$, then the appeal of
switching $o^{1-j}_i$ to $r^{1-j}_i$ is 
$L_{d(C)} - \frac{H_{d(C)} + L_{d(C)}}{2} < 0$.
\item If $\val^{\sigma}(o^j_{\inp(i)}) = H_{d(C)}$, then the appeal of switching
$o^{1-j}_i$ to $r^{1-j}_i$ is 
$H_{d(C)} - \frac{H_{d(C)} + L_{d(C)}}{2} > 0$.
\end{itemize}
So, $o^{1-j}_i$ can switch to $r^{1-j}_i$ if and only if the~$i$-th output bit
of~$C$ is~$1$. Recall that when circuit $1-j$ is in output mode, we have that
$o^{1-j}_i$ outputs true if and only if $\sigma(o^{1-j}_i) = l^{1-j}_i$. This is
the reason why we needed $C$ to be in negated form, because this means that
$o^{1-j}_i$ will be switched to $r^{1-j}_i$ if and only if the~$i$-th output bit
of~$F$ is~$0$. Thus, the output
of circuit~$j$ will be correctly copied into the input bits of circuit $1-j$.

Finally, we describe how the gadget transitions between the phases. When we move
from phase $j$ to phase $1-j$, the roles of the two circuits are switched. Thus,
the input bits of circuit $j$ must be switched from output mode to copy mode,
and the input bits of circuit $1-j$ must be switched from copy mode to output
mode. To do this, we must ensure that, for every input bit $i$, the states
$l^j_i$, $r^j_i$, $l^{1-j}_i$, and $r^{1-j}_i$, are all switched before before
phase $1-j$ begins. Moreover, we must switch $o^{j}_i$ to $l^j_i$, so that it is
ready to copy in the next phase. 

The probabilities used in the appeal reduction gadgets are specifically chosen
for this task. As our proof will show, none of these states will be switched
before the output of circuit $j$ is copied into the input of circuit $1-j$, and
once the copy has taken place, the probabilities ensure that they switch in a
specific order. In particular, $o^{1-j}_i$ must \emph{not} be switched during
this process, because doing so would destroy the input value that we will use
during phase $1-j$.

\subsection{\org gates}

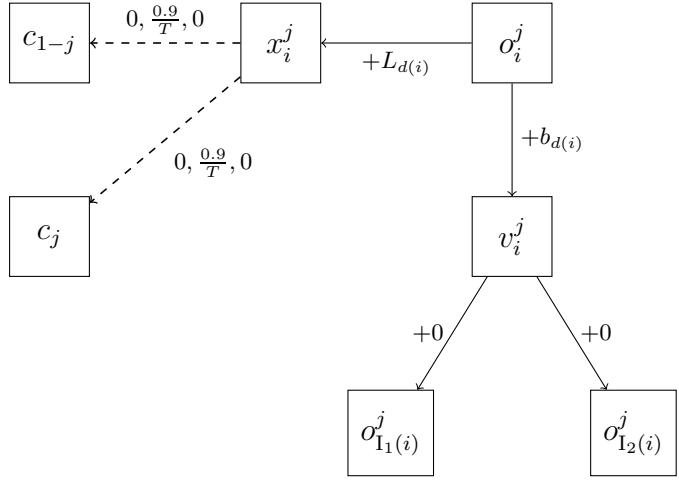
\begin{figure}
\begin{center}
\begin{tikzpicture}[node distance=3.5cm,auto]
\node[state,font=\large] (v) [] {$v^j_i$};
\node[state,font=\large] (i_1) [below left=1.5cm and 0.5cm of v] {$o^j_{\inp_1(i)}$};
\node[state,font=\large] (i_2) [below right=1.5cm and 0.5cm of v] {$o^j_{\inp_2(i)}$};
\node[state,font=\large] (o) [above=1.5cm of v] {$o^j_i$};
\node[state,font=\large] (c) [left=2cm of o] {$x^j_i$};
\node[state,font=\large] (cj) [below left=1.5cm and 2cm of c] {$c_j$};
\node[state,font=\large] (cp) [left=2cm of c] {$c_{1-j}$};
\path[->] 
	(v) edge node [anchor=east] {$+0$} (i_1)
	(v) edge node [anchor=west] {$+0$} (i_2)
	(o) edge node [] {$+b_{d(i)}$} (v)
	(o) edge node [] {$+L_{d(i)}$} (c)
	;
\path[arg]
	(c) edge node {$0,\frac{\xj}{T},0$} (cj)
	(c) edge node [swap] {$0,\frac{\xj}{T},0$} (cp)
	;
\end{tikzpicture} 
\end{center}
\caption{The gadget for an \org gate $i$ in circuit $j$.}
\label{fig:or}
\end{figure}

For every \org gate $i$ and every $j \in \{0, 1\}$, our construction will
include a copy of the gadget shown in Figure~\ref{fig:or}. As the circuit is
computing in phase $j$ we will have that $x^j_i$ takes the action towards $c_j$.
The probabilities on the appeal reduction gadgets at $x^j_i$ ensure that $x^j_i$
only switches to $c_{1-j}$ after the output of the circuit $j$ has been copied
to circuit $1-j$.

The purpose of the state $v^j_i$ is to select the maximum of the two inputs to
gate $i$. It should be fairly clear that if exactly one of the two input gates
is true, then $v^j_i$ will switch towards that input. If both input gates have
the same truth value, then it is irrelevant which action $v^j_i$ chooses.

The state $o^j_i$ ensures that the gate outputs the correct value. To see this,
suppose that $\sigma$ is a policy with $\sigma(o^j_i) = v^j_i$. 
\begin{itemize}
\item If at least one of the input gates is true, then we will have
$\val^{\sigma}(v^j_i) = \val^{\sigma}(c_j) + H_{d(i) - 1}$. In this case we
have:
\begin{equation*}
\val^{\sigma}(o^j_i) = \val^{\sigma}(c_j) + H_{d(i) - 1} + b_{d(i)} = \val^{\sigma}(c_j) + H_{d(i)}.
\end{equation*}
Thus, if one of the two input gates is true, the output of gate $i$ will be
true.
\item If both input gates are false, then we will have
$\val^{\sigma}(v^j_i) = \val^{\sigma}(c_j) + L_{d(i) - 1}$. In this case we
have:
\begin{align*}
\val^{\sigma}(o^j_i) &= \val^{\sigma}(c_j) + L_{d(i) - 1} + b_{d(i)} \\
&< \val^{\sigma}(c_j) + L_{d(i)}.
\end{align*}
So, since $x^j_i$ takes the action towards $c_j$, we have that switching $o^j_i$
to $x^j_i$ has appeal $L_{d(i)} - (L_{d(i) - 1} + b_{d(i)}) < 0$. Therefore
$o^j_i$ will be switched to $x^j_i$, and we will have $\val^{\sigma}(o^j_i) =
\val^{\sigma}(c_j) + L_{d(i)}$, as required.
\end{itemize}

\subsection{\notg gates}

\begin{figure}
\begin{center}
\begin{tikzpicture}[node distance=3.5cm,auto]
\node[state,font=\large] (v) [] {$o^j_i$};
\node[state,font=\large] (i) [below of=v] {$o^j_{\inp(i)}$};
\node[state,font=\large] (a) [left=3cm of v] {$a^j_{i}$};
\node[state,font=\large] (c_1) [left=3.5cm of a] {$c_{1-j}$};
\node[state,font=\large] (c_0) [below left=2cm and 3.5cm of a] {$c_j$};
\path[->] 
	(v) edge node [] {$0$} (i)
	;
\path[arg]
	(v) edge node [] {$1, \frac{1}{b_{d(i)}}, 0$} (a)
	(a) edge node [swap] {$0, p_{1}, -T + H_{d(i) - 1}$} (c_1)
	(a) edge node [] {$0, p_{2}, 0$} (c_0)
	;
\end{tikzpicture} 
\end{center}
\caption{A gadget for a \notg gate $i$ in circuit $j$.}
\label{fig:not}
\end{figure}
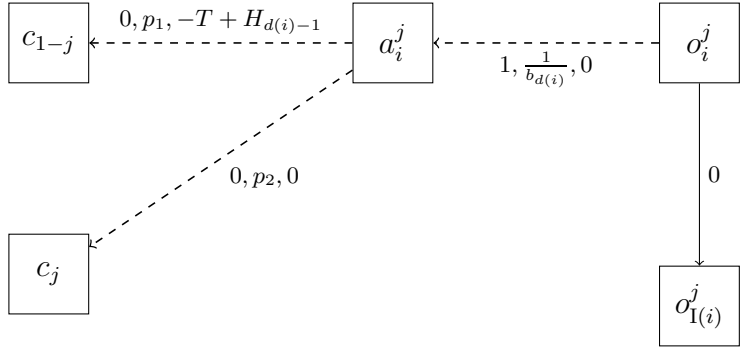

For each \notg gate $i$, and each $j \in \{0, 1\}$, our construction will include
one copy of the gadget shown in Figure~\ref{fig:not}. The probabilities used in
Figure~\ref{fig:not} are defined as follows.
\begin{align*}
p_{1} &= \frac{\ajprime}{H_{d(i)}} &
p_{2} &= \frac{\aj}{2T - H_{d(i)-1}} 
\end{align*}
Since $H_{d(i)} > 4$ we must have that $p_1 < 1$, and we must have $p_2 < 1$
because $2T - H_{d(i)-1} > \aj$.

The state $a^j_i$ has a very important role in this gadget. At the start of
phase $j$ we have that $a^j_i$ takes the action to $c_j$. In this configuration,
since by assumption all \notg gates have depth greater than $2$, we have that if
$\sigma$ is a policy with $\sigma(o^j_i) = o^j_{\inp(i)}$, then the appeal of
switching $o^j_i$ to $o^j_{\inp(i)}$ is at least: 
\begin{equation*}
(\val^{\sigma}(c_j) + L_{d(i)-1}) - (\val^{\sigma}(c_j) + b_{d(i)}) > 0,
\end{equation*}
whenever $d(i) \ge 2$. This has the effect of ensuring that $o^j_i$ is switched
to $o^{j}_{\inp(i)}$ before $a^j_i$ is switched.

We now describe what happens after $a^j_i$ is switched to $c_{1-j}$. Let
$\sigma$ be a policy in which $\sigma(a^j_i) = c_{1-j}$ and $\sigma(o^j_i) =
o^j_{\inp(i)}$. Let us consider the two possible outputs for gate $\inp(i)$.
\begin{itemize}
\item If gate $\inp(i)$ is false, then we have $\val^{\sigma}(o^j_i) =
\val^{\sigma}(c_j) + L_{d(i)-1}$. Therefore, we can apply
Lemma~\ref{lem:argappeal} to argue that the appeal of switching $o^j_i$ to
$a^j_i$ is:
\begin{align*}
\frac{1}{b_{d(i)}} \cdot (H_{d(i) - 1} - L_{d(i) - 1}) + 1 
&= \frac{1}{b_{d(i)}} \cdot b_{d(i) - 1} + 1 \\
&= \frac{1}{3^{d(C) - d(i) + 2}} \cdot 3^{d(C) - d(i) + 3} + 1 \\
&= 4.
\end{align*}
So, $o^j_i$ will be switched to $a^j_i$. If $\sigma'$ is this new policy, then
we will have:
\begin{equation*}
\val^{\sigma'}(o^j_i) = \val^{\sigma'}(c_j) + H_{d(i) - 1} + b_{d(i)} =
\val^{\sigma'}(c_j) + H_{d(i)}. 
\end{equation*}
This is the correct output value for gate $i$ in the case where gate $\inp(i)$
is false.
\item If gate $\inp(i)$ is true, then we have
$\val^{\sigma}(o^j_i) = \val^{\sigma}(c_j) + H_{d(i)-1}$. Again we can apply
Lemma~\ref{lem:argappeal} to argue that the appeal of switching $o^j_i$ to
$a^j_i$ is:
\begin{equation*}
\frac{1}{b_{d(i)}} \cdot (H_{d(i) - 1} - H_{d(i) -1}) + 1 = 1.
\end{equation*}
Note that the action is now less appealing than in the previous case. This is
the fundamental property that our construction exploits. We will ensure that
when Dantzig's rule computes the output of the circuit $j$, and copies it into
circuit $1-j$, it always switches an action with appeal strictly greater than
$1$. Thus, while the circuit is computing, an action with appeal $1$ will
not be switched. So $o^j_i$ will not be switched away from $o^j_{\inp(i)}$, and
we have:
\begin{equation*}
\val^{\sigma}(o^j_i) = \val^{\sigma}(c_j) + H_{d(i)-1} = \val^{\sigma}(c_j) + L_{d(i)}.
\end{equation*}
This is the correct output for \notg gate $i$ when gate $\inp(i)$ is true.
\end{itemize}

Note that in our description so far, $a^j_i$ has played the role of
``activating'' the gadget: the gadget will not compute the not of its input
until $a^j_i$ switches to $c_{1-j}$. It is for this reason that the appeal of
switching $a^j_i$ to $c_{1-j}$ has been very carefully chosen. Let $\sigma$ be a
policy with $\sigma(a^j_i) = c_{1-j}$. Recall that in phase $j$ we have
$\val^{\sigma}(c_{1-j}) = \val^{\sigma}(c_j) + T$. Therefore, we can apply
Lemma~\ref{lem:argappeal} to show that the appeal of switching $a^j_i$ to
$c_{1-j}$ is:
\begin{equation*}
p_1 \cdot (T - T + H_{d(i) - 1}) = \ajprime
\end{equation*}
This is a crucial property, because Dantzig's rule will always switch the action
with highest appeal. Therefore, \notg gates with depth $k$ will be considered by
Dantzig's rule before \notg gates with depth $k+1$. This will allow us to show
that Dantzig's rule computes the outputs of \notg gates in order of depth, which is
needed in our proof of correctness.

\section{The proof}
\label{sec:proof}

In this section, we prove that $\actionswitch$ is PSPACE-complete, by
reducing from $\bitswitch$. To this end, the circuit iteration instance $(F,
B^I, z)$, where $n = |B^I|$ is the bit-length of $F$ and $B^I$. This instance
will be fixed throughout this section. We also fix $C$ to be the negated form of
the circuit implementing $F$. We start by proving properties of the gadgets used
in $\const(C)$.

The following lemma is shown in Appendix~\ref{app:zero}. We will use this lemma
to show that, in  our initial policy for $\const(C)$, every state has expected
average-reward $0$, and in an optimal policy for $\const(C)$, every state has
expected average-reward $0$. As we explained in Section~\ref{sec:definitions},
this justifies the use of expected total-reward notation.
\begin{lemma}
\label{lem:zero}
We have:
\begin{itemize}
\item If $\sigma$ is a policy with $\sigma(l^0_i) = \sigma(r^0_i) = c_j$
for every input bit $i$, then $G^{\sigma}(s) = 0$ for every state $s$.
\item
Let $\sigma^*$ be an optimal policy in $\const(C)$. We have $G^{\sigma^*}(s) =
0$ for every state $s$.
\end{itemize}
\end{lemma}

\subsection{The clock}

We prove that the clock generates a sequence of $2^n$ different phases. Note
that since the clock has no actions to other gadgets, we can treat it separately
from the rest of the construction. We say that $\sigma$ is an initial policy for
the clock if $\sigma(i) = i-1$ for all states $i$ with $1 \le i \le n$. The
following lemma is shown in Appendix~\ref{app:clock}. 
\begin{lemma}
\label{lem:clock}
Let $\sigma_0$ be an initial policy for the clock. Dantzig's rule goes through an
exponential sequence of policies $\langle \sigma_0, \sigma_1, \dots
\sigma_{2^n-1}\rangle$ with the following properties:
\begin{itemize}
\itemsep1mm
\item If $j$ is even, then $\val^{\sigma_j}(c_1) = \val^{\sigma_j}(c_0) + T$.
\item If $j$ is odd, then $\val^{\sigma_j}(c_0) = \val^{\sigma_j}(c_1) + T$.
\item Dantzig's rule always switches an action with appeal in the range $[0.25,
0.5]$.
\end{itemize}
\end{lemma}

Lemma~\ref{lem:clock} shows the clock will go through an exponential sequence of
phases. Moreover, it shows that the appeal of advancing the clock always lies
in the range $[0.25, 0.5]$. Thus, if $\sigma$ is a policy in phase $j$, and if
Dantzig's rule switches from $\sigma$ to $\sigma'$ by switching an action with
appeal strictly greater than $0.5$, then $\sigma'$ will also be in phase $j$. We
will use this fact frequently throughout the rest of the proof. 

Note that the first phase of the clock is phase $0$, so at the start of our
construction we must load $B^I$ into the input bits of circuit $0$. Moreover,
the final phase of the clock is phase $1$, so at the end of our construction
circuit $0$ will copy $F^{2^n}(B^I)$ from the outputs of circuit $1$.

\subsection{Computing the circuit outputs}

\paragraph{\bf Coherent policies.}
In this section, we show that our circuit gadgets will correctly compute the
function $F$. That is, if we are in phase $j$, and if the input bits of circuit
$j$ are currently holding a bit-string $B$, then Dantzig's rule will eventually
switch to a policy in which the outputs of circuit~$j$ give $F(B)$.

Recall that during phase $j$ the input bits in circuit $j$ are in output
mode, and the input bits in circuit $1-j$ are in copy mode. Also recall that,
in order for \org gate $i$ in circuit $j$ to function correctly, we must have that
$x^j_i$ takes the action to $c_j$. We formalise these conditions, along with
some other technical conditions for circuit $1-j$, by defining \emph{coherent}
policies. We say that a policy~$\sigma$ is \emph{coherent} in phase $j$ if the
following conditions hold for every gate $i$:
\begin{itemize}
\itemsep1mm
\item If $i$ is an input bit then:
\vspace{1.5mm}
\begin{itemize}
\itemsep1mm
\item We have $\sigma(l^j_i) = c_j$ and $\sigma(r^j_i) = c_j$.
\item We have $\sigma(l^{1-j}_i) = c_{j}$ and $\sigma(r^{1-j}_i) = o^j_{\inp(i)}$.
\end{itemize}
\item If $i$ is an \org gate then we have $\sigma(x^j_i) = c_j$ and we have and
$\sigma(x^{1-j}_i) = c_j$.
\item If $i$ is a \notg gate then we have $\sigma(a^{1-j}_i) = c_j$.
\end{itemize}

The following Lemma, which is proved in Appendix~\ref{app:noswitch}, will be
crucial for our proof.
\begin{lemma}
\label{lem:noswitch}
Suppose that we are in phase $j$, and let $\sigma$ be a coherent policy. 
The following states have no actions with appeal strictly greater than $3.5$:
\begin{itemize}
\item $l^j_i$, $r^j_i$, $o^j_i$, $l^{1-j}_i$, $r^{1-j}_i$, and $o^{1-j}_i$ for every input bit $i$.
\item $x^j_i$ and $x^{1-j}_i$ for every \org gate $i$.
\item $a^{1-j}_i$ for every \notg gate $i$.
\end{itemize}
\end{lemma}
Intuitively, Lemma~\ref{lem:noswitch} states that if $\sigma$ is a coherent
policy, and if Dantzig's rule moves to a policy~$\sigma'$ by switching an action
with appeal at least $3.5$, then $\sigma'$ is also coherent. As we will show,  
when the construction evaluates the gates in circuit $j$, and copies the output
into the input bits
of circuit $1-j$, it always switches an action with appeal at least $3.5$.
Therefore, as this is done, the policy is always coherent.

\paragraph{\bf Circuit computation.}

Our goal is to show that circuit $j$ correctly computes the function $F$. In
order to do this, we use the following definition of correctness. Suppose that
we are in phase $j$, let~$\sigma$ be a policy, and let $B \in \{0, 1\}^k$ be
an input bit-string for $C$. We say that a gate $i$ is \emph{$B$-correct} 
in~$\sigma$ if one of the following conditions holds:
\begin{itemize}
\item If $C(B, i) = 1$, then we have $\val^{\sigma}(o^j_i) = \val^{\sigma}(c_j)
+ H_i$.
\item If $C(B, i) = 0$, then we have $\val^{\sigma}(o^j_i) = \val^{\sigma}(c_j) + L_i$.
\end{itemize}

We also use the following definition of a \emph{final} gate.
Suppose that we are in phase $j$, and let $\sigma$ be a coherent policy. We
say that a state $s$ is final if, for every action $a \in A_s$, we have
$\appeal^{\sigma}(a) \le 3.5$. Using this notion, we give an inductive
definition of finality for gates. We say that a gate $i$ is \emph{final} in
$\sigma$ if all gates $i'$ with $d(i') < d(i)$ are final, and one of the
following conditions is satisfied:
\begin{itemize}
\item $i$ is an \org gate, and $o^j_i$, $v^j_i$, and $x^j_i$ are final.
\item $i$ is a \notg gate, and both $o^j_i$ and $a^j_i$ are final.
\item $i$ is an input bit, and $o^j_i$, $l^j_i$ and $r^j_i$ are final.
\end{itemize}

\begin{lemma}
\label{lem:final}
Suppose that we are in phase $j$, and let $\sigma$ be a coherent policy. Suppose
further that Dantzig's rule moves from $\sigma$ to a policy $\sigma'$ by
switching an action with appeal at least $3.5$. If gate $i$ is final in
$\sigma$, then it is also final in $\sigma'$.
\end{lemma}

Our proof of correctness will be by induction over the depth of the gates. 
Lemma~\ref{lem:final} is critical for making this induction work, because once
we have shown that all gates with depth $k$ are correct and final, we can then
guarantee that these gates will not change their value while we consider the
gates with depth $k+1$.

We use the following definition as the base case of our induction. Let $B \in
\{0, 1\}^k$ be an input bit-string for $C$. We say that a policy $\sigma$ is an
\emph{initial policy for $B$} in phase $j$ if the following conditions are
satisfied.
\begin{itemize}
\itemsep1mm
\item $\sigma$ is coherent.
\item For every input bit $i$:
\vspace{1.5mm}
\begin{itemize}
\itemsep1mm
\item If the $i$-th bit of $B$ is $1$, then we have $\sigma(o^j_i) = l^j_i$.
\item If the $i$-th bit of $B$ is $0$, then we have $\sigma(o^j_i) = r^j_i$.
\item We have $\sigma(o^{1-j}_i) = l^j_i$.
\end{itemize}
\item For every \notg gate $i$, we have $\sigma(a^j_i) = c_j$.
\end{itemize}
Note that these conditions ensure that all input bits are $B$-correct in an
initial policy for $B$. Moreover, if $\sigma$ is an initial policy for $B$, then
since $\sigma$ is required to be coherent, we can apply
Lemma~\ref{lem:buffappeal} to argue that all input bits are final in $\sigma$.
This is used as the base case for the following inductive lemma, which is proved
in Appendix~\ref{app:circuitinductive}.

\begin{lemma}
\label{lem:circuitinductive}
Suppose that we are in phase $j$, let $B \in \{0, 1\}^n$ be an input bit-string
for $C$, and let $k > 0$. Suppose that the following three assumptions are met:
\begin{itemize}
\itemsep1mm
\item Dantzig's rule was started at $\sigma$, which is an initial policy for $B$.
\item Dantzig's rule has only ever switched actions with appeal greater than or
equal to $3.5 +
\frac{1}{2 \cdot k}$.
\item Dantzig's rule has arrived at $\sigma'$ in which all gates with depth at
most $k$ are final and $B$-correct. 
\end{itemize}
When Dantzig's rule is applied to $\sigma'$, it will move to a policy $\sigma''$
in which all gates with depth at most $k+1$ are final and $B$-correct. Moreover,
Dantzig's rule will only switch actions with appeal greater than or equal to
$3.5 + \frac{1}{2 \cdot (k+1)}$ while moving from $\sigma'$ to $\sigma''$.
\end{lemma}

Applying Lemma~\ref{lem:circuitinductive} inductively allows us to conclude
that, if $\sigma$ is an initial policy for some bit-string $B$, then when
Dantzig's rule is applied to $\sigma$, we will eventually move to a policy
$\sigma'$ in which all output gates in circuit $j$ are final and $B$-correct. In
the next lemma, which is proved in Appendix~\ref{app:circuit}, we show that the
input bits of circuit $1-j$ copy the output bits of circuit $j$.
\begin{lemma}
\label{lem:circuit}
Suppose that we are in phase $j$, let $B \in \{0, 1\}^n$ be an input bit-string
for $C$, and let $\sigma$ be an initial policy for $B$.
While making only
switches with appeal strictly greater than $3.5$, Dantzig's rule will eventually
switch to a policy $\sigma'$ in which, for every input bit $i$, we have:
\begin{itemize}
\itemsep1mm
\item If the $i$-th bit of $F(B)$ is $0$, then we have $\sigma'(o^{1-j}_i) = r^{1-j}_i$.
\item If the $i$-th bit of $F(B)$ is $1$, then we have that $\sigma'(o^{1-j}_i) =
l^{1-j}_i)$, and moreover, $o^{1-j}_i$ is never switched away from $l^{1-j}_i$
at any policy between $\sigma$ and $\sigma'$.
\end{itemize}
\end{lemma}

To sum up the results in this section, we introduce the following definition. We
say that $\sigma$ is a \emph{final policy for $B$} if the following conditions
are satisfied:
\begin{itemize}
\itemsep1mm
\item $\sigma$ satisfies the conditions of initial policy for $B$.
\item For every gate $i$ we have that $i$ is final and $B$-correct in $\sigma$.
\item For every input bit $i$ we have:
\vspace{1.5mm}
\begin{itemize}
\itemsep1mm
\item If $C(B, \inp(i)) = 0$, then $\sigma(o^{1-j}_i) = l^j_i$.
\item If $C(B, \inp(i)) = 1$, then $\sigma(o^{1-j}_i) = r^j_i$.
\end{itemize}
\end{itemize}
Lemmas~\ref{lem:circuitinductive} and~\ref{lem:circuit} combine to prove the
following lemma.
\begin{lemma}
Suppose that we are in phase $j$, let $B$ be an input bit-string, and let
$\sigma$ be an initial policy for $B$. When Dantzig's rule is applied to
$\sigma$, it will make a sequence of switches that each have appeal at least
$3.5$, and it will arrive at a policy $\sigma'$ that is a final policy for $B$.
\end{lemma}

\subsection{Phase transition}

Let $B$ be a bit-string. The results from the previous section imply that,
starting from an initial policy for $B$ in phase $j$, then Dantzig's rule will
eventually move to a final policy for $B$ in phase $j$. Once this has occurred,
the construction transitions into phase $1-j$, so that the
computation can be continued. In this section, we prove that this will occur.

The goal is to show that, if $\sigma$ is a final policy for some bit-string $B$
in phase $j$, then Dantzig's rule will move to an initial policy for $F(B)$ in
phase $1-j$. To do this, several events must occur. The input bits in circuit
$j$ must be switched to copy mode, and the input bits in circuit $1-j$ must be
switched to output mode. Moreover, the states $x^j_i$ in the \org gates, and the
states $a^j_i$ in the \notg gates must be switched from $c_j$ to $c_{1-j}$.

The following lemma, which is proved in Appendix~\ref{app:transition}, describes
how Dantzig's rule achieves these tasks. In particular, the order of these tasks
is controlled by the probabilities that we have chosen for the appeal reduction
gadgets in our construction. One crucial point to note here is that at no point
are the states $o^{1-j}_i$ switched for the input bits $i$, so the information
that was copied into these states remains intact.

\begin{lemma}
\label{lem:transition}
Suppose that we are in phase $j$, let $B$ be an input bit-string for $C$, and
let $\sigma_f$ be a final policy for $B$. When Dantzig's rule is applied to
$\sigma_f$, the following sequence of events takes place.
\begin{enumerate}
\itemsep1mm
\item For every input bit $i$, the state $l^{1-j}_i$ is switched to
$c_{1-j}$.
\item For every input bit $i$, the state $r^{1-j}_i$ is switched to $c_{1-j}$.
\item For every input bit $i$, the state $l^j_i$ is switched to $c_{1-j}$. 
\item For every input bit $i$, if $o^j_i$ takes the action towards $r^j_i$, then
it is switched to $l^j_i$.
\item For every input bit $i$, the state $r^j_i$ is switched to
$o^{1-j}_{\inp(i)}$.
\item For every \notg gate $i$, the state $a^{1-j}_i$ is switched to $c_{1-j}$.
\item For every \org gate $i$, the states $x^j_i$ and $x^{1-j}_i$ are both
switched to $c_{1-j}$.
\item An action in the clock is switched, and we move to phase $1-j$.
\end{enumerate}
At the end of this sequence, we will have arrived a coherent phase $1-j$ policy
$\sigma_n$, that is an initial policy for $F(B)$ in phase $1-j$.
\end{lemma}

\subsection{\actionswitch is PSPACE-complete}

Finally, we can provide a proof for Theorem~\ref{thm:actionswitch}. Recall that
our circuit iteration instance is $(F, B^I, z)$. We reduce this to
$\actionswitch(\const(C), \sigma_{\text{init}}, a)$, where:
\begin{itemize}
\label{page:init}
\item $\sigma_{\text{init}}$ is an initial policy for the clock, and an
initial policy for $B^I$ in phase $0$.
\item $a$ is the action from $o^0_{z}$ to $r^0_z$.
\end{itemize}
Note that $\sigma_{\text{init}}$ satisfies the condition of
Lemma~\ref{lem:zero}. Furthermore, recall that  $\bitswitch(F, B^I, z)$ only
requires us to decide something in the case where the $z$-th bit of $B^I$ is
$1$. Thus, we have that $\sigma_{\text{init}}$ does not use action $a$, and
therefore we have a valid instance of $\actionswitch$.


By Lemmas~\ref{lem:clock},~\ref{lem:circuit} and~\ref{lem:transition}, we know
that when Dantzig's rule is applied to $\sigma_{\text{init}}$, it will simulate
$2^n$ iterations of $F(B^I)$. 
Note that circuit
$0$ outputs $F^i(B^I)$ for every odd $i$, and circuit $1$ outputs $F^i(B^I)$
for every even $i$. Thus, for each even $i$, circuit~$1$ computes $F^{i}(B^I)$,
and circuit $0$ copies this value into its input bits. So, by
Lemma~\ref{lem:circuit}, the state $o^0_z$ switches to $r^0_z$ only in the case
where the $z$-th output bit of circuit $1$ is a $1$, which corresponds to the
$z$-th output bit of $F^i(B^I)$ being a $0$, for some even $i \le 2^n$.
Therefore, we have that $\actionswitch(\const(C), \sigma_{\text{init}}, a)$
is a ``yes'' instance if and only if $\bitswitch(F, B^I, z)$ is a ``yes''
instance, and we have completed the proof of Theorem~\ref{thm:actionswitch},
which then directly implies Theorem~\ref{thm:basisentry}.


\subsection{\dms is PSPACE-complete}

In this section, we give an overview of the proof of for
Theorem~\ref{thm:mdpend}. A formal proof of this theorem can be found in
Appendix~\ref{app:mdpend}. In order to show that $\dms$ is PSPACE-complete, we
make a slight modification to $\const(C)$, which we now describe. Let $S$ be the
set of states of $\const(C)$, and let $\sigma$ be an optimal policy for
$\const(C)$. We define $W = \max \{\val^{\sigma}(s) \; : \; s \in S\}$ to be the
largest possible value of a state in $\const(C)$. We modify $\const(C)$ by
adding an extra gadget, and extra actions at the states $l^0_z$ and $r^0_z$, as
shown in Figure~\ref{fig:modification}, and described formally in
Appendix~\ref{sec:modend}. We call this modified construction $\const(C, z)$.

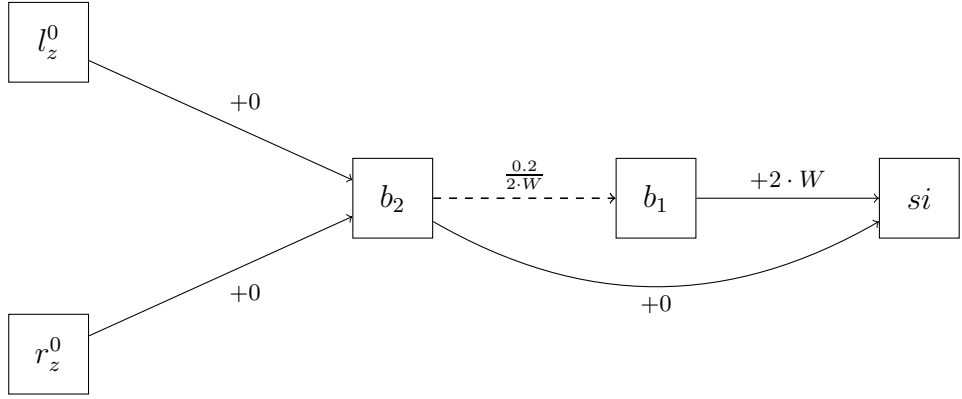
\begin{figure}
\begin{center}
\begin{tikzpicture}[node distance=3.5cm,auto]
\node[state,font=\large] (s) [] {$\sink$};
\node[state,font=\large,left of=s] (1) [] {$b_1$};
\node[state,font=\large,left of=1] (2) [] {$b_2$};
\node[state,font=\large,above left=1cm and 3.5cm of 2] (l) [] {$l^0_z$};
\node[state,font=\large,below left=1cm and 3.5cm of 2] (r) [] {$r^0_z$};
\path[arg]
	(2) edge node [] {$\frac{0.2}{2 \cdot W}$} (1)
	;
\path[->]
	(1) edge node [] {$+2 \cdot W$} (s)
	(2) edge [bend right] node [swap] {$+0$} (s)
	(l) edge node [] {$+0$} (2)
	(r) edge node [swap] {$+0$} (2)
	
	;
\end{tikzpicture} 
\end{center}
\caption{The extra gadget needed in $\const(C, z)$.}
\label{fig:modification}
\end{figure}

Recall that $\sink$ is the sink state in the clock. This gadget adds two new
states: $b_1$ and $b_2$. In the initial policy, $b_2$ will select the action
that goes directly to $\sink$, thus there will be no incentive for $l^0_z$ or
$r^0_z$ to switch to $b_2$. Note that the appeal of switching $b_2$ to $b_1$ is
$0.2$, and that this is smaller than the appeal of advancing the clock given in
Lemma~\ref{lem:clock}. Thus, the construction will proceed as normal until the
clock has gone through all $2^n$ phase switches, and it will be switched
immediately after the final phase of the construction. Note that at this point,
since $2^n$ is even, we will have that the choice made at $o^0_z$ will hold the
$z$-th bit of $F^{2^n}(B^I)$. 

When $b_2$ switches to $b_1$, the value of $b_2$ will rise to $2 \cdot W$. Since
the value of $l^0_z$ and $r^0_z$ can be at most $W$, we have that the appeal of
switching $l^0_z$ and $r^0_z$ to $b_2$ is at least $W$. Thus, Dantzig's rule
will immediately switch both of these states to $b_2$, and it can never switch
them away from~$b_2$. The key point here is that once $l^0_z$ and $r^0_z$ have
both been switched to $b_2$, the state $o^0_z$ is now indifferent between its
two actions. So no matter what choice is made at $o^0_z$, policy iteration can
not now switch~$o^0_z$. So, even though policy iteration may continue to switch
actions elsewhere in the construction, the information stored in $o^0_z$ cannot
be destroyed. Thus, when policy iteration terminates, the choice made at $o^0_z$
will determine the $z$-th bit of $F^{2^n}$. For this reason, we have that $\dms$
is PSPACE-complete, and therefore we have shown both Theorem~\ref{thm:mdpend}
and Theorem~\ref{thm:dls}. 


\bibliographystyle{abbrv}
\bibliography{references}

\newpage
\appendix

\section{A formal definition of the construction}
\label{app:construction}

Let $C$ be a boolean circuit. In this section, we give a formal definition of
$\const(C) = (S, (A_s)_{s \in S}, p, r)$. 

Before we begin, let us describe some shorthand notation. When we say that state
$s$ has a deterministic action to state $t$ with reward $q$, we mean that $A_s$
contains an action $a$ with $r(a) = q$, and $p(t, a) = 1$. We also give some
notation for the appeal reduction gadget. Let $s$ and $t$ be two states. We will
use $\gadget(s, t, r_d, p, r_f)$ to specify the following:
\begin{itemize}
\item A new state $(s, t)$ is added to $S$.
\item An action $a_{(s, t)}$ is added to $A_s$ with $r(a_{(s, t)}) = r_d$ where:
\begin{itemize}
\item $p((s, t), a_{(s, t)}) = p$.
\item $p(s, a_{(s, t)}) = 1-p$.
\end{itemize}
\item An action $a_{t}$ is added to $A_{(s, t)}$ with $r(a_t) = r_f$ and $p(t,
a_t) = 1$.
\end{itemize}

\paragraph{\bf The clock.}
The clock consists of the following states:
\begin{itemize}
\item A state $\sink$. This state has a deterministic action to $\sink$ with
reward $0$. 
\item A state $\sink'$. This state has a deterministic action to $\sink$ with
reward $T \cdot 2^{n+1}$.
\item A state $0$. This state has a deterministic action to $\sink$ with reward
$0$.
\item A state $1'$. This state has a single action $a$ with $r(a) = 0$, and
$p(\sink, a) = 0.5$, and $p(\sink', a) = 0.5$.
\item A state $1$. The actions at $1$ are $\gadget(1, 0, 0, \alpha_1, 0)$ and
$\gadget(1, 1', 0, \alpha_1, 0)$.
\item For each integer $i$ with $2 \le i \le n$, we have:
\begin{itemize}
\item A state $i'$. This state has a single action $a$ with $r(a) = 0$, and
$p((i-1)', 0.5)$ and $p(i-2, 0.5)$.
\item A state $i$. The actions at $i$ are $\gadget(i, i-1, 0, \alpha_i, 0)$ and
$\gadget(i, i', 0, \alpha_i, 0)$.
\end{itemize}
\item A state $c_0$ with a deterministic action to $n$ with reward $0$. 
\item A state $c_1$. This state has a single action $a$ with $r(a) = 0$, and
$p(n-1, a) = 0.5$ and $p(n', a) = 0.5$.
\end{itemize}

\paragraph{\bf Input bits.}
For each input bit $i$ in circuit $C$, and each $j \in \{0, 1\}$, the
construction contains the following states:
\begin{itemize}
\item A state $l^j_i$. The actions at $l^j_i$ are $\gadget(l^j_i, c_{1-j}, 0,
p_5, -\frac{T}{2} + \frac{H_{d(C)} + L_{d(C)}}{2})$ and $\gadget(l^j_i, c_j, 0,
p_4, H_0)$.
\item A state $r^j_i$. The actions at $r^j_i$ are $\gadget(r^j_i, c_j, 0, p_6,
L_0)$ and $\gadget(r^j_i, o^{1-j}_{\inp(i)}, 0, p_7, -\frac{T}{2})$.
\item A state $o^j_i$. This state has a deterministic action to $r^j_i$ with
reward $0$, and it also has $\gadget(o^j_i, l^j_i, 0, p_3, 0)$.
\end{itemize}

\paragraph{\bf \org gates}
For each \org gate $i$ in circuit $C$, and each $j \in \{0, 1\}$, the
construction contains the following states:
\begin{itemize}
\item A state $x^j_i$. The actions at this state are $\gadget(x^j_i, c_j, 0,
\frac{\xj}{T}, 0)$ and $\gadget(x^j_i, c_{1-j}, 0, \frac{\xj}{T}, 0)$.
\item A state $v^j_i$. This state has a deterministic action to
$o^j_{\inp_1(i)}$ with reward $0$, and a deterministic action to
$o^j_{\inp_2(i)}$ with reward $0$.
\item A state $o^j_i$. This state has a deterministic action to $x^j_i$ with
reward $L_{d(i)}$, and a deterministic action to $v^j_i$ with reward $b_{d(i)}$.
\end{itemize}

\paragraph{\bf \notg gates}
For each \notg gate $i$ in circuit $C$, and each $j \in \{0, 1\}$, the
construction contains the following states:
\begin{itemize}
\item A state $a^j_i$. The actions at this state are $\gadget(a^j_i, c_j, 0,
p_2, 0)$ and $\gadget(a^j_i, c_{1-j}, 0, p_1, -T + H_{d(i) - 1})$.
\item A state $o^j_i$. This state has a deterministic action to $o^j_{\inp(i)}$
with reward $0$, and it also has an action $\gadget(o^j_i, a^j_i, 1, \frac{1}{b_{d(i)}},
0)$.
\end{itemize}

\subsection{Modifications for Theorem~\ref{thm:mdpend}}
\label{sec:modend}

In order to prove Theorem~\ref{thm:mdpend}, we have to add an extra gadget to
the construction. In this section, we formally define $\const(C,z) = (S,
(A_s)_{s \in S}, p, r)$. Firstly, $\const(C,z)$ contains every state and every
action from $\const(C)$. It also contains the following states: 
\begin{itemize}
\item A state $b_1$. This state has a deterministic action to $\sink$ with
reward $2 \cdot W$.
\item A state $b_2$. This state has a deterministic action to $\sink$ with
reward $0$, and it has an action $\gadget(b_2, b_1, 0, \frac{0.2}{2 \cdot W},
0)$.
\item An extra deterministic action is added from $l^0_z$ to $b_2$ with reward
$0$.
\item An extra deterministic action is added from $r^0_z$ to $b_2$ with reward
$0$.
\end{itemize}

\section{Proof of Lemma~\ref{lem:argappeal}}
\label{app:argappeal}

\begin{proof}
We begin by showing the first property. In this case we have:
\begin{equation*}
\val^{\sigma}(s') = \val^{\sigma}(t) + r_f.
\end{equation*}
Therefore, since $\sigma(s) = a$, we have:
\begin{equation*}
\val^{\sigma}(s) = p \cdot (\val^{\sigma}(t) + r_f) + (1 - p) \cdot \val^{\sigma}(s) + r_d.
\end{equation*}
Solving this equation for $\val^{\sigma}(s)$ yields:
\begin{equation*}
\val^{\sigma}(s) = \val^{\sigma}(t) + r_f + \frac{r_d}{p}.
\end{equation*}
This completes the proof of the first property.

We now turn our attention to the second property. In this case we have: 
\begin{align*}
\appeal^{\sigma}(a) &= p \cdot\left(\val^{\sigma}(t) + r_f\right) + (1 - p)
\cdot \val^{\sigma}(s) + r_d - \val^{\sigma}(s)\\
&= p \cdot\left(\val^{\sigma}(t) - \val^{\sigma}(s) + r_f\right) + r_d \\
&= p \cdot\left(\val^{\sigma}(s) + b - \val^{\sigma}(s) + r_f\right) + r_d \\
&= p \cdot (b + r_f) + r_d.
\end{align*}
This completes the proof of the second property.
\qed
\end{proof}

\section{Equivalence of Dantzig's pivot and switching rules}
\label{app:dantzig}

\newcommand{\sbar}{\bar{S}}
\newcommand{\reals}{\mathbb{R}}
\newcommand{\cc}{\mathbf{c}}
\newcommand{\bb}{\mathbf{b}}
\newcommand{\e}{\mathbf{e}}
\newcommand{\x}{\mathbf{x}}
\newcommand{\y}{\mathbf{y}}
\newcommand{\A}{\mathbf{A}}
\newcommand{\B}{\mathbf{B}}
\newcommand{\0}{\mathbf{0}}
\newcommand{\bfs}{\textsc{Bfs}\xspace}
\newcommand{\primal}{\ensuremath{(P)}\xspace}
\newcommand{\dual}{\ensuremath{(D)}\xspace}

Even though our construction is technically a multi-chain MDP with the average
reward criterion, we do not require the fully general linear programming
formulation for this setting (for that see Section 9.3 (pg. 462) of
\cite{Put94}). We utilize the special features of our construction, such the
fact that, for every policy considered during policy iteration, the only
recurrent state is $\sink$, and the fact that the initial policy and all optimal
policies have 0 expected average reward, in order to define a simpler linear
program and an initial basic feasible solution with the property that policy
iteration with Dantzig's switching rule will perform identically to the simplex
method with Dantzig's pivot rule applied to the linear program. Let $\bar{S} = S
\setminus \{\sink\}$, let $n = |\bar{S}|$, and let $m = |A|$. 
For convenience, we pick an ordering of all actions and associate them with
$[m]$.

We first construct our linear programs. We start by writing down
Equation~\eqref{eqn:optval} as an LP, which we will call the \emph{dual}. We
then take the dual of this LP, to create an LP which we call the \emph{primal}.
We will execute the simplex method on the primal LP.



\begin{alignat}{5}
\textsc{Primal ($P$):} \quad &\text{maximize}   
& \sum_{s \in \sbar} \sum_{a \in A_s} r(a) \cdot x_a \nonumber\\
&\text{subject to} 
&\quad \sum_{a \in A_j} x_a - \sum_{s \in \sbar} \sum_{a \in A_{s}} p(s,a) \cdot x_a & =
\frac{1}{n} & \quad \forall j \in \sbar \label{eq:primal_bfs}\\
			&& x_a & \ge 0 & \forall a \in A
\end{alignat}

\begin{alignat}{5}
\textsc{Dual ($D$):} \quad 
& \text{minimize} \quad & \sum_{j \in S} \frac{1}{n} \cdot v_j & \nonumber \\
& \text{subject to} \quad 
& v_s - \sum_{j \in \sbar} p(j,a) \cdot v_j & \ge r(a) \quad & \forall s \in
\sbar,\ a \in A_s \label{eq:dual_values}
\end{alignat}

Parts of the following exposition on Linear Programming are closely based on~\cite{HKZ14},
which deals with the LP formulation of discounted deterministic MDPs.
We assume that \primal and \dual are written down in the following
standard forms, respectively:
\begin{alignat*}{7}
	\max && \quad \cc^\top & \x &\quad\quad \min && \quad \bb^\top & \y\\ 
	\text{s.t.} && \A&\x= \bb &\quad\quad \text{s.t.} && \A^\top & \y \ge \cc
	\label{eq:lp_standard_dconst}\\
						 && &\x \ge \0.
\end{alignat*}
The constraint matrix $\A$ is an $m \times n$ matrix.
%
Let $B \subset [m]$ with $|B|=n$. 
We let $\B=\A_B \in \reals^{n \times n}$ be the $n \times n$ matrix obtained by
selecting the columns of $\A$ whose indices belong to $B$.
If the columns of $\B$ are linearly independent, then $B$ is a \emph{basis} and
$\B$ is a \emph{basis matrix}.
Then there is a unique vector $\x \in \reals^m$ for which $\A\x=\bb$ and
$\x_i=0$ for $i \notin B$.
If we let $N = [m] \setminus B$, then 
$$
x_B := \B^{-1}\bb\quad \text{and}\ \x_N:=\0.
$$
The vector $\x$ is the \emph{basic solution} corresponding to $B$
If $\x_B \ge \0$, then $\x$ is said to be a \emph{basic feasible solution}
(\bfs).
A \bfs is a vertex of the polyhedron corresponding to \primal.
The variables in $\{x_i \ \mid \ i \in B\}$ are referred to as \emph{basic} 
variables, while the variables in $\{x_i \ \mid \ i \in N\}$ are referred to as
\emph{non-basic} variables.
Each basis $B$ also defines a solution to the dual as follows:
\begin{equation}
\label{eq:dual_solution}
\y = ((\A_B)^{-1})^\top \cc_B.
\end{equation}
This dual solution is feasible if and only if the primal solution is optimal.
The objective function can also be expressed as $\cc^\top\x=\cc_B^\top\x_B +
\bar{\cc}\x$, where 
\begin{equation}
\label{eq:reduced_costs}
\bar{\cc}= \cc - (\cc_B^\top\B^{-1}\A)^\top = \cc - \A^\top y.
\end{equation}
The vector $\bar{\cc} \in \reals^m$ is referred to as the vector of
\emph{reduced costs}.
Note that $\bar{\cc}_B = 0$, i.e., the reduced costs of the basic variables
are all 0.
Since \primal is written as a maximization problem, our reduced costs are
positive and Dantzig's pivot rule will choose the variable with the largest
reduced cost to enter the basis, and if $\bar{\cc} \le 0$, then $\x$ is an
\emph{optimal} solution.

Next, we show how to derive an initial basic feasible solution for \primal from
an initial policy for our MDP construction.

\begin{lemma}
\label{lem:initpol_upper}
Let $\sigma_{\text{init}}$ be the initial policy for $\const(C)$, as defined on
page \pageref{page:init}. Let $\mathbf{M} \cdot \x = \mathbf{b}$ be the linear system
obtained from~\eqref{eq:primal_bfs} by setting $x_a = 0$ for all $s \in \sbar$
and $a \in A$ for which $\sigma(s) \ne a$. The matrix $\mathbf{M}$ is upper triangular
with non-zero entries on the diagonal, and is thus non-singular.
\end{lemma}
\begin{proof}
Firstly, we show that the matrix is upper triangular. This follows from the fact
that, in $\sigma_{\text{init}}$, there cannot be a pair of states $s$ and $s'$,
with $s' \ne s$ such that:
\begin{itemize}
\item The probability of moving from $s$ to $s'$ under $\sigma_{\text{init}}$ is
strictly greater than $0$.
\item The probability of moving from $s'$ to $s$ under $\sigma_{\text{init}}$ is
strictly greater than $0$.
\end{itemize}
This property essentially says that there are no cycles of length greater than
one in
$\sigma_{\text{init}}$. Note that, there are cycles of length $1$, but that all
of these cycles arise due to an appeal reduction gadget.

There clearly cannot be a cycle with length greater than $1$ in the clock,
because apart from the cycles of length $1$ in the appeal reduction gadgets,
there are no other cycles in the clock. It can be verified that the only
possible way to create a cycle with length greater than $1$ in the circuits is
to choose the action from $r^0_i$ to $o^{1}_{\inp(i)}$ and the action from
$r^0_{i'}$ to $o^0_{\inp(i)}$ for some pair of input bits $i$ and $i'$. But this
is impossible for $\sigma_{\text{init}}$, because by defintion we have
$\sigma_{\text{init}}(l^0_i) = \sigma_{\text{init}}(r^0_i) = c_0$.
So,~$\mathbf{M}$ must be upper triangluar.

Next we must argue that the diagonal elements of $\mathbf{M}$ are non-zero. Note
that the diagonal elements of $\mathbf{M}$ are zero if and only if there is a
state $s$ such that $p(s, \sigma_{\text{init}}(s)) = 1$. But no action in
$\const(C)$ has this property, because in every appeal reduction gadget, we have
$p > 0$. Therefore, every diagonal element of $\mathbf{M}$ is non-zero, and thus
we have shown that $\mathbf{M}$ is invertible.  \qed
\end{proof}

Now we prove that there is a one-to-one correspondence between the basic
feasible solutions encoutered by Dantzig's pivot rule and the (deterministic)
policies encountered by Dantzig's switching rule.
To do this we use the following correspondence between policies and bases.

\begin{definition}
For a (deterministic) policy $\sigma$, we define a corresponding 
basis $B(\sigma)$ as follows:
$$
B(\sigma) := \{\ a\ |\ \exists\ s \in S\ \text{s.t.}\ \sigma(s) = a\}.
$$
\end{definition}

\begin{lemma}
\label{lem:init_bfs}
The vector $\x_{B(\sigma_{\text{init}})}$ is basic feasible solution of the primal \primal.
\end{lemma}

\begin{proof}
By Lemma~\ref{lem:initpol_upper}, the columns of $\A_B$ are linearly independent.
\qed
\end{proof}

We now show that, starting from $\x_{B(\sigma_{\text{init}})}$ every basic
feasible solution encountered by Dantzig's rule will correspond to a basis
$B(\sigma)$ for some deterministic policy $\sigma$.

\begin{lemma}
Starting from $\x_B(\sigma_{\text{init}})$,
Dantzig's pivot rule (with any type of degeneracy resolution) will go through a 
finite sequence of basic feasible solutions corresponding to the bases:
$$
B(\sigma_{\text{init}=0}), B(\sigma_1), B(\sigma_2), \ldots, B(\sigma_l).
$$
for some sequence $\sigma_1, \sigma_2, \ldots$ of distinct (deterministic) policies.
\end{lemma}

\begin{proof}
We prove the claim by induction.
The base case is Lemma~\ref{lem:init_bfs}.
Suppose that, as inductive hypothesis, that the $k$-th basis is 
$B(\sigma_k)$ for some deterministic policy $\sigma_k$. We need to show that
the $(k+1)$-the basis corresponds to
$B(\sigma_{k+1})$ for some policy $\sigma_{k+1}$.
In the pivot operation that produces basis $k+1$ from $k$, one non-basic
variable becomes basic, say $x_a$ for some $a \in A_s$ and $s \in S$.
To prove the claim we require that the basic variable that leaves the basis 
is the unique one in $B(\sigma_k)$ that corresponds to another action at $s \in
S$. 
Towards a contradiction, suppose this is not the case.
Then, a basic variable, say $x_{a'}$ for some action $a' \in A_{s'}$ for
some state $s' \ne s$ leaves the basis.
However, in that case there is no way to satisfy the constraint in
\eqref{eq:primal_bfs} that corresponds to $s'$, since now all $x_i$ for $i \in
A_{s'}$ will be zero, the left-hand side of this constraint will be non-positive, while
the right-hand side is positive. 
This contradicts the fact that the $(k+1)$-th basis corresponds to a basic
feasible solution.
\qed
\end{proof}

Finally, we show that the reduced costs of a basic feasible solution that
corresponds to a basis $B(\sigma)$ for some (deterministic) policy $\sigma$ are
the appeals used for Dantzig's switching rule.

\begin{lemma}
Let $\sigma$ by a (deterministic) policy and let $B(\sigma)$ be the resulting
basis for \primal.
For every action $a$ in $B(\sigma)$, the corresponding reduced cost, defined by
\eqref{eq:reduced_costs}, is $\appeal^\sigma(a)$, as defined in
\eqref{eq:appeal}.
\end{lemma}

\begin{proof}
According to \eqref{eq:dual_solution}, the basis $B(\sigma)$
defines a dual solution that satisfies \eqref{eq:val}, thus
$\y = \val^\sigma$.
By definition, the reduced costs that correspond to $B(\sigma)$ equal $\cc -
\A^\top y$. Thus, for each action $a \in B(\sigma)$, the reduced cost is:
$$
r(a) - v_s + \sum_{j \in \sbar} p(j,a) \cdot v_j.
$$
In other words, the reduced costs of every action $a \in B(\sigma)$ equals
$\appeal^\sigma(a)$, as required.
\qed
\end{proof}

Thus, when we apply Dantzig's pivot rule to \primal starting from
$x_{B(\sigma_{\text{init}})}$ for some initial policy~$\sigma_{\text{init}}$,
we exactly simulate Dantzig's switcing rule for policy iteration starting from 
$\sigma_{\text{init}}$.

\section{Proof of Lemma~\ref{lem:clock}}
\label{app:clock}

\newcommand*\xor{\mathbin{\oplus}}


We start by defining a reflected binary Gray code, which is a well-known
construction that arises, for example, in the Klee-Minty examples that first 
showed that the simplex algorithm could be exponential in the worst case.

\paragraph{Notation.} Throughout this proof, we use the following notation.
\begin{itemize}
\itemsep2mm
\item
$n$ represents a fixed positive integer.
\item
\Index, short for iterations, denotes the set $\{0,1,\ldots,2^n-1\}$. 
\item
$\bit(x,i)$ denotes the $i$-th bit of $x \in \Index$, which corresponds to
$2^{i-1}$ for $i=1,\ldots,n$, i.e., when we talk about the $i$-th bit we count
from right to left.
\item
$\shift(x,i)$ denotes $\lfloor x/2^i \rfloor$, that is, a shift to the right (and truncation) by $i$ bits.
\item
$x \xor y$ denotes the bit-wise exclusive or (XOR) of integers $x$ and $y$.
\item
For $i \in [n]$, we define $f(i) := n - i + 1$. Note that $f$ is an involution.
\item 
We often consider $j + 2^i$. Note that for $k \le i$ we have
$\bit(j,k)=\bit(j+2^i,k)$, i.e., only bits $\bit(j+2^i,k)$ for $k>i$ differ
from $\bit(j,k)$. In the following, during intermediate calculations, $j+2^i$ may be larger that $2^n-1$, but
whenever this occures a \shift operation ensures that no non-zero bits in
positions greater than $n$ are ever
``ignored'' when we consider only $n$ bits in the end.
\end{itemize}

\begin{definition}[Reflected Binary Gray Code]
\label{def:gcxor}
We define $g: \Index \mapsto \{0,1\}^n$ according to the following 
bit-wise definition. For $i=1,\ldots,n$ and
$j \in \Index$:
\begin{equation} 
	\bit(g(j),i) = \bit(j \xor \shift(j,1),f(i)).
\end{equation}
Equivalently,
\begin{equation}
\label{eq:gcxor}
\bit(g(j),i) = 
\begin{cases}
	1 & \mbox{iff}\quad \bit(j,f(i)) \ne \bit(j,f(i)+1), \\
	0 & \mbox{iff}\quad \bit(j,f(i)) = \bit(j,f(i)+1).
\end{cases}
\end{equation}
We define the sequence $G := \bigl(g(0),g(1),\ldots,g(2^n-1)\bigr)$.
\end{definition}

$G$ contains all possible bit-strings of length $n$ exactly once, with 
\begin{alignat*}{2}
	g(0) & =0^n \\
	g(1) & =10^{n-1} \\
	g(2^n-1) & =0^{n-1}1.
\end{alignat*}

$g(i+1)$ is obtained from $g(i)$ by flipping a single bit, as described in
following lemma.

\begin{lemma}
\label{lem:gcflip}
For $j \in \Index$, bit-string $g(j+1)$ is obtained from $g(j)$ by flipping bit 
$f(\lsz(j))$, where 
$$
\lsz(j) := \min \{k \in [n] \ \mid \ \bit(j,k) = 0\}$$ is the position of the
least significant zero bit of $j$.  
\end{lemma}

\begin{proof}
	For $i > \lsz(j)$, we have $\bit(j,i)=\bit(j+1,i)$, and thus
	$\bit(g(j),f(i))=\bit(g(j+1),f(i))$. For $i=1,\ldots,\lsz(j)-1$, we have 
\begin{alignat*}{4}
	\bit(j,i) & =\bit(j,i+1) & = 1 & \ne \bit(j,\lsz(j)), & \mbox{and}\\
	\bit(j+1,i) & =\bit(j+1,i+1) & = 0 & \ne \bit(j+1,\lsz(j)).
\end{alignat*}
Thus, we have for $i=1,\ldots,\lsz(j)-1$:
\begin{alignat*}{2}
	\bit(g(j),f(i)) & =\bit(g(j),f(i+1)) & = 0.
\end{alignat*}
And, since we have:
\begin{alignat*}{3}
\bit(j,\lsz(j)) & =\bit(j+1,\lsz(j)) & \mbox{and}
\bit(j,\lsz(j)+1) & \ne \bit(j+1,\lsz(j)+1),
\end{alignat*}
we have shown that 
$$
\bit(g(j),f(\lsz(j)) \ne \bit(g(j+1),f(\lsz(j))),$$
which completes the proof.\qed
\end{proof}

The sequence $G$ corresponds to the sequence of policies that 
Dantzig's rule will generate, where the bit-strings $g(j)$ are interpreted
as policies according to the following definition.


\begin{definition}
\label{def:bit-string_interp}
For $j \in \Index$ and $i \in [n]$, the bit-string $g(j)$ corresponds to the
policy $\sigmagj$ as follows.
$$
\sigmagj(i) =
\begin{cases}
	i-1 & \mbox{if} \quad \bit(g(j),i) = 1\\
	i'  & \mbox{if} \quad \bit(g(j),i) = 0.
\end{cases}
$$
\end{definition}

For brevity, we often use $g(j)$ to mean the policy \sigmagj.
In the initial policy $g(1) = 0^n$, all states $n,n-1,\ldots,1$ point
right. We will show that:
\begin{itemize}
\item
$g(j)$ contains a switchable state if and only if there exists an $i$
with $\bit(j,i)=0$ (Corollary~\ref{cor:switchability}).
This implies that the final policy is $g(2^n-1)=0^{n-1}1$, where all states
point to the right, except $1$, which points down. 
\item
Lemma~\ref{lem:gcflip} shows that $g(j+1)$ is obtained from $g(j)$ by 
flipping the bit in position $f(\lsz(j))$ of $g(j)$. 
\item
Lemma~\ref{cor:switchk} shows that at $g(j)$, Dantzig's rule
flips $f(\lsz(j))$ and thus moves to $g(j+1)$. 
\end{itemize}
This implies that there are $2^n-1$ iterations in total. The main result of 
this section is 
Lemma~\ref{lem:clock_works}, which characterizes the values of $c_0$ and $c_1$
at different iterations and bounds the appeal of actions within the clock.


In the clock, since all rewards are 0 except for the reward for the unique action
at $\sink'$, which goes from $\sink'$ to $\sink$, the value of state is the probability of
reaching $\sink'$ multiplied by that reward, $T \cdot 2^{n+1}$. Consequently, for
brevity we work with the following scaled values to avoid the repetition of $T$
throughout our calculations.

\begin{definition}
For a policy $\sigmagj$, and for all $i \in [n]$, we define:
\begin{equation}
\label{def:valgjnobar}
\valgj(i) := T^{-1} \cdot \valgjnobar(i).
\end{equation}
\end{definition}

\begin{lemma}
\label{lem:clock_rec_values}
For every $j \in \Index$, and the resulting policy $\sigmagj$, we have:
\begin{alignat}{3}
	\valgj(0)  & =0
	\label{eq:val_0}
	\\
	\valgj(1') & =2^n
	\label{eq:val_1_prime}
	\\
	\valgj(2') & = \half \valgj(0) + \half \valgj(1') = 2^{n-1}
	\label{eq:val_2_prime}
	\\
	\valgj(i') & =\half \valgj((i-1)') + \half \valgj(i-2) &\quad \mbox{for } i\ge1
	\label{eq:val_i_prime_rec}
	\\
	\valgj(i) & = \begin{cases} \valgj(i-1) &\mbox{if } \bit(g(j),i) = 0, \\ 
										 \valgj(i')  & \mbox{if } \bit(g(j),i) = 1.
									 \end{cases} & \quad \mbox{for } i\ge1.
	\label{eq:val_i_rec}
\end{alignat}

\end{lemma}

\begin{proof}
	These follow by the definition of $\valgj$ and $g(j)$. \qed
\end{proof}

The following three lemmas are easy technical lemmas that are are used only in the proof of
Lemma~\ref{lem:clock_values}.

\begin{lemma}
\label{lem:bitshift1}
The following holds for all $j \in \Index$ and $a \in [n]$:
$$
\shift(j,a+1) + \shift(j + 2^{a},a+1)= \shift(j,a).
$$
\end{lemma}

\begin{proof}
We consider two cases. 
First, suppose that $\bit(j,a+1) = 0$. Then, 
$j + 2^a$ differs from $j$ only in bit $a+1$, and 
so 
$$\shift(j + 2^a,a+1) = \shift(j,a+1).$$
Then we have $2 \cdot \shift(j,a+1) = \shift(j,a)$ as required, which
holds because $\bit(j,a+1) = 0$.
Second, suppose that $\bit(j,a+1) = 1$. Then, 
$$\shift(j + 2^a,a+1) = \shift(j,a+1) + 1.$$
Then we have $2 \cdot \shift(j,a+1) + 1 = \shift(j,a)$ as required, which
holds because $\bit(j,a+1) = 1$.
\qed
\end{proof}

\begin{lemma}
\label{lem:bitshift2}
If $\bit(j,a) = \bit(j,a+1)$ the following holds:
\begin{equation}
\shift\bigl(j + 2^{a-1},a\bigr) = 2 \cdot \shift\bigl(j + 2^a,a+1)\bigr).
\end{equation}
\end{lemma}

\begin{proof}
Since $\bit(j,a) = \bit(j,a+1)$, then we have for all $k > a+1$ that
$$
\bit(j + 2^{a-1},k)=\bit(j + 2^{a},k),
$$
and $\bit(j + 2^{a-1},a+1)=0$. The claim follows.
\qed
\end{proof}

\begin{lemma}
\label{lem:bitshift3}
If $\bit(j,a) \ne \bit(j,a+1)$ the following holds:
\begin{equation}
1 + 2 \cdot \shift\bigl(j,a+1\bigr) = \shift\bigl(j + 2^{a-1},a\bigr).
\end{equation}
\end{lemma}

\begin{proof}
We consider two cases. First, suppose that $\bit(j,a) = 0$. Then, because
$\bit(j,a+1) = 1$ by assumption, we have that 
$$
1 + 2 \cdot \shift\bigl(j,a+1\bigr) = \shift\bigl(j,a\bigr).
$$
Moreover, then, since $\bit(j,a) = 0$, the bits $j + 2^{a-1}$ and $j$
differ only in position $a$, so  $\shift\bigl(j + 2^{a-1},a\bigr) =
\shift\bigl(j,a\bigr)$, as required.  Second, suppose that $\bit(j,a) = 1$. 
Then, because $\bit(j,a+1) = 0$, we have that 
$$1 + 2 \cdot \shift\bigl(j,a+1\bigr) = \shift\bigl(j,a\bigr) + 1.$$
Moreover, then, since $\bit(j,a) = 1$, the bits of $j + 2^{a-1}$ and $j$
differ only in positions $a$ and $a+1$, which are both flipped.
In particular, $\bit(j + 2^{a-1},a+1) = 1$, so  $\shift\bigl(j + 2^{a-1},a\bigr) =
\shift\bigl(j,a\bigr) + 1$, as required.
\qed
\end{proof}


\begin{lemma}
\label{lem:clock_values}
Let $i \in [n]$ and $j \in \Index$. 
Let:
\begin{alignat}{2}
X(j,i) & = 2^{f(i)} + \shift(j,f(i)+1) \cdot 2^{f(i)+1},
\label{eq:clockavgvalues}\\
Y(j,i) & = \shift(j + 2^{f(i)-1},f(i)) \cdot 2^{f(i)}.
\label{eq:clockvalues}
\end{alignat}
We have
\begin{alignat}{2}
\valgj(i')& = X(j,i),
\label{eq:val_i_prime}\\
\valgj(c_1)& = X(j,n+1) = 1 + 2 \cdot \shift(j,1),
\label{eq:val_c1}\\
\valgj(i) & = Y(j,i).
\label{eq:val_i}
\end{alignat}
\end{lemma}

\begin{proof}
We prove the claim by induction on $i$.

\paragraph{Base case:}
For the base cases, we show that \eqref{eq:val_i_prime} holds for $i=1,2$,
and that \eqref{eq:val_i} holds for $i=1$.

According to~\eqref{eq:val_1_prime} and~\eqref{eq:val_2_prime} in
Lemma~\ref{lem:clock_rec_values}, for all $j\in \Index$ we have
$\valgj(1') =2^n$ and $\valgj(2') = 2^{n-1}$.
This agrees with \eqref{eq:val_i_prime}, since for $i=1,2$, we have 
$\shift(j,f(i)+1) = 0$, and thus $X(j,i) = 2^{f(i)}$, as required.
According to Definition~\ref{def:gcxor} we have:
$$
\bit(g(j),1) = 
\begin{cases} 
0 & \mbox{if } j < 2^{n-1},\\
1 & \mbox{if } j \ge 2^{n-1}.\\
\end{cases}
$$
Note that $f(1)=n$.
If $j < 2^{n-1}$, then $\bit(g(j),1) = 0$ and state $1$ points to state $0$, so
$\valgj(i)=0$, which agrees with \eqref{eq:val_i}, since $\shift(j + 2^{n-1},n) = 0$.
If $2^{n-1} \le j \le 2^n - 1$, then $\bit(g(j),1) = 1$ and state $1$ points to state
$1'$, so $\valgj(i)=2^n$, which agrees with \eqref{eq:val_i}, 
since $\shift(j + 2^{n-1},n) = 1$.

\paragraph{Induction:}

We now prove that \eqref{eq:val_i_prime} holds for for $i > 2$.
As induction hypothesis we assume that \eqref{eq:val_i_prime} and
\eqref{eq:val_i} hold for all $k<i$. 
By \eqref{eq:val_i_prime_rec}, we have:
\begin{alignat*}{3}
\valgj(i') & = \half \valgj((i-1)') + \half \valgj(i-2)\\
				& = \half X(j,i-1) + \half Y(j,i-2) & \mbox{By the inductive hypothesis.}\\
				& = 2^{f(i)} + \bigl( \shift(j,f(i)+2) + \shift(j + 2^{f(i)+1},f(i)+2) \bigr) \cdot 2^{f(i)+1} \\
				& = 2^{f(i)} + \shift(j,f(i)+1) \cdot 2^{f(i)+1} & \mbox{By
			  Lemma~\ref{lem:bitshift1} with $a=f(i)+1$.}\\
				& = X(j,i). & \mbox {As required.}
\end{alignat*}
This proves \eqref{eq:val_i_prime}, and using exactly the same reasoning,
this shows \eqref{eq:val_c1} too:
\begin{alignat*}{2}
	\valgj(c_1) & = \half \valgj(n') + \half \valgj(n-1) \\
				& = 1 + 2 \cdot \bigl(\shift(j,2) + \shift(j+2,2)\bigr) \\
				& = 1 + 2 \cdot \shift(j,1). & \quad \mbox{By Lemma~\ref{lem:bitshift1} with $a=1$.}
\end{alignat*}
We now prove that \eqref{eq:val_i} holds for for $i > 1$.
We consider two cases:

\begin{enumerate}

\item $\bit(g(j),i) = 0$. Then by \eqref{eq:val_i_rec} we have that 
$\valgj(i)  = \valgj(i-1)$, and we need to show that
$Y(j,i) = Y(j,i-1) = \valgj(i-1)$, where the second equality 
holds by the inductive hypothesis.
By \eqref{eq:clockvalues}, we have
\begin{alignat*}{3}
\label{eq:value1}
Y(j,i-1) & = \shift(j + 2^{f(i-1)-1},f(i-1)) \cdot 2^{f(i-1)}\\
         & = \shift(j + 2^{f(i)},f(i)+1) \cdot 2^{f(i)+1} \\
         & = \bigl(2 \cdot \shift(j + 2^{f(i)},f(i)+1)\bigr) \cdot 2^{f(i)} \\
         & = \shift(j + 2^{f(i)-1},f(i)) \cdot 2^{f(i)} & \quad \mbox{By Lemma~\ref{lem:bitshift2} with $a=f(i)$.} \\
         & = Y(j,i).
\end{alignat*} 

\medskip
\item $\bit(g(j),i) = 1$. Therefore, according to~\eqref{eq:val_i_rec}, we
have $\valgj(i)  = \valgj(i'))$, and we need to show that $Y(j,i)
= X(j,i)$, since we have just proved that $ \valgj(i')=X(j,i)$.  
Since $\bit(g(j),i) = 1$, we have $\bit(j,f(i)) \ne \bit(j,f(i)+1)$ by
Definition~\ref{def:gcxor} and we can apply Lemma~\ref{lem:bitshift3} in
the following.  We have:
\begin{alignat}{3}
	X(j,i) & = 2^{f(i)} + \shift(j,f(i)+1) \cdot 2^{f(i)+1} \\
		   & = \bigl(1 + 2 \cdot \shift(j,f(i)+1) \bigr) \cdot  2^{f(i)} & \label{eq:proof1}\\
	       & = \shift(j + 2^{f(i)-1},f(i)) \cdot 2^{f(i)} & \quad \mbox{By
Lemma~\ref{lem:bitshift3} with $a=f(i)$.}\\
	       & = Y(j,i).  
	\label{eq:proof2}
\end{alignat}
\end{enumerate}
This completes the induction and proof.
\qed
\end{proof}

Next we show, a number of lemmas that we will use to show that Dantzig's rule follows the sequence $G$, as defined in Definition~\ref{def:gcxor}.

\begin{lemma}
\label{lem:switchable}
For $j \in \Index \setminus \{2^n-1\}$, state $i$ is switchable with appeal
\clockappeal if $\bit(j,f(i))=0$.
\end{lemma}

\begin{proof} 
Suppose that $\bit(j,f(i))=0$.  We consider two cases.
First, suppose that $\bit(g(j),i) = 0$.
Then, \eqref{def:gcxor} implies that we have:
$$
\bit(j,f(i))=\bit(j,f(i)+1)=0.
$$
Then, according to \eqref{eq:val_i_prime} and \eqref{eq:val_i}, we have:
\begin{alignat*}{3}
\valgj(i-1) & = \shift(j + 2^{f(i-1)-1},f(i-1)) \cdot 2^{f(i-1)}\\
			     & = \shift(j + 2^{f(i)},f(i)+1) \cdot 2^{f(i)+1}\\
			     & = \shift(j,f(i)+1)) \cdot 2^{f(i)+1} & \mbox{Since $\bit(j,f(i)+1)=0$.}\\
				 & < 2^{f(i)} + \shift(j,f(i)+1) \cdot 2^{f(i)+1}\\
				 & = \valgj(i').
\end{alignat*}
Thus $i$ is switchable, and the appeal is
$$
\alpha_i \cdot T \cdot \bigl(\valgj(i')-\valgj(i-1)\bigr)= \alpha_i \cdot T \cdot 2^{f(i)} = \clockappeal,
$$

Second, suppose that $\bit(g(j),i) = 1$.
Then, \eqref{def:gcxor} implies that $\bit(j,f(i)) \ne \bit(j,f(i)+1)$. 
Thus, we have:
$$
\bit(j,f(i))=0, \quad \bit(j,f(i)+1)=1.
$$
Then, according to \eqref{eq:val_i_prime} and \eqref{eq:val_i}, we have:
\begin{alignat*}{5}
\valgj(i') & = 2^{f(i)} + \shift(j,f(i)+1) \cdot 2^{f(i)+1}\\
				& < \shift(j + 2^{f(i)},f(i)+1) \cdot 2^{f(i)+1} & \quad\quad \mbox{Since $\bit(j,f(i)+1)=1$.}\\
			    & = \valgj(i-1) 
\end{alignat*}
Thus $i$ is switchable, and by Lemma~\ref{lem:argappeal}, the appeal is
$$
\alpha_i \cdot T \cdot \bigl(\valgj(i-1)-\valgj(i')\bigr)= \alpha_i \cdot T \cdot 2^{f(i)-1} = \clockappeal,
$$
as claimed.
This completes the proof. 
\qed
\end{proof}

\begin{lemma}
\label{lem:notswitchable}
For $j \in \Index \setminus \{2^n-1\}$, state $i$ is not switchable under $g(j)$ if $\bit(j,f(i)) = 1$. 
\end{lemma}

\begin{proof}
Suppose that $\bit(j,f(i))=1$.  We consider two cases.
First, suppose that $\bit(g(j),i) = 0$.
Then, \eqref{def:gcxor} implies that we have:
$$
\bit(j,f(i))=\bit(j,f(i)+1)=1.
$$
Then, according to \eqref{eq:val_i_prime} and \eqref{eq:val_i}, we have:
\begin{alignat*}{3}
\valgj(i-1) & = \shift(j + 2^{f(i-1)-1},f(i-1)) \cdot 2^{f(i-1)}\\
			     & = \shift(j + 2^{f(i)},f(i)+1) \cdot 2^{f(i)+1}\\
			     & > \bigl(1 + \shift(j,f(i)+1)) \bigr) \cdot 2^{f(i)+1} & \mbox{Since $\bit(j,f(i)+1)=1$.}\\
				 & > 2^{f(i)} + \shift(j,f(i)+1) \cdot 2^{f(i)+1}\\
				 & = \valgj(i').
\end{alignat*}
Thus $i$ is not switchable.

Second, suppose that $\bit(g(j),i) = 1$.
Then, \eqref{def:gcxor} implies that $\bit(j,f(i)) \ne \bit(j,f(i)+1)$. 
Thus, we have:
$$
\bit(j,f(i))=1, \quad \bit(j,f(i)+1)=0.
$$
Then, according to \eqref{eq:val_i_prime} and \eqref{eq:val_i}, we have:
\begin{alignat*}{5}
\valgj(i') & = 2^{f(i)} + \shift(j,f(i)+1) \cdot 2^{f(i)+1}\\
				& > \shift(j,f(i)+1) \cdot 2^{f(i)+1} & \\
				& = \shift(j + 2^{f(i)},f(i)+1) \cdot 2^{f(i)+1} & \quad\quad \mbox{Since $\bit(j,f(i)+1)=0$.}\\
			    & = \valgj(i-1) 
\end{alignat*}
Thus $i$ is not switchable.
This completes the proof. 
\qed
\end{proof}

\begin{corollary}
\label{cor:switchability}
For $j \in \Index \setminus \{2^n-1\}$, state $i$ is switchable if and only if $\bit(j,f(i))=0$,
and if $i$ is switchable, then the appeal of the switch is \clockappeal.
\end{corollary}

\begin{proof}
This follows immediately from Lemmas~\ref{lem:switchable} and~\ref{lem:notswitchable}.
\qed
\end{proof}

\begin{lemma}
\label{lem:switch_leftmost}
Dantzig's rule will always switch
the switchable state with highest index, i.e.,
$$
\max_{i \in [n]} \ \{ i \ \mid \ i \mbox{ is switchable}\}.
$$
\end{lemma}

\begin{proof}
According to Lemma~\ref{cor:switchability}, for $i \in [n]$, if $i$ is
switchable, the appeal of the other action at $i$ is \clockappeal, which is
an increasing function of $i$. This proves the claim.
\qed
\end{proof}

\begin{corollary}
\label{cor:switchk}
For $j \in \Index \setminus \{2^n-1\}$, and $i \in [n]$, let $k = f(\lsz(j))$, which is
well-defined since the binary representation of~$j$ contains at least one zero
bit. Dantzig's rule will switch action $k$ at $g(j)$. 

\end{corollary}

\begin{proof}
Because $\bit(j,\lsz(j))=0$ and all less significant bits in in $j$
are 1, according to Corollary~\ref{cor:switchability}, state $k$ is switchable,
and all states $i > k$ are not switchable under $g(j)$.
Thus, the claim follows from Lemma~\ref{lem:switch_leftmost}.
\qed
\end{proof}

Now, we can prove that Dantzig's rule will progress the clock through the sequence $G$.

\begin{lemma}
\label{lem:seqG}
Dantzig's rule, started from $g(0)=0^n$ will progress the clock through the
sequence $G$, taking $2^n-1$ iterations, and finishing at
$g(2^n-1)=0^{n-1}1$.
\end{lemma}

\begin{proof}
By Corollary~\ref{cor:switchability}, for all $j \in \Index \setminus \{2^n-1\}$, there
exists a switchable state under $g(j)$, and thus since the first $n$ bits of
$2^n-1$ are all 1, Dantzig's rule will terminate at $g(2^n-1)=0^{n-1}1$.
Lemma~\ref{lem:gcflip} shows that $g(j+1)$ is obtained from $g(j)$ by 
flipping the bit in position $f(\lsz(j))$ of $g(j)$. 
Lemma~\ref{cor:switchk} shows that at $g(j)$, Dantzig's rule
flips $f(\lsz(j))$ and thus moves to $g(j+1)$. 
This implies that Dantzig's rule takes $2^n-1$ iterations in total to find
the optimal policy of the clock.
\qed
\end{proof}

\begin{lemma}
\label{lem:clock_works}
For all $j \in \Index$ we have:
\begin{alignat}{3}
\label{eq:clock_tick1}
\valgjnobar(c_0) &= \valgjnobar(c_1) + T & \quad \text{for even}\ j \\
\label{eq:clock_tick2}
\valgjnobar(c_1) & = \valgjnobar(c_0) + T & \quad \text{for odd}\ j,
\end{alignat}
and the appeal of the switched action (and thus all actions) is in the range $[0.25,0.5]$.
\end{lemma}

\begin{proof}
According to Corollary~\ref{lem:seqG}, Dantzig's rule will go through the sequence
of policies $G$.
According to Lemma~\ref{lem:clock_values}, for $j \in \Index$, we have:
\begin{alignat*}{2}
	\valgj(c_0) & = \valgj(n) \\
				& = \shift(j + 2^{f(n)-1},f(n)) \cdot 2^{f(n)-1}.\\
				& = 2 \cdot \shift(j+1,1)\\
				& = \begin{cases} 
					j & \mbox{for even } j \\
					j+1 & \mbox {for odd } j.
					\end{cases}.\\
	\valgj(c_1) & = 1 + 2 \cdot \shift(j,1)\\
				& = \begin{cases} 
					j+1 & \mbox{for even } j \\
					j & \mbox {for odd } j.
					\end{cases}.
\end{alignat*}
This shows \eqref{eq:clock_tick1} and \eqref{eq:clock_tick2}.
The appeal of switchable actions is always given by \clockappeal for some $i
\in [n]$, which is in $[0.25,0.5]$ as required.
\qed
\end{proof}

Lemmas~\ref{lem:seqG} and~\ref{lem:clock_works} combine to complete the proof 
of Lemma~\ref{lem:clock}.

\section{Proof of Lemma~\ref{lem:pspace}}
\label{app:pspace}

\begin{proof}
We will use the same reduction in order to show that both of the problems are
PSPACE-complete. We reduce from the halting problem of a deterministic Turing
machine with polynomial tape length. More formally, let $\mathcal{T}$ be a
deterministic Turing machine with set of states $Q$, and a binary tape alphabet,
let $I \in \{0, 1\}^k$ be an initial input tape of length $k$, and let $n \in
\poly(k)$ be a bound on the amount of space used by $\mathcal{T}$ while
processing $I$. Clearly, the halting problem for $\mathcal{T}$ is
PSPACE-complete.

We construct a circuit iteration instance as follows. Firstly,
we modify $\mathcal{T}$ so that, when it accepts, it writes a $0$ to tape cell
$n+1$ and then halts. Let $\mathcal{T'}$ be this new Turing machine, and let
$Q'$ be the set of states used by $\mathcal{T'}$.
We will construct a function $F : \{0, 1\}^{n'} \rightarrow \{0, 1\}^{n'}$ where
$n' = (n+1) + \log n + \log|Q'|$. We allocate bits $1$ through $n'$ in the
following way:
\begin{itemize}
\item The first $n + 1$ bits are used to hold the contents of the tape.
\item The next $\log n$ bits are used to hold the position of the tape head.
\item The final $\log |Q'|$ bits are used to hold the current state of the
Turing machine
\end{itemize}
That is, each bit-string $B \in \{0, 1\}^{n'}$ holds all of the information
about the configuration of $\mathcal{T'}$ during its computation. For each bit-string $B
\in \{0, 1\}^{n'}$, we define $F(B)$ to be next configuration of $\mathcal{T'}$
after one transition has been made. This can clearly be computed in polynomial
time, and therefore can be represented as a binary circuit that has size
polynomial in $|\mathcal{T'}|$.

Let $B^I \in \{0, 1\}^{n'}$ be defined such that:
\begin{itemize}
\item The first $k$ bits of $B^I$ are $I$.
\item The next $n-k$ bits of $B^I$ are $0$.
\item The $n+1$th bit of $B^I$ is a $1$.
\item The next $\log n$ bits correspond to the initial tape head position of $\mathcal{T'}$.
\item The final $\log |Q|$ bits correspond to the initial state of
$\mathcal{T'}$.
\end{itemize}
Under these definitions, $B^I$ corresponds to the initial configuration of
$\mathcal{T}$ on input~$I$. 

Our circuit iteration instance is $(F, B^I, n+1)$. Note that, since
$\mathcal{T'}$ uses at most $n+1$ cells on its tape, we must have that if
$\mathcal{T'}$ halts on $I$, then it does so in at most $2^{n+1}$ transitions.
Moreover, we have $2^{n'} > 2^{n+1}$. Moreover, since the machine writes a $0$
to tape cell $n+1$ if and only if it halts,  this proves that $\circuitvalue(F,
B^I, n+1)$ is PSPACE-complete.

For the problem \bitswitch, observe that the $n+1$th bit of $B^I$ is $1$, and
that $\mathcal{T'}$ only writes a $0$ to tape cell $n+1$ in the case where it
accepts $I$. Our goal is to show that there is an even $i \le 2^{n'}$ such that
the $(n+1)$-th bit of $F^i(B^I)$ is a $0$ if and only if $\mathcal{T'}$ accepts
$I$. So, let~$t$ be the number of steps that $\mathcal{T'}$ takes before it
terminates. If~$t$ is even, then we are done. Otherwise, note that $t \le
2^{n+1} < 2^{n'}$, so $t+1$ is an even number with $t+1 \le 2^{n'}$ such that
$F^{t+1}(B^I)$ is a $0$. Therefore, we have shown that the problem
$\bitswitch(F, B^i, n+1)$ is PSPACE-complete.
\qed
\end{proof}

\section{Proof of Lemma~\ref{lem:final}}
\label{app:final}

\begin{proof}
We will provide three different proofs for three distinct cases. Firstly,
suppose that $i$ is an input bit. 
Note that the appeals given in
Lemma~\ref{lem:buffappeal} imply that one cannot move from a coherent
policy to a non coherent policy by switching an action with appeal
strictly greater than $3.5$. Therefore we must have that $\sigma'$ is also a
coherent policy. So, we can apply Lemma~\ref{lem:buffappeal} to argue
that $o^j_i$, $l^j_i$ and $r^j_i$ are final in $\sigma'$.

Now suppose that $i$ is an \org gate. Note that, by definition, all gates $i'$
with $d(i') < d(i)$ are also final. Therefore, the action switched by Dantzig's
rule cannot be contained in a gadget belonging to a gate $i'$ with $d(i') <
d(i)$. Moreover, by Lemma~\ref{lem:clock}, the action switched by Dantzig's
rule cannot be contained in the clock. 
Therefore, we have: 
\begin{align*}
\val^{\sigma}(o^j_{\inp_1(i)}) &= \val^{\sigma'}(o^j_{\inp_1(i)}) \\
\val^{\sigma}(o^j_{\inp_2(i)}) &= \val^{\sigma'}(o^j_{\inp_2(i)}) \\
\val^{\sigma}(c_j) &= \val^{\sigma'}(c_{j}).
\end{align*}
Moreover, since gate $i$ is final we must have $\sigma(v^j_i) = \sigma'(v^j_i)$,
$\sigma(x^j_i) = \sigma'(x^j_i)$, and $\sigma(o^j_i) = \sigma'(o^j_i)$.
Therefore, for every action $a \in A_{o^j_i} \cup A_{v^j_i} \cup A_{x^j_i}$, we
have $\appeal^{\sigma}(a) = \appeal^{\sigma'}(a)$. Thus, gate $i$ is final
in $\sigma'$.

Finally, suppose that $i$ is a \notg gate. This case is very similar to the \org
gate case. In particular, we can use the same argument to prove that:
\begin{align*}
\val^{\sigma}(o^j_{\inp(i)}) &= \val^{\sigma'}(o^j_{\inp_1(i)}) \\
\val^{\sigma}(c_j) &= \val^{\sigma'}(c_{j}).
\end{align*}
Since $i$ is final, we must have $\sigma(o^j_i) = \sigma'(o^j_i)$ and
$\sigma(a^j_i) = \sigma'(a^j_i)$. Therefore, for every action $a \in A_{o^j_i}
\cup A_{a^j_i}$, we have $\appeal^{\sigma}(a) = \appeal^{\sigma'}(a)$. Thus,
gate $i$ is final in $\sigma'$.
\qed
\end{proof}

\section{Proof of Lemma~\ref{lem:circuitinductive}}
\label{app:circuitinductive}

Before we proceed to prove Lemma~\ref{lem:circuitinductive}, we first prove some
basic properties about our gate gadgets. 

Let us explain some terminology that will be used throughout this section.
Suppose that we are in a policy $\sigma_0$, and let $a$ be an action with appeal
$x$. Throughout this proof we will use the phrase ``Dantzig's rule will
eventually switch action $a$ with appeal $x$,'' to imply that Dantzig's rule will move through a
sequence of policies $\sigma_0, \sigma_1, \dots, \sigma_m$ where: 
\begin{itemize} 
\item Action $a$ has appeal $x$ in $\sigma_{m-1}$, and is switched to produce
$\sigma_m$. 
\item For every $j < m-1$, we have that the action switched by Dantzig's rule in $\sigma_j$ has appeal at least $x$.
\end{itemize} 
Note that if $x > 3.5$, then due to Lemma~\ref{lem:final}, every
gate $i$ that is final in $\sigma_0$ is also final in $\sigma_m$. We need this
notation because there are very often ties for the highest appeal action: for
example, if two
\org gates have the same depth, and the same action needs to be switched in both
of them, then both actions will have the same appeal, and we do not care which
one is switched first. This notation allows us to say that both will be
switched, without caring about the order in which they are switched. Moreover,
$\org$ states in gates with depth strictly greater than $k+1$ may also switch
while we are processing the gates with depth $k$. Since these are irrelevant for
the depth $k$ proof, this notation allows us to ignore these switches.

\subsection{\org gates}

We begin with the \org gate gadget. Let $i$ be an \org gate. The following lemma
describes the conditions under which $v^j_i$ wants to switch. Specifically, it
shows that if both $\inp_1(i)$ and $\inp_2(i)$ are $B$-correct and final, and if
$v^j_i$ does not currently select the maximum of the two, then switching $v^j_i$
must have high appeal.

\begin{lemma}
\label{lem:orv}
Suppose that we are in phase $j$, let $\sigma$ be a coherent policy, and let
$i$ be an \org gate. Suppose that $\sigma(v^j_i) =
o^j_{\inp_m(i)}$ for $m \in \{1, 2\}$, and let $\bar{m}$ be the other input. If
$\val^{\sigma}(o^j_{\inp_m(i)}) = L_{d(i) -1}$, and 
$\val^{\sigma}(o^j_{\inp_{\bar{m}}(i)}) = H_{d(i) -1}$
then the appeal of switching $v^j_i$ to $o^j_{\inp_m(i)}$ is at least $27$.
\end{lemma}
\begin{proof}
We have that the appeal of switching $v^j_i$ to $o^j_{\inp_m(i)}$ is: 
\begin{equation*}
H_{d(i) - 1} - L_{d(i) -1} = b_{d(i) - 1}.
\end{equation*}
Recall that $b_k$ decreases as $k$ increases, therefore we have $b_{d(i) - 1}
\ge b_{d(c) - 1} = 3^3 = 27$. \qed
\end{proof}

The following lemma describes the conditions under which $o^j_i$ wants to
switch. Specifically, if at least one of the two inputs to the gate is $1$, then
$o^j_i$ can switch to $v^j_i$ with high appeal. On the other hand, if both
inputs are $0$, then $o^j_i$ can switch to $x^j_i$ with high appeal.

\begin{lemma}
\label{lem:oro}
Suppose that we are in phase $j$, let $\sigma$ be a coherent policy, and let
$i$ be an \org gate.
\begin{enumerate}
\item If $\val^{\sigma}(v^j_i) = \val^{\sigma}(c_j) + H_{d(i)-1}$ then:
\begin{enumerate}
\item If $\sigma(o^j_i) = x^j_i$ then the appeal of switching $o^j_i$ to $v^j_i$
is at least $9$.
\item If $\sigma(o^j_i) = v^j_i$ then the appeal of switching $o^j_i$ to $x^j_i$
is strictly less than $0$. 
\end{enumerate}
\item If $\val^{\sigma}(v^j_i) = \val^{\sigma}(c_j) + L_{d(i)-1}$ then:
\begin{enumerate}
\item If $\sigma(o^j_i) = x^j_i$ then the appeal of switching $o^j_i$ to $v^j_i$
is strictly less than $0$. 
\item If $\sigma(o^j_i) = v^j_i$ then the appeal of switching $o^j_i$ to $x^j_i$
is at least $18$.
\end{enumerate}
\end{enumerate}
\end{lemma}
\begin{proof}
For case 1, if $\sigma(o^j_i) = x^j_i$, then the appeal of switching $o^j_i$
to $v^j_i$ is $H_{d(i)-1} + b_{d(i)} - L_{d(i)} = b_{d(i)}$. Recall that $b_k$
decreases as $k$ increases, thus we have that $b_{d(i)} \geq b_{d(C)} = 9$. On
the other hand, if $\sigma(o^j_i) = v^j_i$, then by the converse of the previous
argument, we have that the appeal of switching $o^j_i$
to $v^j_i$ is $L_{d(i)} - H_{d(i)-1} - b_{d(i)} = b_{d(i)} < 0$, since
$b_{d(i)}$ is always positive.

For case 2, if $\sigma(o^j_i) = x^j_i$, then the appeal of switching $o^j_i$ to
$v^j_i$ is 
\begin{align*}
L_{d(i) - 1} + b_{d(i)} - L_{d(i)} &= b_{d(i)} - b_{d(i)-1} \\
&= 3^{d(C) - d(i) + 2} - 3^{d(C) - d(i) + 3} \\
&< 0.
\end{align*}
On the other hand, if $\sigma(o^j_i) = v^j_i$, then the appeal of switching
$o^j_i$ to $x^j_i$ is:
\begin{align*}
L_{d(i)}  - 
L_{d(i) - 1} - b_{d(i)} -  &=  b_{d(i)-1} - b_{d(i)} \\
&= 3^{d(C) - d(i) + 3} - 3^{d(C) - d(i) + 2} \\
&\geq 3^{d(C) - d(C) + 3} - 3^{d(C) - d(C) + 2} \\
&= 18.
\end{align*}
\qed
\end{proof}

In our final lemma concerning \org gates, we show that the gate always outputs the
correct value. Taking into account Lemmas~\ref{lem:orv} and~\ref{lem:oro}, we
can see that there are only two cases to consider: 
\begin{itemize}
\item If both inputs are $0$
then $o^j_i$ takes the action to $x^j_i$.
\item If at least one input is $1$, then $o^j_i$ takes the action towards
$v^j_i$, and $v^j_i$ takes the action towards the input bit that is $1$.
\end{itemize}
The following lemma shows that in either case, the value of $o^j_i$ is $B$-correct.


\begin{lemma}
\label{lem:orcorrect}
Suppose that we are in phase $j$, let $\sigma$ be a coherent policy, and let
$i$ be an \org gate.
\begin{enumerate}
\item If $\val^{\sigma}(v^j_i) = \val^{\sigma}(c_j) + H_{d(i)-1}$ and
$\sigma(o^j_i) = v^j_i$, then $\val^{\sigma}(o^j_i) = \val^{\sigma}(c_j) +
H_{d(i)}$.
\item If $\sigma(o^j_i) = x^j_i$, then $\val^{\sigma}(o^j_i) = \val^{\sigma}(c_j) + L_{d(i)}$
\end{enumerate}
\end{lemma}
\begin{proof}
The first part follows from the fact that $H_{d(i)-1} + b_{d(i)} = H_{d(i)}$.
The second part follows from the reward of $L_{d(i)}$ on the action between
$o^j_i$ and $x^j_i$. \qed
\end{proof}

\subsection{Not gates}

We now consider the \notg gate gadgets. Let $i$ be a \notg gate. There are two modes
of operation of the \notg gate depending upon the action chosen at $a^j_i$. In the
following lemma, we describe the behaviour of the \notg gate in the case where
$a^j_i$ chooses the action towards $c_j$. Specifically, we show that if
$\inp(i)$ is final and $B$-correct, then the action from $o^j_i$ to $o^j_{\inp(i)}$
has high appeal. Moreover, we give a precise value for the appeal of switching
$a^j_i$ to $c_{1-j}$.

\begin{lemma}
\label{lem:notaj}
Suppose that we are in phase $j$, let $\sigma$ be a $j$ coherent policy, and let
$i$ be a \notg gate. If $\sigma(a^j_i) = c_j$ then:
\begin{enumerate} 
\item The appeal of switching $a^j_i$ to $c_{1-j}$ is $\ajprime$.
\item If $\sigma(o^j_i) = a^j_i$ and $\val^{\sigma}(o^j_{\inp(i)}) \ge \val^{\sigma}(c_j) + L_{d(i) - 1}$,
then the appeal of switching $o^j_i$ to $o^j_{\inp(i)}$ is at least $8$.
\item If $\sigma(o^j_i) = o^j_{\inp(i)}$ and $\val^{\sigma}(o^j_{\inp(i)}) \ge \val^{\sigma}(c_j) + L_{d(i) - 1}$,
then the appeal of switching $o^j_i$ to $a^j_{i}$ is strictly less than $0$.
\end{enumerate}
\end{lemma}
\begin{proof}
We begin with the first claim. Since $\val^{\sigma}(c_{1-j}) =
\val^{\sigma}(c_j) + T$, we can apply Lemma~\ref{lem:argappeal} to argue that
the appeal of switching $a^j_i$ to $c_{1-j}$ is:
\begin{equation*}
p_{1} \cdot H_{d(i)-1}  = \ajprime. 
\end{equation*}

For the second claim, observe that $\val^{\sigma}(o^j_i) = \val^{\sigma}(c_j) +
b_{d(i)}$. Recall that, by assumption, for every \notg gate $i$, we have $d(i) \ge
2$. Note also that $L_k$ increases as $k$ increases. Therefore, we have
$\val^{\sigma}(o^j_{\inp(i)}) \ge \val^{\sigma}(c_j) + L_{1}$.
Hence, the appeal of switching 
$o^j_i$ to $o^j_{\inp(i)}$ is at least:
\begin{align*}
L_1 - b_{d(i)} &\ge L_1 - b_2 \\
&= 3^{d(C) + 2} - 3^{d(C)} \\
&\ge 3^{2} - 3^{0} \\
&= 8.
\end{align*}

For the third claim, we can apply Lemma~\ref{lem:argappeal} to argue that the
appeal of switching $o^j_i$ to $a^j_{i}$ is:
\begin{equation*}
\frac{1}{b_{d(i)}} \cdot (1 - L_{d(i) - 1}) \le 
\frac{1}{b_{d(i)}} \cdot (1 - L_{1}) < 0.
\end{equation*}
\qed
\end{proof}

In the next lemma, we describe the behaviour of the \notg gate when $a^j_i$
chooses the action towards $c_{1-j}$. In particular, we show that if $\inp(i)$
is final and $B$-correct, then the appeal of switching $o^j_i$ to $a^j_i$ is $1$ in
the case where $C(B, \inp(i))$ is $1$, and $4$ in the case where $C(B, \inp(i))$
is~$0$. This is the critical property that we need: while it is processing the
circuit, Dantzig's rule will only switch actions with appeal greater than $3.5$,
so it will switch in the case where $C(B, \inp(i)) = 0$, and it will not switch
when $C(B, \inp(i)) = 1$.

\begin{lemma}
\label{lem:nota1j}
Suppose that we are in phase $j$, let $\sigma$ be a coherent policy, and
let~$i$ be a \notg gate. If $\sigma(a^j_i) = c_{1-j}$ then:
\begin{enumerate}
\item The appeal of switching $a^j_i$ to $c_j$ is strictly less than $0$.
\item If $\sigma(o^j_i) = o^j_{\inp(i)}$ then:
\begin{enumerate}
\item If $\val^{\sigma}(o^j_{\inp(i)}) = L_{d(i)-1}$, then the appeal of
switching $o^j_i$ to $a^j_i$ is $4$.
\item If $\val^{\sigma}(o^j_{\inp(i)}) = H_{d(i)-1}$, then the appeal of
switching $o^j_i$ to $a^j_i$ is $1$.
\end{enumerate}
\item If $\sigma(o^j_i) = a^j_i$, and if $\val^{\sigma}(o^j_{\inp(i)}) \le
H_{d(i)-1}$, then the appeal of switching $o^j_i$ to $o^j_{\inp(i)}$ is
strictly less than $0$.
\end{enumerate}
\end{lemma}
\begin{proof}
We begin with the first claim. Since we are in phase $j$, we have
$\val^{\sigma}(c_{1-j}) = \val^{\sigma}(c_j) + T$. Thus, we can apply
Lemma~\ref{lem:argappeal} to argue that the appeal of switching 
$a^j_i$ to $c_j$ is:
\begin{equation*}
p_2 \cdot (0 - H_{d(i)}) < 0.
\end{equation*}

We now consider the second claim. First suppose that 
$\val^{\sigma}(o^j_{\inp(i)}) = L_{d(i) -1 }$. Then, by
Lemma~\ref{lem:argappeal} we have that the appeal of switching $o^j_i$ to
$a^j_i$ is:
\begin{align*}
& \frac{1}{b_{d(i)}} \cdot (H_{d(i) - 1} - L_{d(i) - 1}) + 1 \\
&= \frac{1}{b_{d(i)}} \cdot b_{d(i) - 1} + 1 \\
&= \frac{1}{3^{d(C) - d(i) + 2}} \cdot 3^{d(C) - d(i) + 3} + 1 \\
&= 4.
\end{align*}
On the other hand, if $\val^{\sigma}(o^j_{\inp(i)}) = H_{d(i) -1 }$, then the
appeal of switching $o^j_i$ to $a^j_i$ is:
\begin{equation*}
\frac{1}{b_{d(i)}} \cdot (H_{d(i) - 1} - H_{d(i) -1}) + 1 = 1.
\end{equation*}
Thus, we have shown both parts of the second claim.

Finally, we consider the third claim of this lemma. By Lemma~\ref{lem:argappeal}
we have that the appeal of switching $o^j_i$ to $o^j_{\inp(i)}$ is at most:
\begin{equation*}
H_{d(i) - 1} - (H_{d(i) -1} + b_{d(i)}) < 0.
\end{equation*}
\qed
\end{proof}

Finally, we show that the value of $o^j_i$ is $B$-correct, once the necessary
actions have been switched.

\begin{lemma}
\label{lem:notcorrect}
Suppose that we are in phase $j$, and let $\sigma$ be a coherent policy. We
have: 
\begin{itemize}
\item If $\sigma(o^j_i) = o^j_{I(i)}$ and $\val^{\sigma}(o^j_{I(i)}) = H_{d(i)-1}$, then $\val^{\sigma}(o^j_i) = L_{d(i)}$.
\item If $\sigma(o^j_i) = a^j_{i}$,  then $\val^{\sigma}(o^j_i) = H_{d(i)}$.
\end{itemize}
\end{lemma}
\begin{proof}
In the case where $\sigma(o^j_i) = o^j_{I(i)}$ the claim follows easily:
\begin{equation*}
\val^{\sigma}(o^j_i) = \val^{\sigma}(O(i)) = H_{d(i)-1} = L_{d(i)}.
\end{equation*}
For the case where $\sigma(o^j_i) = a^j_i$, we can apply Lemma~\ref{lem:argappeal} to obtain:
\begin{equation*}
\val^{\sigma}(o^j_i) = H_{d(i)-1} + b_{d(i)} = H_{d(i)}.
\end{equation*}
\qed
\end{proof}

\subsection{The proof}

We can now present the proof of Lemma~\ref{lem:circuitinductive}.

\begin{proof}
Let $i$ be a gate with depth $k+1$.
Since $k > 0$, we have that $i$ cannot be an input bit,
so we have two cases to consider. Firstly, suppose that $i$ is an \org gate.
The following sequence of events occurs:
\begin{itemize}
\item Firstly, if $\val^{\sigma'}(o^j_{\inp_1(i)}) \ne
\val^{\sigma'}(o^j_{\inp_2(i)})$, and if $\val^{\sigma'}(v^j_i) \ne
\val^{\sigma'}(c_j) + H_{d(i) - 1}$, then by Lemma~\ref{lem:orv} Dantzig's rule
will eventually switch $v^j_i$ with appeal at least $27$. Once this has occurred
Lemma~\ref{lem:orv} implies that $v^j_i$ is final.
\item At this point, if $o^j_i$ does not give the correct output, then by
Lemma~\ref{lem:oro} Dantzig's rule will eventually switch $o^j_i$ with appeal
$9$. Once this has occurred, Lemma~\ref{lem:oro} implies that $o^j_i$ is final.
\item Once the above two steps have occurred, Lemma~\ref{lem:orcorrect} implies
that gate $i$ is $B$-correct.
\end{itemize}
Thus, Dantzig's rule will eventually arrive at a policy $\sigma''$ in which
gate $i$ is both final and $B$-correct, while switching only actions of appeal
strictly greater than $9$.

We now prove the case for when $i$ is a \notg gate. Recall that, by assumption, as
Dantzig's rule moved from $\sigma$ to $\sigma'$ it only switched actions with
appeal at least $3.5 + \frac{1}{2\cdot d(i-1)} > 3.5 + \frac{1}{2 \cdot d(i)}$.
Therefore, by Lemma~\ref{lem:buffappeal} we must still have $\sigma'(a^j_i) =
c_j$. So, the following sequence of events occurs:
\begin{itemize}
\item Firstly, by Lemma~\ref{lem:notaj}, we have that if $\sigma(o^j_i) \ne
o^j_{\inp(i)}$, then Dantzig's rule will eventually switch $o^j_i$ to 
$o^j_{\inp(i)}$ with appeal at least $8$.
\item Then, by Lemma~\ref{lem:notaj}, we have that Dantzig's rule will eventually
switch $a^j_i$ to $c_{1-j}$ with appeal $3.5 + \frac{1}{2 \cdot d(i)}$.
\item There are now two cases:
\begin{itemize}
\item If $\val^{\sigma}(o^j_{\inp(i)}) = \val^{\sigma}(c_j) + L_{d(i)-1}$, then
by Lemma~\ref{lem:nota1j} we have that Dantzig's rule will eventually switch
$o^j_i$ to $a^j_i$ with appeal $4$. Once this has occurred,
Lemma~\ref{lem:nota1j} implies that gate $i$ is final, and
Lemma~\ref{lem:notcorrect} implies that gate $i$ is $B$-correct.
\item If $\val^{\sigma}(o^j_{\inp(i)}) = \val^{\sigma}(c_j) + H_{d(i)-1}$, then
nothing further is switched in gate $i$, because Lemma~\ref{lem:nota1j} implies
that gate $i$ is final, and Lemma~\ref{lem:notcorrect} implies that gate $i$ is
$B$-correct.
\end{itemize}
\end{itemize}
Thus, Dantzig's rule will eventually arrive at a policy $\sigma''$ in which
gate $i$ is both final and $B$-correct, while switching only actions of appeal
strictly greater than $3.5 + \frac{1}{2 \cdot d(i)}$.

Since the argument made above holds for every gate $i$ with depth $k+1$, we have
that Dantzig's rule will eventually move to a policy $\sigma'''$ in which all
gates with depth $k+1$ are both final and $B$-correct, while switching only
actions of appeal strictly greater than $3.5 + \frac{1}{2 \cdot (k+1)}$.
\qed
\end{proof}

\section{Proof of Lemma~\ref{lem:circuit}}
\label{app:circuit}

Before we begin, we give two lemmas about the behaviour of the input bit
gadgets.

\subsection{Input bits}

Recall that, in an initial policy $\sigma$ for $B$, we have that
$\sigma(o^{1-j}_i) = l^{1-j}_i$. Our goal is to show that $o^{1-j}_i$ switches
to $r^{1-j}_i$ only in the case where $C(B, \inp(i)) = 1$. 
The following lemma characterises the conditions under which $o^{1-j}_i$ wants
to switch to $r^{1-j}_i$.

\begin{lemma}
\label{lem:seven}
Suppose that we are in phase $J$, and let $\sigma$ be a policy with
$\sigma(o^{1-j}_i) = l^{1-j}_i$, $\sigma(l^{1-j}_i) = c_j$, $\sigma(r^{1-j}_i) =
o^j_{\inp(i)}$ for some input bit $i$. We have:
\begin{itemize}
\item If $\val^{\sigma}(o^j_{\inp(i)}) \le \val^{\sigma}(c_j) + L_{d(C)}$, then
the appeal of switching $o^{1-j}_i$ to $r^{1-j}_i$ is strictly less than~$0$.
\item If $\val^{\sigma}(o^j_{\inp(i)}) = \val^{\sigma}(c_j) + H_{d(C)}$, then
the appeal of switching $o^{1-j}_i$ to $r^{1-j}_i$ is $4.5$.
\end{itemize}
\end{lemma}
\begin{proof}
We have:
\begin{equation*}
\val^{\sigma}(o^{1-j}_i) = -\frac{T}{2} + \frac{H_{d(C)} + L_{d(C)}}{2}.
\end{equation*}
On the other hand, we have:
\begin{equation*}
\val^{\sigma}(r^{1-j}_i) = -\frac{T}{2} + \val^{\sigma}(o^{j}_{\inp(i)}).
\end{equation*}
So, first suppose that
$\val^{\sigma}(o^j_{\inp(i)}) \le \val^{\sigma}(c_j) + L_{d(C)}$. In this case,
the appeal of switching $o^{1-j}_i$ to $r^{1-j}_i$ can be at most:
\begin{align*}
L_{d(C)} - \frac{H_{d(C)} + L_{d(C)}}{2} &\le 
L_{d(C)} - \frac{L_{d(C)} + L_{d(C)}}{2} \\
&= 0.
\end{align*}
Now suppose that we have
$\val^{\sigma}(o^j_{\inp(i)}) = \val^{\sigma}(c_j) + H_{d(C)}$.
In this case, the appeal of switching $o^{1-j}_i$ to $r^{1-j}_i$ is:
\begin{align*}
H_{d(C)} - \frac{H_{d(C)} + L_{d(C)}}{2} &= \frac{H_{d(C)} - L_{d(C)}}{2} \\
&= \frac{b_{d(C)}}{2} \\
&= \frac{3^{2}}{2} \\
&= 4.5
\end{align*}
\qed
\end{proof}

We also need the fact that, if $o^{1-j}_i$ switches to $r^{1-j}_i$, then it
cannot be switched back to $l^{1-j}_i$. The following lemma proves this.

\begin{lemma}
\label{lem:eighter}
Suppose that we are in phase $j$, and let $\sigma$ be a policy with
$\sigma(o^{1-j}_i) = r^{1-j}_i$, $\sigma(l^{1-j}_i) = c_j$, and
$\sigma(r^{1-j}_i) = o^j_{\inp(i)}$ for some input bit $i$.  If
$\val^{\sigma}(o^j_{\inp(i)}) = \val^{\sigma}(c_j) + H_{d(C)}$, then the appeal
of switching $o^{1-j}_i$ to $l^{1-j}_i$ is strictly less than $0$.
\end{lemma}
\begin{proof}
We have:
\begin{equation*}
\val^{\sigma}(o^{1-j}_i) = \val^{\sigma}(c_j) - \frac{T}{2} + \frac{H_{d(C)} +
L_{d(C)}}{2}.
\end{equation*}
On the other hand, by assumption we have:
\begin{equation*}
\val^{\sigma}(r^{1-j}_i) = \val^{\sigma}(c_j) - \frac{T}{2} + H_{d(C)}.
\end{equation*}
Therefore, by Lemma~\ref{lem:argappeal}, the appeal of switching $o^{1-j}_i$ to
$l^{1-j}_i$ is at most:
\begin{align*}
p_3 \cdot (\frac{H_{d(C)} + L_{d(C)}}{2} - H_{d(C)})
&\le p_3 \cdot (\frac{H_{d(C)} + H_{d(C)}}{2} - H_{d(C)}) \\
&=0.
\end{align*}
\qed
\end{proof}

\subsection{The proof}

We can now give the proof of Lemma~\ref{lem:circuit}.

\begin{proof}
By Lemma~\ref{lem:circuitinductive}, we have that Dantzig's rule will eventually
switch to a policy $\sigma''$ in which, for every input bit $i$, we have that
gate $\inp(i)$ is final and $B$-correct. 

Recall that circuit $C$ is the negated form of the circuit implementing $F$.
Thus, for each input bit $i$, we have that $C(B, \inp(i)) = 1$ if and only if
the $i$-th bit of $F(B)$ is $0$.

Now we can apply Lemmas~\ref{lem:seven} and~\ref{lem:eighter}. Let $i$ be an
input bit. We consider two cases:
\begin{itemize}
\item Firstly, if $C(B, \inp(i)) = 1$, then we must have
$\val^{\sigma''}(o_{\inp(i)}) = \val^{\sigma''}(c_j) + H_{d(C)}$. So, if
$\sigma''(o^{1-j}_i) = l^{1-j}_i$, we can apply Lemma~\ref{lem:seven} implies
that Dantzig's rule will eventually switch $o^{1-j}_i$ to $r^{1-j}_i$ with appeal
$4.5$. Once we have $\sigma(o^{1-j}_i)$ we can apply Lemma~\ref{lem:eighter} to
argue that $o^{1-j}_i$ will not be switched away from $r^{1-j}_i$.
\item Secondly, if $C(B, \inp(i)) = 0$, then we must have 
$\val^{\sigma''}(o_{\inp(i)}) = \val^{\sigma''}(c_j) + L_{d(C)}$. Note that 
$\val^{\sigma''}(c_j) = \val^{\sigma}(c_j)$, and therefore, since policy
iteration can never decrease the value of a state, we must have
$\val^{\sigma'''}(o^j_{\inp(i)}) \le \val^{\sigma'''}(c_j) + L_{d(C)}$ for all
policies $\sigma'''$ that Dantzig's rule passed through as it moved from $\sigma$
to $\sigma''$. Thus, we can apply Lemma~\ref{lem:seven} to argue that Dantzig's
rule cannot have ever switched $o^{1-j}_i$ to $r^{1-j}_i$. \qed
\end{itemize}
\end{proof}

\section{Upper and lower bounds on circuit values}
\label{app:upper}

In this section, we show several preliminary results concerning upper and lower
bounds of the values in our circuit gadgets. These will be used frequently in
subsequent proofs.

The following lemma gives upper bounds on the value of the gates, based on the
value of the input bits.

\begin{lemma}
\label{lem:subupper}
Suppose that we are in phase $j$. Let $\sigma$ be a policy, and let $l \in \{0,
1\}$. Furthermore, suppose that $\sigma$ satisfies the following properties:
\begin{itemize}
\item For every \org gate $i'$, we have $\sigma(x^{j}_{i'}) = c_j$ and we have
$\sigma(x^{1-j}_{i'}) = c_j$.
\item For every \notg gate $i$, we have $\sigma(a^{1-j}_{i'}) = c_j$.
\end{itemize}
If, for every input bit $i'$, we have $\val^{\sigma}(o^l_{i'}) \le
\val^{\sigma}(c_j) + B + H_0$, for some non-negative constant $B$, then for all gates
$i$ we have $\val^{\sigma}(o^l_i) \le \val^{\sigma}(c_j) + B + H_{d(i)}$.
\end{lemma}
\begin{proof}
We will prove this claim by induction over depth. For input bits, which are the
gates with depth $0$, the claim holds by assumption.

We prove two versions of the inductive step. First, suppose that the claim has
been shown for all gates with depth $k$, and let $i$ be an \org gate with depth $k
+ 1$. If $\sigma(o^l_i) = x^l_i$ then by our assumptions on $\sigma$ we have:
\begin{align*}
\val^{\sigma}(o^l_i) &= \val^{\sigma}(c_j) + L_{k+1} \\
&\le \val^{\sigma}(c_j) + H_{k+1} \\
&\le \val^{\sigma}(c_j) + B + H_{k+1}. 
\end{align*}
On the other hand, if $\sigma(o^l_i) = c_j$ then by the inductive hypothesis we
have:
\begin{align*}
\val^{\sigma}(o^l_i) &\le \val^{\sigma}(c_j) + B + H_{k} + b_{k+1}\\
&= \val^{\sigma}(c_j) + B + H_{k+1}.
\end{align*}
This completes the \org gate case of the inductive step.

For the second version of the inductive step, suppose that the claim has been
shown for all gates with depth $k$, and let $i$ be a \notg gate with depth $k+1$.
Note that since we are in phase $j$, no matter which action is chosen by $a^j_i$, we must have
$\val^{\sigma}(a^j_i) \le \val^{\sigma}(c_j) + H_k$. On the other hand, since by
assumption we have $\sigma(a^{1-j}_i) = c_j$, we have 
$\val^{\sigma}(a^{1-j}_i) \le \val^{\sigma}(c_j)$. 
Therefore, if
$\sigma(o^l_i) = a^l_i$, then by Lemma~\ref{lem:argappeal}, we have:
\begin{align*}
\val^{\sigma}(o^l_i) &\le \val^{\sigma}(c_j) + H_{k} + b_{k+1}\\
&= \val^{\sigma}(c_j) + H_{k+1}\\
&\le \val^{\sigma}(c_j) + B + H_{k+1}.
\end{align*}
On the other hand, if $\sigma(o^l_i) = o^l_{\inp(i)}$, then by the inductive
hypothesis we have:
\begin{align*}
\val^{\sigma}(o^l_i) &\le \val^{\sigma}(c_j) + B + H_k \\
&< \val^{\sigma}(c_j) + B + H_{k+1}.
\end{align*}
\qed
\end{proof}

We now prove an upper bound for the gates in circuit $j$, for a certain class of
policies.
\begin{lemma}
\label{lem:gateupperj}
Suppose that we are in phase $j$, and let $\sigma$ be a policy where:
\begin{itemize}
\item For every input bit $i'$ we have $\sigma(l^j_{i'}) = \sigma(r^j_{i'}) =
c_j$. 
\item For every \org gate $i'$, we have $\sigma(x^{j}_{i'}) = c_j$.
\item For every \notg gate $i'$, we have $\sigma(a^{1-j}_{i'}) = c_j$.
\end{itemize}
For every gate $i$ we have:
\begin{equation*}
\val^{\sigma}(o^j_i) \le \val^{\sigma}(c_j) + H_{d(i)}.
\end{equation*}
\end{lemma}
\begin{proof}
Due to our assumptions about $\sigma$, we have $\val^{\sigma}(o^j_i) \le H_0$ for every input bit $i$. Thus, we can
apply Lemma~\ref{lem:subupper} with $B=0$ in order to prove this claim.
\qed
\end{proof}

The following lemma gives an upper bound for the gates in circuit $1-j$ in
coherent policies.
\begin{lemma}
\label{lem:gateupper}
Suppose that we are in phase $j$, and let $\sigma$ be a coherent policy. For
every gate $i$, we have
$\val^{\sigma}(o^{1-j}_i) \le \val^{\sigma}(c_j) + H_{d(i)}$.
\end{lemma}
\begin{proof}
Let $i$ be an input bit.
If $\sigma(o^{1-j}_i) = l^{1-j}_i$ we have 
\begin{equation*}
\val^{\sigma}(o^{1-j}_i) \le \val^{\sigma}(c_j) - \frac{T}{2} + \frac{H_{d(C)} +
L_{d(C)}}{2} < \val^{\sigma}(c_j).
\end{equation*}
On the other hand, if $\sigma(o^{1-j}_i) = r^{1-j}_i$, then 
we can apply Lemma~\ref{lem:gateupperj} to argue that:
\begin{equation*}
\val^{\sigma}(o^{j}_{\inp(i)}) \le \val^{\sigma}(c_j) + H_{d(C)}.
\end{equation*}
Thus, we have:
\begin{equation*}
\val^{\sigma}(o^{1-j}_i) \le \val^{\sigma}(c_j) + H_{d(C)} - \frac{T}{2} <
\val^{\sigma}(c_j).
\end{equation*}
Therefore, we have $\val^{\sigma}(o^{1-j}_i) \le \val^{\sigma}(c_j)$, and we can
again apply Lemma~\ref{lem:subupper} with $B=0$ to prove this lemma. 
\qed
\end{proof}

We now prove lower bounds for the gates in circuit $j$ in policies where
$\sigma(l^j_i) = c_j$ and $\sigma(r^j_i) = c_j$. 
\begin{lemma}
\label{lem:gatelower}
Suppose that we are in phase $j$, and let $\sigma$ be a policy with
$\sigma(l^j_i) = c_j$ and $\sigma(r^j_i) = c_j$ for every input bit $i$. For
every gate $i$ we have $\val^{\sigma}(o^j_i) \ge \val^{\sigma}(c_j)$.
\end{lemma}
\begin{proof}
We will proceed via induction over gate depth. For the base case, we consider
the input bits, which are the gates with depth $0$. For these gates, we can use
the assumption that $\sigma(l^j_i) = \sigma(r^j_i) = c_j$ to argue that either
$\val^{\sigma}(o^j_i) = \val^{\sigma}(c_j) + L_0$ or $\val^{\sigma}(o^j_i) =
\val^{\sigma}(c_j) + H_0$ must hold. Therefore in either case the lemma holds.

We prove two versions of the inductive step. First, suppose that the claim has
been shown for all gates with depth $k$, and let $i$ be an \org gate with depth $k
+ 1$. If $\sigma(o^j_i) = x^j_j$ then no matter which action is chosen at
$x^j_i$, we have:
\begin{align*}
\val^{\sigma}(o^j_i) \ge \val^{\sigma}(c_j) + L_{k+1} \ge \val^{\sigma}(c_j).
\end{align*}
On the other hand, if $\sigma(o^j_i) = c_j$ then by the inductive hypothesis we
have:
\begin{equation*}
\val^{\sigma}(o^j_i) \ge \val^{\sigma}(c_j) + b_{k+1} \ge \val^{\sigma}(c_j). 
\end{equation*}

For the second version of the inductive step, suppose that the claim has been
shown for all gates with depth $k$, and let $i$ be a \notg gate with depth $k+1$.
Note that since we are in phase $j$, no matter which action is chosen by $a^j_i$, we must have
$\val^{\sigma}(a^j_i) \ge \val^{\sigma}(c_j)$. Therefore, if
$\sigma(o^j_i) = a^j_i$, then by Lemma~\ref{lem:argappeal}, we have:
\begin{equation*}
\val^{\sigma}(o^j_i) \ge \val^{\sigma}(c_j) + b_{k+1} \ge
\val^{\sigma}(c_j). 
\end{equation*}
On the other hand, if $\sigma(o^j_i) = o^j_{\inp(i)}$, then by the inductive
hypothesis we have $\val^{\sigma}(o^j_i) \ge \val^{\sigma}(c_j)$. This
completes the proof of Lemma~\ref{lem:gatelower}.
\qed
\end{proof}



The following lemma shows upper bound on the value of the gates in circuit $1-j$
in the case where $\sigma(l^{1-j}_i) \ne c_j$ or $\sigma(r^{1-j}_i) \ne c_j$.

\begin{lemma}
\label{lem:gateupperjprime}
Suppose that we are in phase $j$, and let $\sigma$ be a policy with
\begin{itemize}
\item For every input bit $i'$ we have $\sigma(l^j_{i'}) = \sigma(r^j_{i'}) =
c_{j}$. 
\item For every \org gate $i'$, we have $\sigma(x^{j}_{i'}) = c_j$ and we have
$\sigma(x^{1-j}_{i'} = c_j$.
\item For every \notg gate $i'$, we have $\sigma(a^{1-j}_{i'}) = c_j$.
\end{itemize}
For
every gate $i$ we have
\begin{equation*}
\val^{\sigma}(o^{1-j}_i) \le \val^{\sigma}(c_{1-j}) + H_{d(i)}.
\end{equation*}
\end{lemma}
\begin{proof}
Let $i$ be an input bit. First observe that if $\sigma(l^{1-j}_i) = c_{j}$, then
we have:
\begin{align*}
\val^{\sigma}(l^{1-j}_i) 
&= \val^{\sigma}(c_{1-j}) -\frac{T}{2} + \frac{H_{d(C)} + L_{d(C)}}{2} \\
&\le \val^{\sigma}(c_{1-j}).
\end{align*}
On the other hand, if we have 
$\sigma(l^{1-j}_i) = c_{1-j}$, then we have
\begin{equation*}
\val^{\sigma}(l^{1-j}_i) = \val^{\sigma}(c_j) + H_0. 
\end{equation*}
Thus, in either case, we have $\val^{\sigma}(l^{1-j}_i) \le 
\val^{\sigma}(c_{1-j}) + H_0$.
So, if $\sigma(o^{1-j}_i) = l^j_i$
then we have:
\begin{align*}
\val^{\sigma}(o^{1-j}_i) 
\le \val^{\sigma}(c_{1-j}) + H_0. 
\end{align*}

On the other hand, if 
$\sigma(o^{1-j}_i) = r^{1-j}_i$, then we again have two cases to consider. Firstly,
if $\sigma(r^{1-j}) = c_{1-j}$, then we have:
\begin{equation*}
\val^{\sigma}(o^{1-j}_i) = \val^{\sigma}(c_{1-j}) + L_0 <
\val^{\sigma}(c_{1-j}) + H_0.
\end{equation*}
If $\sigma(r^{1-j}) = o^j_{\inp(i)}$, then 
by Lemma~\ref{lem:gateupperj} we have:
\begin{align*}
\val^{\sigma}(o^{1-j}_i) 
&\le \val^{\sigma}(c_{j}) + H_{d(C)} - \frac{T}{2} \\
&\le \val^{\sigma}(c_{j}) \\
&\le \val^{\sigma}(c_{1-j}) \\
&\le \val^{\sigma}(c_{1-j}) + H_0.
\end{align*}
Thus, in either case, we can apply Lemma~\ref{lem:subupper} with $B=0$ to
complete the proof of this lemma.
\qed
\end{proof}

The following lemma is a version of Lemma~\ref{lem:subupper} that has 
no preconditions on the choice made at $x^j_i$ or $a^j_i$.
\begin{lemma}
\label{lem:subupper2}
Suppose that we are in phase $j$. Let $\sigma$ be a policy, and let $l \in \{0,
1\}$. Furthermore, suppose that 
for every input bit $i'$, we have $\val^{\sigma}(o^l_{i'}) \le
\val^{\sigma}(c_{1-j}) + B + H_0$, for some non-negative constant $B$. For all gates
$i$ we have $\val^{\sigma}(o^l_i) \le \val^{\sigma}(c_{1-j}) + B + H_{d(i)}$.
\end{lemma}
\begin{proof}
We will prove this claim by induction over depth. For input bits, which are the
gates with depth $0$, the claim holds by assumption.

We prove two versions of the inductive step. First, suppose that the claim has
been shown for all gates with depth $k$, and let $i$ be an \org gate with depth $k
+ 1$. If $\sigma(o^l_i) = x^l_i$ then, no matter which action is chosen at
$x^l_i$, we have:
\begin{align*}
\val^{\sigma}(o^l_i) &\le \val^{\sigma}(c_{1-j}) + L_{k+1} \\
&\le \val^{\sigma}(c_{1-j}) + H_{k+1} \\
&\le \val^{\sigma}(c_{1-j}) + B + H_{k+1}. 
\end{align*}
On the other hand, if $\sigma(o^l_i) = c_j$ then by the inductive hypothesis we
have:
\begin{align*}
\val^{\sigma}(o^l_i) &\le \val^{\sigma}(c_{1-j}) + B + H_{k} + b_{k+1}\\
&= \val^{\sigma}(c_{1-j}) + B + H_{k+1}.
\end{align*}
This completes the \org gate case of the inductive step.

For the second version of the inductive step, suppose that the claim has been
shown for all gates with depth $k$, and let $i$ be a \notg gate with depth $k+1$.
Note that since we are in phase $j$, no matter which action is chosen by $a^l_i$, we must have
$\val^{\sigma}(a^l_i) \le \val^{\sigma}(c_{1-j}) + H_k$. 
Therefore, if
$\sigma(o^l_i) = a^l_i$, then by Lemma~\ref{lem:argappeal}, we have:
\begin{align*}
\val^{\sigma}(o^l_i) &\le \val^{\sigma}(c_{1-j}) + H_{k} + b_{k+1}\\
&= \val^{\sigma}(c_{1-j}) + H_{k+1}\\
&\le \val^{\sigma}(c_{1-j}) + B + H_{k+1}.
\end{align*}
On the other hand, if $\sigma(o^l_i) = o^l_{\inp(i)}$, then by the inductive
hypothesis we have:
\begin{align*}
\val^{\sigma}(o^l_i) &\le \val^{\sigma}(c_{1-j}) + B + H_k \\
&< \val^{\sigma}(c_{1-j}) + B + H_{k+1}.
\end{align*}
\qed
\end{proof}

The following lemma gives an upper bound on the gate values for the case where,
for every input bit $i$, we have that both $l^{1-j}_i$ and $r^{1-j}_i$ have
switched to $c_{1-j}$.
\begin{lemma}
\label{lem:gateupperfinal}
Suppose that we are in phase $j$, and let $\sigma$ be a policy with
$\sigma(l^{1-j}_{i'}) = \sigma(r^{1-j}_{i'}) = c_{1-j}$ for every input bit
$i'$.  For every gate $i$, we have:
\begin{itemize}
\item $\val^{\sigma}(o^j_i) \le \val^{\sigma}(c_{1-j}) + H_{d(i)}$, and
\item $\val^{\sigma}(o^{1-j}_i) \le \val^{\sigma}(c_{1-j}) + H_{d(i)}$.
\end{itemize}
\end{lemma}
\begin{proof}
We start with the circuit $1-j$. Since by assumption we have
$\sigma(l^{1-j}_{i'}) = 
\sigma(r^{1-j}_{i'}) = c_{1-j}$, for every input bit $i'$, we must have
$\val^{\sigma}(o^{1-j}_{i'}) \le \val^{\sigma}(c_{1-j}) + H_0$ for every input
bit $i'$. Therefore, we can apply Lemma~\ref{lem:subupper} with $B = 0$ to
prove the second claim of this lemma.

Now consider the circuit $j$. For every input bit $i'$, we can apply the second
part of this lemma to prove that $\val^{\sigma}(o^{1-j}_{\inp(i')}) \le
\val^{\sigma}(c_{1-j}) + H_{d(C)}$. Therefore, no matter what choices $\sigma$
makes at $l^j_{i'}$ or $r^j_{i'}$, we have $\val^{\sigma}(o^{j}_{i'}) \le
\val^{\sigma}(c_{1-j}) + H_0$. Thus, we can apply Lemma~\ref{lem:subupper} with
$B = 0$ to prove the first claim of this lemma.
\qed
\end{proof}

Finally, we show a lower bound on the gate values with a different set of
assumptions to the ones used in Lemma~\ref{lem:gatelower}.
\begin{lemma}
\label{lem:gatelowerfinal}
Suppose that we are in phase $j$, and let $\sigma$ such that, for
every gate $i'$ we have:
\begin{itemize}
\item If $i'$ is an input bits, then we have $\sigma(l^{1-j}_{i'}) = \sigma(l^{1-j}_{i'}) = c_{1-j}$. 
\item If $i'$ is an \org gate, then we have $\sigma(x^{1-j}_{i'}) = c_{1-j}$. 
\item If $i'$ is a \notg gate, then we have $\sigma(a^{1-j}_{i'}) = c_{1-j}$.
\end{itemize}
For every gate $i$, we have:
\begin{equation*}
\val^{\sigma}(o^{1-j}_i) \ge \val^{\sigma}(c_{1-j}).
\end{equation*}
\end{lemma}
\begin{proof}
We will prove this claim by induction over depths. For the gates $i$ with depth
$0$, which are the input bits, the claim holds trivially, because we have
$\sigma(l^{1-j}_i) = \sigma(r^{1-j}_i) = c_{1-j}$.

We will prove two versions of the inductive step. For the first version, assume
that the claim holds for all gates with depth at most $k$, and let $i$ be an
\org
gate with depth $k+1$. If $\sigma(o^{1-j}_i) = x^{1-j}_i$, then since
$\sigma(x^{1-j}_i) = c_{1-j}$, we clearly have $\val^{\sigma}(o^{1-j}_i) \ge
\val^{\sigma}(c_{1-j})$. On the other hand, if $\sigma(o^{1-j}_i) = v^{1-j}_i$,
then we can apply the inductive hypothesis to argue that $\val^{\sigma}(o^{1-j}_i) \ge
\val^{\sigma}(c_{1-j})$.

For the second version of the inductive step, suppose that the claim holds for
all gates with depth at most $k$, and let $i$ be a \notg gate with depth $k + 1$.
If $\sigma(o^{1-j}_i) = a^{1-j}_i$, then since $\sigma(a^{1-j}_i) = c_{1-j}$,
we have  $\val^{\sigma}(o^{1-j}_i) \ge \val^{\sigma}(c_{1-j})$. On the other
hand, if $\sigma(o^{1-j}_i) = o^{1-j}_{\inp(i)}$, then we can invoke the
inductive hypothesis in order to argue that $\val^{\sigma}(o^{1-j}_i) \ge
\val^{\sigma}(c_{1-j})$. 
\qed
\end{proof}

\section{Proof of Lemma~\ref{lem:zero}}
\label{app:zero}

\begin{proof}
We begin by proving the first claim. It is not difficult to verify
that for every state $s$, if we follow $\sigma$ from $s$, then will eventually
arrive at the state $\sink$. The only action available at $\sink$ is a self-loop
with reward $0$, so this implies that the expected average-reward of all state
$s$ under $\sigma$ must be $0$. Therefore, we have $G^{\sigma}(s) = 0$ for every
state $s$.

We now prove the second claim. The fact that there exists a policy $\sigma$ with
$G^{\sigma}(s) = 0$ for every state $s$ follows from the first part of this
lemma. We now prove that we cannot have $G^{\sigma}(s) > 0$, for any policy
$\sigma$ or state $s$. This holds because the only way of avoiding reaching
$\sink$ is to take the transition between $r^j_i$ and $o^{1-j}_{\inp(i)}$
infinitely often, for some state $i$. But this action has reward $-\frac{T}{2}$,
and by Lemma~\ref{lem:subupper}, the maximum possible reward that can be
obtained on a path from $o^{1-j}_{\inp(i)}$ to $r^j_i$ is $2 \cdot H_{d(C)}$.
Since $2 \cdot H_{d(C)} - \frac{T}{2} < 0$, we have that any such loop has a
negative average reward. Therefore, the optimal policy $\sigma^*$ must satisfy
$G^{\sigma^*}(s) = 0$ for every state $s$.
\qed
\end{proof}

\section{Proof of Lemma~\ref{lem:noswitch}}
\label{app:noswitch}

In order to prove Lemma~\ref{lem:noswitch}, we will give explicit upper bounds
on every state. The following lemma states these upper bounds, and the rest of
this section will be dedicated to proving that they are correct. The proofs in
this section will make use of the upper and lower bounds shown in
Appendix~\ref{app:upper}.

\begin{lemma}
\label{lem:buffappeal}
Suppose that we are in phase $j$, and let $\sigma$ be a coherent policy. For
every input bit $i$ we have:
\begin{enumerate}
\item The appeal of switching $l^j_i$ to $c_{1-j}$ is $\lonej$.
\item The appeal of switching $r^j_i$ to $o_{\inp(i)}$ is strictly less than~$0$.
\item The appeal of switching $l^{1-j}_i$ to $c_{1-j}$ is $\lj$.
\item The appeal of switching $r^{1-j}_i$ to $c_{1-j}$ is at most $\rjprime$.
\item If $\sigma(o^j_i) = l^j_i$, then the appeal of switching $o^j_i$ to
$r^j_i$ is strictly less than~$0$.
\item If $\sigma(o^j_i) = r^j_i$, then the appeal of switching $o^j_i$ to
$l^j_i$ is strictly less than $\magic$.
\end{enumerate}
For every other gate $i$ we have:
\begin{enumerate}
\item If $i$ is an \org gate, then the appeal of switching $x^j_i$ to $c_{1-j}$ is $\xj$, and the appeal of switching $x^{1-j}_i$ to $c_{1-j}$ is $\xj$
\item If $i$ is a \notg gate, then the appeal of switching $a^{1-j}_i$ to
$c_{1-j}$ is $\aj$.
\end{enumerate}
\end{lemma}

The following lemma proves the first claim of Lemma~\ref{lem:buffappeal}.

\begin{lemma}
\label{lem:lto1j}
Suppose that we are in phase $j$, and let $\sigma$ be a policy with
$\sigma(l^j_i) = c_j$.
For every input bit $i$, the appeal of switching $l^j_i$ to $c_{1-j}$ is $1.6$.
\end{lemma}
\begin{proof}
Since we are in phase $j$, we must have that:
\begin{equation*}
\val^{\sigma}(c_{1-j}) = \val^{\sigma}(c_j) + T.
\end{equation*}
Moreover, we have that: 
\begin{equation*}
\val^{\sigma}(l^j_i) = \val^{\sigma}(c_j) + H_0.
\end{equation*} 
Thus, by Lemma~\ref{lem:argappeal}, we have that the appeal of switching $l^j_i$
to $c_{1-j}$ is:
\begin{equation*}
p_{5} \cdot (T -\frac{T}{2} + \frac{H_{d(C)} + L_{d(C)}}{2} - H_0) =
\lonej.
\end{equation*}
\qed
\end{proof}

The next lemma proves the second claim of Lemma~\ref{lem:buffappeal}.

\begin{lemma}
Suppose that we are in phase $j$, and let $\sigma$ be a coherent policy.
The action between $r^j_i$ and $o^{1-j}_{I(i)}$ is not switchable.
\end{lemma}
\begin{proof}
Since $\sigma$ is coherent, we have $\val^{\sigma}(r^j_i) = \val^{\sigma}(c_j) + L_0$. On the
other hand, by Lemma~\ref{lem:gateupper} we have:
\begin{equation*}
\val^{\sigma}(o^{1-j}_{I(i)}) \le \val^{\sigma}(c_{j}) + H_{d(C)}.
\end{equation*}
Thus, by Lemma~\ref{lem:argappeal} the appeal of switching $r^j_i$ to $o^{1-j}_{I(i)}$ is:
\begin{align*}
p_{7} \cdot (-\frac{T}{2} + H_{d(C)} - L_0) &\le 
p_{7} \cdot (-\frac{3^{d(C)+6}}{2} + 2\cdot 3^{d(C)+2}) \\
&<0.
\end{align*}
\qed
\end{proof}

The next lemma proves the third claim of Lemma~\ref{lem:buffappeal}.

\begin{lemma}
\label{lem:third}
Suppose that we are in phase $j$, and let $\sigma$ be a policy with
$\sigma(l^{1-j}) = c_j$ for some input bit $i$. The appeal of switching
$l^{1-j}_i$ to $c_{1-j}$ is $\lj$.
\end{lemma}
\begin{proof}
By Lemma~\ref{lem:argappeal}, we have:
\begin{equation*}
\val^{\sigma}(l^{1-j}_i) = \val^{\sigma}(c_{j}) - \frac{T}{2} + \frac{H_{d(C)} + L_{d(C)}}{2}. 
\end{equation*}
Moreover, since we are in phase $j$, we have that $\val^{\sigma}(c_{1-j}) =
\val^{\sigma}(c_{j}) + T$. Thus, by Lemma~\ref{lem:argappeal} we have that the
appeal of switching $l^{1-j}_i$ to $c_{1-j}$ is:
\begin{equation*}
p_{4} \cdot (T + H_0 + \frac{T}{2} - \frac{H_{d(C)} + L_{d(C)}}{2})
= \lj.
\end{equation*}
\qed
\end{proof}

The next lemma proves the fourth claim of Lemma~\ref{lem:buffappeal}

\begin{lemma}
\label{lem:fourth}
Suppose that we are in phase $j$ and let $\sigma$ be a policy with
$\sigma(l^j_{i'}) = c_j$ and $\sigma(r^j_{i'}) = c_j$ for every input bit $i'$.
For every input bit $i$, the appeal of switching $r^{1-j}_i$ to $c_{1-j}$ is at
most \rjprime.
\end{lemma}
\begin{proof}
Since we are in phase $j$, we have $\val^{\sigma}(c_{1-j}) = \val^{\sigma}(c_j)
+ T$.
By Lemma~\ref{lem:gatelower}, we have
$\val^{\sigma}(o^j_{\inp(i)}) \ge \val^{\sigma}(c_j)$. Therefore, we can apply
Lemma~\ref{lem:argappeal} to argue that the appeal of switching $r^{1-j}_i$ to
$o^{j}_{\inp(i)}$ is at most:
\begin{align*}
p_{6} \cdot (T + L_0 + \frac{T}{2}) &= 
\rj \cdot \frac{\frac{3T}{2} + L_0}{\frac{3T}{2} + L_0 - H_{d(C)}} \\
&= \rj + \frac{ \rj \cdot H_{d(C)}}{\frac{3T}{2} + L_0 - H_{d(C)}} \\
&\le \rj + \frac{ \rj \cdot H_{d(C)}}{\frac{3T}{2} - H_{d(C)}} \\
&\le \rj + \frac{ \rj \cdot 2 \cdot 3^{d(C) + 2}}{\frac{3}{2} \cdot 3^{d(C) + 6}
- 2 \cdot 3^{d(C) + 2}} \\
&= \rj + \frac{ \rj \cdot 2 \cdot 3^{d(C) + 2}}{1.5 \cdot 3^4 \cdot 3^{d(C) + 2}
- 2 \cdot 3^{d(C) + 2}} \\
&= \rj + \frac{ \rj \cdot 2 \cdot 3^{d(C) + 2}}{119.5 \cdot 3^{d(C) + 2}} \\
&\le \rj + 0.1.
\end{align*}
\qed
\end{proof}

The next lemma proves the fifth claim of Lemma~\ref{lem:buffappeal}.

\begin{lemma}
\label{lem:five}
Suppose that we are in phase $j$, and let $\sigma$ be a policy with
$\sigma(l^j_i) = \sigma(r^j_i) = c_j$ and with $\sigma(o^j_i) = l^j_i$ for some
input bit $i$. The appeal of switching $o^j_i$ to $r^j_i$ is strictly less
than~$0$.
\end{lemma}
\begin{proof}
We have:
\begin{equation*}
\val^{\sigma}(o^j_i) = \val^{\sigma}(c_j) + H_0.
\end{equation*}
On the other hand, we have:
\begin{equation*}
\val^{\sigma}(r^j_i) = \val^{\sigma}(c_j) + L_0.
\end{equation*}
Thus, the appeal of switching $o^j_i$ to $r^j_i$ is $L_0 - H_0 < 0.$
\qed
\end{proof}

The next lemma proves the sixth claim of Lemma~\ref{lem:buffappeal}.

\begin{lemma}
\label{lem:six}
Suppose that we are in phase $j$, and let $\sigma$ be a policy with
$\sigma(l^j_i) = \sigma(r^j_i) = c_j$ and with $\sigma(o^j_i) = r^j_i$ for some
input bit $i$. The appeal of switching $o^j_i$ to $l^j_i$ is strictly less
than $\magic$.
\end{lemma}
\begin{proof}
We have:
\begin{equation*}
\val^{\sigma}(o^j_i) = \val^{\sigma}(c_j) + L_0.
\end{equation*}
On the other hand, we have:
\begin{equation*}
\val^{\sigma}(l^j_i) = \val^{\sigma}(c_j) + H_0.
\end{equation*}
Thus, we can apply Lemma~\ref{lem:argappeal} to argue that the appeal of
switching $o^j_i$ to $l^j_i$ is:
\begin{align*}
p_{3} \cdot (H_0 - L_0) 
&= \frac{\bl \cdot (H_0 - L_0)}{\frac{3T}{2} + H_0} \\
&\le \frac{\bl \cdot H_0}{\frac{3T}{2}} \\
&\le \frac{\bl \cdot 2 \cdot 3^{d(C)+2}}{\frac{3}{2} \cdot 3^{d(C)+6}} \\
&= \frac{\bl \cdot 2}{121.5} \\
&<\magic.
\end{align*}
\end{proof}

The next lemma proves the seventh claim of Lemma~\ref{lem:buffappeal}.

\begin{lemma}
\label{lem:nine}
Suppose that we are in phase $j$, let $m \in \{0,1\}$, and let $\sigma$ be a policy with
$\sigma(x^m_i) = c_j$ for some \org gate $i$. The appeal of switching $x^m_i$ to
$c_{1-j}$ is $\xj$.
\end{lemma}
\begin{proof}
We have:
\begin{equation*}
\val^{\sigma}(x^m_i) = \val^{\sigma}(c_j).
\end{equation*}
Thus, we can apply Lemma~\ref{lem:argappeal} to argue that the appeal of
switching $x^m_i$ to $c_{1-j}$ is:
\begin{equation*}
\frac{\xj}{T} \cdot T = \xj.
\end{equation*}
\qed
\end{proof}

The next lemma proves the eight claim of Lemma~\ref{lem:buffappeal}.

\begin{lemma}
\label{lem:ten}
Suppose that we are in phase $j$, and let $\sigma$ be a policy with
$\sigma(a^{1-j}_i) = c_j$ for some \notg gate $i$. The appeal of switching 
$a^{1-j}_i$ to $c_{1-j}$ is $\aj$.
\end{lemma}
\begin{proof}
We have:
\begin{equation*}
\val^{\sigma}(a^{1-j}_i) = \val^{\sigma}(c_j) - T + H_{d(i) - 1}.
\end{equation*}
On the other hand, we have:
\begin{equation*}
\val^{\sigma}(c_{1-j}) = \val^{\sigma}(c_j) + T.
\end{equation*}
Thus, we can apply Lemma~\ref{lem:argappeal} to argue that the appeal of
switching $a^{1-j}_i$ to $c_{1-j}$ is:
\begin{equation*}
p_{2} \cdot (2T - H_{d(i) - 1}) = \aj.
\end{equation*}
\qed
\end{proof}

\section{Proof of Lemma~\ref{lem:transition}}
\label{app:transition}

Our goal is to show that, when Dantzig's rule is applied to $\sigma$, the
sequence of events specified in Lemma~\ref{lem:transition} will occur. To do
this, we split events specified in Lemma~\ref{lem:transition} into four
different \emph{stages} as follows:
\begin{itemize}
\item Stage 1 encompasses the following event:
\begin{enumerate}
\setcounter{enumi}{0}
\item For every input bit $i$, the state $l^{1-j}_i$ is switched to $c_{1-j}$.
\end{enumerate}
\item Stage 2 encompasses the following event:
\begin{enumerate}
\setcounter{enumi}{1}
\item For every input bit $i$, the state $r^{1-j}_i$ is switched to $c_{1-j}$.
\end{enumerate}
\item Stage 3 encompasses the following events:
\begin{enumerate}
\setcounter{enumi}{2}
\item For every input bit $i$, the state $l^j_i$ is switched to $c_{1-j}$. 
\item For every input bit $i$, if $o^j_i$ takes the action towards $r^j_i$, then
it is switched to $l^j_i$.
\end{enumerate}
\item Stage 4 encompasses the following events:
\begin{enumerate}
\setcounter{enumi}{4}
\item For every input bit $i$, the state $r^j_i$ is switched to
$o^{1-j}_{\inp(i)}$.
\item For every \notg gate $i$, the state $a^{1-j}_i$ is switched to $c_{1-j}$.
\item For every \org gate $i$, the states $x^j_i$ and $x^{1-j}_i$ are both
switched to $c_{1-j}$.
\end{enumerate}
\end{itemize}

For each stage, we will give a single lemma that specifies the upper and lower
bounds on appeals that are necessary in order to show that a given action is
switched. Since the proofs of these lemmas are lengthy, we defer each of them to
their own appendix.

\subsection{Stage 1}

We say that $\sigma$ is a \emph{stage 1 transition} policy for phase $j$ if, for
every gate $i$, the following conditions are satisfied:
\begin{itemize}
\item If $i$ is an input bit then:
\begin{itemize}
\item We have $\sigma(l^j_i) = c_j$ and $\sigma(r^j_i) = c_j$.
\item We have $\sigma(r^{1-j}_i) = o^j_{\inp(i)}$.
\end{itemize}
\item If $i$ is an \org gate then we have $\sigma(x^j_i) = c_j$ and $\sigma(x^{1-j}_i) = c_j$.
\item If $i$ is a \notg gate then we have $\sigma(a^{1-j}_i) = c_j$.
\end{itemize}
Note that we do not place restrictions on the
choice made at $l^{1-j}_i$ for any input bit $i$, because these are the states
that will be switched during stage 1. Moreover, note that every final policy for
$B$ satisfies the requirements of a stage 1 transition policy.

The following lemma is proved in Appendix~\ref{app:stage1appeal}.
\begin{lemma}
\label{lem:stage1appeal}
Suppose that we are in phase $j$, and let $\sigma$ be a stage 1 transition
policy. For every input bit $i$ we have:
\begin{enumerate}
\item The appeal of switching $l^j_i$ to $c_{1-j}$ is $\lonej$.
\item The appeal of switching $r^j_i$ to $o_{\inp(i)}$ is at most $\ro$. 
\item 
\begin{enumerate}
\item If $\sigma(l^{1-j}_i) = c_j$, then the appeal of switching 
$l^{1-j}_i$ to $c_{1-j}$ is at least $\lj$.
\item If $\sigma(l^{1-j}_i) = c_{1-j}$, then the appeal of switching $l^{1-j}_i$
to $c_j$ is strictly less than $0$.
\end{enumerate}
\item The appeal of switching $r^{1-j}_i$ to $c_{1-j}$ is at least $\rj$, and at
most $\rjprime$.
\item If $\sigma(o^j_i) = l^j_i$, then the appeal of switching $o^j_i$ to
$r^j_i$ is strictly less than~$0$
\item If $\sigma(o^j_i) = r^j_i$, then the appeal of switching $o^j_i$ to
$l^j_i$ is strictly less than $\magic$.
\item If $\sigma(o^{1-j}_i) = l^{1-j}_i$, and $\val^{\sigma}(o^j_{\inp(i)}) \le
\val^{\sigma}(c_j) + L_{d(C)}$ then the appeal of switching $o^{1-j}_i$ to
$r^{1-j}_i$ is strictly less than $0$.
\item If $\sigma(o^{1-j}_i) = r^{1-j}_i$, then the appeal of switching $o^{1-j}_i$ to $l^{1-j}_i$ is at most \bl.
\end{enumerate}
For every other gate $i$ we have:
\begin{enumerate}
\setcounter{enumi}{8}
\item If $i$ is an \org gate, then the appeal of switching $x^j_i$ to $c_{1-j}$ is $\xj$, and the appeal of switching $x^{1-j}_i$ to $c_{1-j}$ is $\xj$
\item If $i$ is a \notg gate, then the appeal of switching $a^{1-j}_i$ to
$c_{1-j}$ is $\aj$.
\end{enumerate}
\end{lemma}

Firstly we argue that, for every input bit $i$, the state $o^{1-j}_i$ cannot be
switched during stage $1$. If $\sigma_f(o^{1-j}_i) = r^{1-j}_i$, then this
follows immidiately from part 8 of Lemma~\ref{lem:stage1appeal}. On the other
hand, if $\sigma_f(o^{1-j}_i) = l^{1-j}_i$, then we must show that the
precondition on part 7 of Lemma~\ref{lem:stage1appeal} will always be satisfied.
We do so in the following lemma.

\begin{lemma}
\label{lem:ffs}
Suppose that we are in phase $j$, and let $\sigma$ be a phase $1$ transition
policy in which all gates $i$ are final and correct. Suppose that policy
iteration switches, for some input bit $i'$, the state $l^{1-j}_{i'}$ to $c_{1-j}$.
Let $\sigma'$ be the resulting policy. We have that all gates $i$ are final and
correct in $\sigma'$.
\end{lemma}
\begin{proof}
In every stage $1$ transition policy, we have $\sigma(l^j_i) = \sigma(r^j_i) =
c_j$, for every input bit $i$. Thus, every state $s$ in the gadgets representing
circuit $j$, it is impossible to move from $s$ to $l^{1-j}_{i'}$. Thus,
$\val^{\sigma}(s) = \val^{\sigma'}(s)$. So, for every state $s$ in the gadgets
representing circuit $j$, if $s$ has no outgoing actions to circuit $1-j$, then
$s$ must be correct and final in $\sigma'$.

To complete the proof, we observe that the only states in circuit $j$
that have an action to a state in circuit $1-j$ are the states $r^j_i$ where $i$
is an input bit, but
part $2$ of Lemma~\ref{lem:stage1appeal} proves that these are final and correct
in $\sigma'$.
\qed
\end{proof}

Lemma~\ref{lem:ffs} implies that, if $\val^{\sigma_f}(o^j_{\inp(i)}) \le
\val^{\sigma_f}(c_j) + L_{d(C)}$, for some input bit $i$, then we will continue
to have $\val^{\sigma}(o^j_{\inp(i)}) \le \val^{\sigma}(c_j) + L_{d(C)}$ for
every policy $\sigma$ that is encountered during stage 1. Note that
$\sigma_f(o^{1-j}_i) = l^{1-j}_i$ for an input bit $i$ if and only if 
$\val^{\sigma}(o^j_{\inp(i)}) \le \val^{\sigma}(c_j) + L_{d(C)}$, so we have
shown that $o^{1-j}_{\inp(i)}$ cannot be switched during stage $1$.

We now show that, for every input bit $i$, the Dantzig's rule will switch state
$l^{1-j}_i$ to $c_{1-j}$. Note that we cannot switch away from a stage 1
transition policy without switching some state that appears in
Lemma~\ref{lem:stage1appeal}. However, Lemma~\ref{lem:stage1appeal} shows that
switching $l^{1-j}_i$ to $c_{1-j}$ has appeal at least $\lj$, whereas all other
actions considered in Lemma~\ref{lem:stage1appeal} have appeal strictly less
than $\lj$. Thus, starting at the policy $\sigma_f$, Dantzig's rule will
eventually switch to a policy $\sigma_1$ in which, for every input bit $i$, we
have $\sigma_1(l^{1-j}_i)= c_{1-j}$. The policy $\sigma_1$ will be the first
policy considered in stage 2.

\section{Stage 2}

We say that $\sigma$ is a \emph{stage 2 transition} policy for phase $j$ if, for
every gate $i$, the following conditions are satisfied.
\begin{itemize}
\item If $i$ is an input bit then:
\begin{itemize}
\item We have $\sigma(l^j_i) = c_j$ and $\sigma(r^j_i) = c_j$.
\item We have $\sigma(l^{1-j}_i) = c_{1-j}$.
\end{itemize}
\item If $i$ is an \org gate then we have $\sigma(x^j_i) = c_j$ and $\sigma(x^{1-j}_i) = c_j$.
\item If $i$ is a \notg gate then we have $\sigma(a^{1-j}_i) = c_j$.
\end{itemize}
Note that we do not place any restrictions on the choice made at $r^{1-j}_i$ for
any input bit $i$, because these states will be switched during stage $2$. Also
note that the policy $\sigma_1$ satisfies the requirements of a stage 2
transition policy.

The following lemma is proved in Appendix~\ref{app:stage2appeal}

\begin{lemma}
\label{lem:stage2appeal}
Suppose that we are in phase $j$, and let $\sigma$ be a stage 2 transition
policy. For every input bit $i$ we have:
\begin{enumerate}
\item The appeal of switching $l^j_i$ to $c_{1-j}$ is $\lonej$.
\item The appeal of switching $r^j_i$ to $o_{\inp(i)}$ is at most $\ro$. 
\item The appeal of switching $l^{1-j}_i$ to $c_{j}$ is strictly less than $0$.
\item 
\begin{enumerate}
\item If $\sigma(r^{1-j}_i) = o^j_{\inp(i)}$, then the appeal of switching
$r^{1-j}_i$ to $c_{1-j}$ is at least $\rj$. 
\item If $\sigma(r^{1-j}_i) = c_{1-j}$, then the appeal of switching $r^{1-j}_i$
to $o^j_{\inp(i)}$ is strictly less than $0$. 
\end{enumerate}
\item If $\sigma(o^j_i) = l^j_i$, then the appeal of switching $o^j_i$ to
$r^j_i$ is strictly less than~$0$
\item If $\sigma(o^j_i) = r^j_i$, then the appeal of switching $o^j_i$ to
$l^j_i$ is strictly less than $\magic$.
\item If $\sigma(o^{1-j}_i) = l^{1-j}_i$, then the appeal of switching
$o^{1-j}_i$ to $r^{1-j}_i$ is strictly less than $0$.
\item If $\sigma(o^{1-j}_i) = r^{1-j}_i$, then the appeal of switching $o^{1-j}_i$ to $l^{1-j}_i$ is at most \bl.
\end{enumerate}
For every other gate $i$ we have:
\begin{enumerate}
\setcounter{enumi}{8}
\item If $i$ is an \org gate, then the appeal of switching $x^j_i$ to $c_{1-j}$ is $\xj$, and the appeal of switching $x^{1-j}_i$ to $c_{1-j}$ is $\xj$
\item If $i$ is a \notg gate, then the appeal of switching $a^{1-j}_i$ to
$c_{1-j}$ is $\aj$.
\end{enumerate}
\end{lemma}

We can now argue that, starting at the policy $\sigma_1$, Dantzig's rule will
switch, for every input bit $i$, the state $r^{1-j}_i$ to $c_{1-j}$. This is
because we cannot switch away from a stage 2 transition policy without switching
a state that appears in Lemma~\ref{lem:stage2appeal}. However,
Lemma~\ref{lem:stage2appeal} shows that switching $r^{1-j}_i$ to $c_{1-j}$ has
appeal at least $\rj$, whereas all other actions considered by
Lemma~\ref{lem:stage2appeal} have appeal strictly less than $\rj$. Thus, when
applied to the policy $\sigma_1$, Dantzig's rule will eventually switch to a
policy $\sigma_2$ in which, for every input bit $i$, we have
$\sigma(r^{1-j}_i)=c_{1-j}$.  The policy $\sigma_2$ will be the first policy
considered in stage 3.

\subsection{Stage 3}

We say that $\sigma$ is a \emph{stage 3 transition} policy for phase $j$ if, for
every gate $i$, the following conditions are satisfied.
\begin{itemize}
\item If $i$ is an input bit then:
\begin{itemize}
\item We have $\sigma(r^j_i) = c_j$.
\item We have $\sigma(l^{1-j}_i) = c_{1-j}$.
\item We have $\sigma(r^{1-j}_i) = c_{1-j}$.
\end{itemize}
\item If $i$ is an \org gate then we have $\sigma(x^j_i) = c_j$ and $\sigma(x^{1-j}_i) = c_j$.
\item If $i$ is a \notg gate then we have $\sigma(a^{1-j}_i) = c_j$.
\end{itemize}
Note that we do not place any restrictions on the choice made at
$l^j_i$ for any input bit $i$, because these are the states that will be
switched during stage 3. Note also that $\sigma_2$ satisfies the requirements of
a stage 3 transition policy.

The following lemma is proved in Appendix~\ref{app:stage3appeal}.
\begin{lemma}
\label{lem:stage3appeal}
Suppose that we are in phase $j$, and let $\sigma$ be a stage 3 transition
policy. For every input bit $i$ we have:
\begin{enumerate}
\item
\begin{enumerate}
\item If $\sigma(l^j_i) = c_j$, then the appeal of switching $l^j_i$ to $c_{1-j}$ is $\lonej$.
\item If $\sigma(l^j_i) = c_{1-j}$, then the appeal of switching $l^j_i$ to
$c_j$ is strictly less than $0$. 
\end{enumerate}
\item The appeal of switching $r^j_i$ to $o_{\inp(i)}$ is  at most $\ro$. 
\item The appeal of switching $l^{1-j}_i$ to $c_{j}$ is strictly less than $0$.
\item The appeal of switching $r^{1-j}_i$ to $o^j_{\inp(i)}$ is strictly less than $0$. 
\item If $\sigma(o^j_i) = l^j_i$, then the appeal of switching $o^j_i$ to
$r^j_i$ is strictly less than~$0$
\item If $\sigma(o^j_i) = r^j_i$, then
\begin{enumerate}
\item If $\sigma(l^j_i) = c_{j}$, then the appeal of switching $o^j_i$ to
$l^j_i$ is strictly less than $\magic$.
\item If $\sigma(l^j_i) = c_{1-j}$, then the appeal of switching $o^j_i$ to
$l^j_i$ lies in the range $[\dnc, \dncu]$.
\end{enumerate}
\item If $\sigma(o^{1-j}_i) = l^{1-j}_i$, then the appeal of switching
$o^{1-j}_i$ to $r^{1-j}_i$ is strictly less than $0$.
\item If $\sigma(o^{1-j}_i) = r^{1-j}_i$, then the appeal of switching
$o^{1-j}_i$ to $l^{1-j}_i$ is strictly less than $\magic$.
\end{enumerate}
For every other gate $i$ we have:
\begin{enumerate}
\setcounter{enumi}{8}
\item If $i$ is an \org gate, then the appeal of switching $x^j_i$ to $c_{1-j}$ is $\xj$, and the appeal of switching $x^{1-j}_i$ to $c_{1-j}$ is $\xj$
\item If $i$ is a \notg gate, then the appeal of switching $a^{1-j}_i$ to
$c_{1-j}$ is $\aj$.
\end{enumerate}
\end{lemma}

We can now argue that, when applied to the policy $\sigma_2$, Dantzig's rule
will execute the following two events from Lemma~\ref{lem:transition}:
\begin{enumerate}
\setcounter{enumi}{2}
\item For every input bit $i$, the state $l^j_i$ is switched to $c_{1-j}$. 
\item For every input bit $i$, if $o^j_i$ takes the action towards $r^j_i$, then
it is switched to $l^j_i$.
\end{enumerate}
We first show that event 3 occurs. This follows from the fact that, in order to
switch away from a stage 3 transition policy, we must switch some action
mentioned in Lemma~\ref{lem:stage3appeal}. However, Lemma~\ref{lem:stage3appeal}
shows that switching $l^j_i$ to $c_{1-j}$ has appeal $\lonej$, whereas all other
actions mentioned in Lemma~\ref{lem:stage3appeal} have appeal strictly less than
$\lonej$. Therefore, when it is applied to $\sigma_2$, Dantzig's rule will
eventually switch to a policy $\sigma'_2$ in which $\sigma'_2(l^j_i) = c_{1-j}$
for every input bit $i$.

We can now prove that event 4 occurs when Dantzig's rule is applied to
$\sigma'_2$. This follows from the same line of reasoning: once $l^j_i$ has been
switched to $c_{1-j}$ for all input bits $i$, Lemma~\ref{lem:stage3appeal}
states that the appeal of switching $o^j_i$ to $l^j_i$ is at least $\dnc$.
Moreover, all other actions mentioned in Lemma~\ref{lem:stage3appeal} have
appeal strictly less than $\dnc$. Thus, Dantzig's rule will eventually switch to
a policy $\sigma_3$ in which, for every input bit $i$, we have $\sigma_3(l^j_i)
= c_{1-j}$ and $\sigma_3(o^j_i) = l^j_i$. The policy $\sigma_3$ will be the
first policy considered during stage 4.

\subsection{Stage 4}

We say that $\sigma$ is a \emph{stage 4 transition} policy for phase $j$ if, for
every gate $i$, the following conditions are satisfied.
\begin{itemize}
\item If $i$ is an input bit then:
\begin{itemize}
\item We have $\sigma(l^j_i) = c_{1-j}$.
\item We have $\sigma(l^{1-j}_i) = c_{1-j}$.
\item We have $\sigma(r^{1-j}_i) = c_{1-j}$.
\item We have $\sigma(o^j_i) = l^j_i$.
\end{itemize}
\end{itemize}
Note that we do not place any restrictions on the choice made at
the states $r^j_i$ for any input bit $i$, $x^j_i$ for any \org gate $i$, or
$a^j_i$ for any \notg gate $i$, because these are the states that will be switched
during stage 4. Note also that the policy $\sigma_3$ satisfied the conditions of
a stage 4 transition policy·

The following lemma is proved in Appendix~\ref{app:stage4appeal}.
\begin{lemma}
\label{lem:stage4appeal}
Suppose that we are in phase $j$, and let $\sigma$ be a stage 4 transition
policy. For every input bit $i$ we have:
\begin{enumerate}
\item The appeal of switching $l^j_i$ to $c_j$ is strictly less than $0$. 
\item 
\begin{enumerate}
\item If $\sigma(r^j_i) = c_j$, then the appeal of switching $r^j_i$ to
$o_{\inp(i)}$ is at most $\ro$. 
\item If $\sigma(r^j_i) = c_j$ and if, for every \org gate $i'$ we have
$\sigma(x^{1-j}_{i'}) = c_{1-j}$ and for every \notg gate $i'$ we have
$\sigma(a^{1-j}_{i'}) = c_{1-j}$, then the appeal of switching $r^j_i$ to 
$o_{\inp(i)}$ is at least \rol.
\item If $\sigma(r^j_i) = o^{1-j}_{\inp(i)}$,
and if, for every \org gate $i'$ we have
$\sigma(x^{1-j}_{i'}) = c_{1-j}$ and for every \notg gate $i'$ we have
$\sigma(a^{1-j}_{i'}) = c_{1-j}$,
then the appeal of switching $r^j_i$
to $c_{j}$ is strictly less than $0$.
\end{enumerate}
\item The appeal of switching $l^{1-j}_i$ to $c_{j}$ is strictly less than $0$.
\item The appeal of switching $r^{1-j}_i$ to $o^j_{\inp(i)}$ is strictly less than $0$. 
\item The appeal of switching $o^j_i$ to
$r^j_i$ is strictly less than~$0$
\item If $\sigma(o^{1-j}_i) = l^{1-j}_i$, then the appeal of switching
$o^{1-j}_i$ to $r^{1-j}_i$ is strictly less than $0$.
\item If $\sigma(o^{1-j}_i) = r^{1-j}_i$, then the appeal of switching
$o^{1-j}_i$ to $l^{1-j}_i$ is strictly less than $\magic$.
\end{enumerate}
For every other gate $i$ we have:
\begin{enumerate}
\setcounter{enumi}{8}
\item If $i$ is an \org gate then for all $l \in \{0, 1\}$:
\begin{enumerate}
\item  If $\sigma(x^l_i) = c_j$, then the appeal of switching $x^l_i$ to
$c_{1-j}$ is $\xj$. 
\item If $\sigma(x^l_i) = c_{1-j}$, then the appeal of switching $x^l_i$ to
$c_j$ is strictly less than $0$.
\end{enumerate}
\item If $i$ is a \notg gate then:
\begin{enumerate}
\item If $\sigma(a^{1-j}_i) = c_j$, then the appeal of switching $a^{1-j}_i$ to $c_{1-j}$ is $\aj$.
\item If $\sigma(a^{1-j}_i) = c_{1-j}$, then the appeal of switching $a^{1-j}_i$
to $c_{j}$ is strictly less than 0.
\end{enumerate}
\end{enumerate}
\end{lemma}

We can now show that, when Dantzig's rule is applied to $\sigma_3$, the
following events from Lemma~\ref{lem:transition} will occur.
\begin{enumerate}
\setcounter{enumi}{4}
\item For every input bit $i$, the state $r^j_i$ is switched to
$o^{1-j}_{\inp(i)}$.
\item For every \notg gate $i$, the state $a^{1-j}_i$ is switched to $c_{1-j}$.
\item For every \org gate $i$, the states $x^j_i$ and $x^{1-j}_i$ are both
switched to $c_{1-j}$.
\end{enumerate}
We begin by showing that event 5 will occur. This follows from the fact that, in
order to switch away from a stage 4 transition policy, one of the actions
mentioned in Lemma~\ref{lem:stage4appeal} must be switched. However,
Lemma~\ref{lem:stage4appeal} shows that the appeal of switching $r^j_i$ to
$o^{1-j}_{\inp(i)}$ is at least $\rol$, whereas all other actions mentioned in
Lemma~\ref{lem:stage4appeal} have appeal strictly less than $\rol$. Thus,
Dantzig's rule will switch from $\sigma_3$ to a policy $\sigma'_3$ with
$\sigma'_3(r^j_i) = o^{1-j}_{\inp(i)}$ for every input bit $i$.

Events 6 and 7 follow using the same line of reasoning. After event 5 has
occurred Lemma~\ref{lem:stage4appeal} shows that, for every \notg gate $i$,
switching $a^{1-j}_i$ to $c_{1-j}$ has appeal $\aj$, whereas all other actions
have appeal strictly less than $\aj$. Similarly, once event 6 has occurred,
Lemma~\ref{lem:stage4appeal} shows that, for every \org gate $i$, switching
$x^j_i$ and $x^{1-j}_i$ to $c_{1-j}$ has appeal $\xj$, whereas all other actions
have appeal strictly less than $\xj$.
Thus Dantzig's rule will eventually switch to a policy $\sigma_4$ in which
events 5, 6, and 7 have taken place.

\subsection{Completing the proof}

Note that $\sigma_4$ is still a stage 4 transition policy. Note that, for the
policy $\sigma_4$, all appeals mentioned in are strictly less than $\magic$. On
the other hand Lemma~\ref{lem:clock} implies that the appeal of advancing
the clock is at least $0.25$. Therefore, when Dantzig's rule is applied
to $\sigma_4$, it will proceed to advance the clock, and move into phase $1-j$. 

Let $\sigma_n$ be the first policy after the clock is advanced. To complete the
proof of Lemma~\ref{lem:transition}, we must argue that $\sigma_n$ is coherent
for phase $1-j$, and that it is an initial policy for $F(B)$. To do so, we must
check the following conditions.
\begin{itemize}
\item Firstly, we must show that $\sigma_n$ is coherent for phase $1-j$. This follows from
events in stages 1 through 4, which explicitly show that $\sigma_n$ satisfies
the conditions of a coherent policy in phase $1-j$. 
\item Secondly, in order to show that $\sigma_n$ is an initial policy for
$F(B)$,  we must show that for every input bit $i$, we have $\sigma_n(o^{1-j}) =
l^j_i$ if and only if $C(B, \inp(i)) = 1$. Note that, by
Lemma~\ref{lem:circuit}, this property holds for the policy $\sigma_f$.
Moreover, we have shown that $o^{1-j}_i$ cannot have been switched at any point
during stages 1 through 4. Therefore, this property must hold for $\sigma_n$.
\item Thirdly, we must show that $\sigma_n(o^{j}_i) = l^j_i$. To do so, we
recall that 
$o^{j}_i$ was switched to $l^j_i$ during stage 3, and was not switched during
stage 4.
\item Finally, we must show that for every \notg gate $i$ we have
$\sigma(a^{1-j}_i) = c_{1-j}$. Recall that this was ensured during stage 4.
\end{itemize}

So, $\sigma_n$ is coherent for phase $1-j$, and an initial policy for $F(B)$ in
phase $1-j$. This completes the proof.

\section{Proof of Lemma~\ref{lem:stage1appeal}}
\label{app:stage1appeal}

A number of the claims in this lemma follow from lemma that we have shown
previously. In particular:
\begin{itemize}
\item The first part of this lemma follows from Lemma~\ref{lem:lto1j}.
\item The fifth claim of the lemma follows from Lemma~\ref{lem:five}.
\item The sixth claim of this lemma follows from Lemma~\ref{lem:six}.
\item The ninth part of this lemma follows from Lemma~\ref{lem:nine}.
\item The tenth claim of this lemma follows from Lemma~\ref{lem:ten}.
\end{itemize}
We now proceed to prove the other claims. The following lemma proves second
claim of Lemma~\ref{lem:stage1appeal}.

\begin{lemma}
\label{lem:second-two}
Suppose that we are in phase $j$, and let $\sigma$ be a policy with
$\sigma(l^j_{i'}) = \sigma(r^j_{i'}) = c_j$ for every input bit $i'$.
For every input bit $i$, 
the appeal of switching $r^j_i$ to $o^{1-j}_{\inp(i)}$ is at most $\ro$.
\end{lemma}
\begin{proof}
We have:
\begin{equation*}
\val^{\sigma}(r^j_i) = \val^{\sigma}(c_j) + L_0.
\end{equation*}
On the other hand, by Lemma~\ref{lem:gateupperjprime}, we have:
\begin{align*}
\val^{\sigma}(o^{1-j}_i) &\le \val^{\sigma}(c_{1-j}) + H_{d(C)} \\
&=\val^{\sigma}(c_{j}) + T + H_{d(C)}. \\
\end{align*}
Thus, we can apply Lemma~\ref{lem:argappeal} to argue that the appeal of
switching $r^j_i$ to $o^{1-j}_{\inp(i)}$ is at most:
\begin{equation*}
p_{7} \cdot (\frac{T}{2} + H_{d(C)} - L_0) = \ro.
\end{equation*}
\qed
\end{proof}

Part 3a of Lemma~\ref{lem:stage1appeal} is proved in Lemma~\ref{lem:third}.
The following lemma proves part 3b of Lemma~\ref{lem:stage1appeal}.

\begin{lemma}
\label{lem:third-two}
Suppose that we are in phase $j$, and let $\sigma$ be a policy, and let $i$ be
an input bit. If $\sigma(l^{1-j}_i) = c_{1-j}$, then the appeal of switching
$l^{1-j}_i$ to $c_j$ is strictly less than $0$.
\end{lemma}
\begin{proof}
In this case we have:
\begin{equation*}
\val^{\sigma}(l^{1-j}_i) = \val^{\sigma}(c_j) + T + H_0.
\end{equation*}
Thus, by Lemma~\ref{lem:argappeal}, we have that the appeal of switching
$l^{1-j}_i$ to $c_j$ is:
\begin{equation*}
p_{5} \cdot (-\frac{T}{2} + \frac{H_{d(C)} + L_{d(C)}}{2} - T - H_0) < 0.
\end{equation*}
\qed
\end{proof}

The upper bound in the fourth claim of Lemma~\ref{lem:stage1appeal} follows from
Lemma~\ref{lem:fourth}. The following lemma proves the lower bound in the fourth
claim of Lemma~\ref{lem:stage1appeal}.

\begin{lemma}
\label{lem:fourth-two}
Suppose that we are in phase $j$, and let $\sigma$ be a policy with
$\sigma(l^j_{i'}) = \sigma(r^j_{i'}) = c_j$ for every input bit $i'$. If
$\sigma(r^{1-j}_i) = o^j_{\inp(i)}$, for some input bit $i$, then the appeal of
switching $r^{1-j}$ to $c_{1-j}$ is at least \rj. 
\end{lemma}
\begin{proof}
We can apply Lemma~\ref{lem:gateupperj} to obtain:
\begin{equation*}
\val^{\sigma}(r^{1-j}_i) \le \val^{\sigma}(c_j) + H_{d(C)} - \frac{T}{2}.
\end{equation*}
On the other hand, we have:
\begin{equation*}
\val^{\sigma}(c_{1-j}) = \val^{\sigma}(c_j) + T.
\end{equation*}
Thus, we can apply Lemma~\ref{lem:argappeal} to argue that the appeal of
switching $r^{1-j}$ to $c_{1-j}$ is at least:
\begin{equation*}
p_{6} \cdot (\frac{3T}{2} + L_0 - H_{d(C)}) = \rj
\end{equation*}
\qed
\end{proof}

The following lemma proves the seventh claim of Lemma~\ref{lem:stage1appeal}.
\begin{lemma}
Suppose that we are in phase $j$, and let $\sigma$ be a policy with
$\sigma(r^{1-j}_i) - o^j_{\inp(i)}$ and $\sigma(o^{1-j}_i) = l^{1-j}_i$.  If
$\val^{\sigma}(o^j_{\inp(i)}) \le \val^{\sigma}(c_j) + L_{d(C)}$ then the appeal
of switching $o^{1-j}_i$ to $r^{1-j}_i$ is strictly less than $0$.
\end{lemma}
\begin{proof}
If $\sigma(l^{1-j}_i) = c_{j}$, then we have:
\begin{equation*}
\val^{\sigma}(o^{1-j}) = 
\val^{\sigma}(c_j) - \frac{T}{2} + \frac{H_{d(C)} + L_{d(C)}}{2}.
\end{equation*}
On the other hand, if $\sigma(l^{1-j}_i) = c_{1-j}$, then we have:
\begin{align*}
\val^{\sigma}(o^{1-j}) &= \val^{\sigma}(c_j) + T + H_0. \\
&\ge \val^{\sigma}(c_j) - \frac{T}{2} + \frac{H_{d(C)} + L_{d(C)}}{2}.
\end{align*}
Thus, in both cases, we have
$\val^{\sigma}(o^{1-j}) \ge 
\val^{\sigma}(c_j) - \frac{T}{2} + \frac{H_{d(C)} + L_{d(C)}}{2}$.

Since $\sigma(r^{1-j}_i) = o^{j}_{\inp(i)}$, we have by assumption that:
\begin{equation*}
\val^{\sigma}(r^{1-j}_i) \le \val^{\sigma}(c_j) + L_{d(C)} - \frac{T}{2}.
\end{equation*}
Thus, the appeal of switching $o^{1-j}_i$ to $r^{1-j}_i$ is at most:
\begin{align*}
L_{d(C)} - \frac{H_{d(C)} + L_{d(C)}}{2} 
&\le L_{d(C)} - \frac{L_{d(C)} + L_{d(C)}}{2} \\
&= 0.
\end{align*}
\qed
\end{proof}

The following lemma proves the eight claim of Lemma~\ref{lem:stage1appeal}.
Note that the first precondition of this lemma is satisfied by
Lemma~\ref{lem:gatelower}, and the second precondition is satisfied no matter
which action is chosen at $l^{1-j}_i$.
\begin{lemma}
\label{lem:eight}
Suppose that we are in phase $j$, and let $\sigma$ be a policy. 
If we have both:
\begin{align*}
\val^{\sigma}(r^{1-j}_i) &\ge \val^{\sigma}(c_j) - \frac{T}{2}\\
\val^{\sigma}(l^{1-j}_i) &\le \val^{\sigma}(c_{1-j}) + H_0\\
\end{align*} 
for some
input bit $i$, then
the appeal of switching $o^{1-j}_i$ to $l^{1-j}_i$ is at most $\bl$.
\end{lemma}
\begin{proof}
We can use our two assumptions, along with Lemma~\ref{lem:argappeal} to argue
that the appeal of switching $o^j_i$ to $r^j_i$ is at most: 
\begin{align*}
p_{3} \cdot (\frac{3T}{2} + H_0) = \bl.
\end{align*}
\qed
\end{proof}

\section{Proof of Lemma~\ref{lem:stage2appeal}}
\label{app:stage2appeal}

Most of the claims made in this lemma follow from lemmas that we have shown
previously. In particular, we have:
\begin{itemize}
\item The first part of this lemma follows from Lemma~\ref{lem:lto1j}.
\item The second part of this lemma follows form Lemma~\ref{lem:second-two}.
\item The third part of this lemma follows from Lemma~\ref{lem:third-two}.
\item Part 4a of this lemma follows from Lemma~\ref{lem:fourth-two}.
\item The fifth claim of the lemma follows from Lemma~\ref{lem:five}.
\item The sixth claim of this lemma follows from Lemma~\ref{lem:six}.
\item The eight part of this lemma follows from Lemma~\ref{lem:nine}.
\item The tenth claim of this lemma follows from Lemma~\ref{lem:ten}.
\end{itemize}
We now proceed to prove the rest of the claims.
The next lemma shows part 4b of Lemma~\ref{lem:stage2appeal}. Observe that
Lemma~\ref{lem:gateupperj} proves that the precondition of this lemma holds in a
stage 2 transition policy.

\begin{lemma}
\label{lem:fourth-three}
Suppose that we are in phase $j$, let $i$ be an input bit, and let $\sigma$ be a
policy with $\sigma(r^{1-j}_i) = c_{1-j}$ and $\val^{\sigma}(o^j_{\inp(i)}) \le
\val^{\sigma}(c_{1-j}) + H_{d(C)}$. The appeal of switching $r^{1-j}_i$ to
$o^j_{\inp(i)}$ is strictly less than $0$.
\end{lemma}
\begin{proof}
We have:
\begin{equation*}
\val^{\sigma}(r^{1-j}_i) = \val^{\sigma}(c_{1-j}) + L_0.
\end{equation*}
On the other hand, we have by assumption:
\begin{equation*}
\val^{\sigma}(o^j_{\inp(i)}) \le \val^{\sigma}(c_{1-j}) + H_{d(C)}.
\end{equation*}
Thus, we can apply Lemma~\ref{lem:argappeal} to argue that the appeal of
switching $r^{1-j}_i$ to $o^j_{\inp(i)}$ is at most:
\begin{equation*}
p_7 \cdot (H_{d(C)}- \frac{T}{2} - L_0) < 0. 
\end{equation*}
\qed
\end{proof}

The following lemma shows part 7 of Lemma~\ref{lem:stage2appeal}. Observe that
Lemma~\ref{lem:gateupperj} proves that the precondition on the value of
$r^{1-j}_i$ holds, even if the action to $o^j_{\inp(i)}$ is chosen at
$r^{1-j}_i$.

\begin{lemma}
Suppose that we are in phase $j$, let $i$ be an input, and let $\sigma$ be a
policy with $\sigma(l^{1-j}_i) = c_{1-j}$ and $\val^{\sigma}(r^{1-j}_i) \le
\val^{\sigma}(c_{1-j}) + L_0$. If $\sigma(o^{1-j}_i) = l^{1-j}_i$, then the
appeal of switching $o^{1-j}_i$ to $r^{1-j}_i$ is strictly less than $0$.
\end{lemma}
\begin{proof}
We have:
\begin{equation*}
\val^{\sigma}(o^{1-j}_i) = \val^{\sigma}(c_{1-j}) + H_0.
\end{equation*}
By assumption, we have: 
\begin{equation*}
\val^{\sigma}(r^{1-j}_i) \le \val^{\sigma}(c_{1-j}) + L_0
\end{equation*}
Thus, the appeal of switching $o^{1-j}_i$ to $r^{1-j}_i$ is at most $L_0 - H_0 <
0$. \qed
\end{proof}

Part 8 of this lemma follows from Lemma~\ref{lem:eight}. Note that the
precondition on the value of $r^{1-j}_i$ is satisfied by
Lemma~\ref{lem:gateupperj} in the case where $\sigma(r^{1-j}_i) =
o^{1-j}_{\inp(i)}$, and the precondition is obviously satisfied when
$\sigma(r^{1-j}_i) = c_{1-j}$. The precondition on the value of $l^{1-j}_i$ is
satisfied because $\sigma(l^{1-j}_i) = c_{1-j}$ for all stage 2 transition
policies.

\section{Proof of Lemma~\ref{lem:stage3appeal}}
\label{app:stage3appeal}

A few of the claims made in this lemma follow from results that we have already
shown. In particular:
\begin{itemize}
\item Part three of Lemma~\ref{lem:stage3appeal} follows from Lemma~\ref{lem:third-two}.
\item The ninth part of this lemma follows from Lemma~\ref{lem:nine}.
\item The tenth claim of this lemma follows from Lemma~\ref{lem:ten}.
\end{itemize}
We now proceed to prove the remaining claims.
Part 1a of this lemma follows from Lemma~\ref{lem:lto1j}. In the following
lemma, we show part 1b of Lemma~\ref{lem:stage3appeal}.
\begin{lemma}
\label{lem:first-two}
Suppose that we are in phase $j$, and let $\sigma$ be a policy with 
$\sigma(l^j_i) = c_{1-j}$. The appeal of switching $l^j_i$ to
$c_j$ is strictly less than $0$.
\end{lemma}
\begin{proof}
We have: 
\begin{equation*}
\val^{\sigma}(l^j_i) = \val^{\sigma}(c_{1-j}) - \frac{T}{2} + \frac{H_{d(C)} +
L_{d(C)}}{2}.
\end{equation*}
On the other hand, we have:
\begin{equation*}
\val^{\sigma}(c_j) = \val^{\sigma}(c_{1-j}) - T.
\end{equation*}
Thus, we can apply Lemma~\ref{lem:argappeal} to argue that the appeal of
switching $l^j_i$ to $c_j$ is:
\begin{equation*}
p_{5} \cdot (\frac{T}{2} + H_0 - \frac{H_{d(C)} + L_{d(C)}}{2}) < 0.
\end{equation*}
\qed
\end{proof}

The following lemma shows part 2 of Lemma~\ref{lem:stage3appeal}
\begin{lemma}
\label{lem:second-three}
Suppose that we are in phase $j$, and let $\sigma$ be a policy with
$\sigma(l^{1-j}_{i'}) = \sigma(r^{1-j}_{i'}) = c_{1-j}$ for every input bit
$i'$. For each input bit $i$, if $\sigma(r^j_i) = c_j$, then the appeal of
switching $r^j_i$ to $o^{1-j}_{\inp(i)}$ is at most \ro.
\end{lemma}
\begin{proof}
We have:
\begin{equation*}
\val^{\sigma}(r^j_i) = \val^{\sigma}(c_j) + L_0.
\end{equation*}
On the other hand, by Lemma~\ref{lem:gateupperfinal}, we have:
\begin{equation*}
\val^{\sigma}(o^{1-j}_{\inp(i)}) \le \val^{\sigma}(c_{1-j}) + H_{d(C)}.
\end{equation*}
Thus, we can apply Lemma~\ref{lem:argappeal} to argue that the appeal of
switching $r^j_i$ to $o^{1-j}_{\inp(i)}$ is at most:
\begin{equation*}
p_{7} \cdot (\frac{T}{2} + H_{d(C)} - L_0) = \ro.
\end{equation*}
\end{proof}

Part four of Lemma~\ref{lem:stage3appeal} follows from
Lemma~\ref{lem:fourth-three}. Note that the precondition on the value of
$o^{1-j}_{\inp(i)}$ is satisfied due to Lemma~\ref{lem:gateupperfinal}.

We now consider part 5 of Lemma~\ref{lem:stage3appeal}. If $\sigma(l^j_i) = c_j$
for input bit $i$, then our upper bound follows from Lemma~\ref{lem:five}.
The following lemma proves part 5 of Lemma~\ref{lem:stage3appeal} for the case
where $\sigma(l^j_i) = c_{1-j}$.
\begin{lemma}
Suppose that we are in phase $j$, and let $\sigma$ be a policy with
$\sigma(l^j_i) = c_{1-j}$ and $\sigma(r^j_i) = c_j$ and $\sigma(o^j_i) = l^j_i$
for some input bit $i$. The appeal of switching $o^j_i$ to $r^j_i$ is strictly
less than $0$. 
\end{lemma}
\begin{proof}
We have:
\begin{align*}
\val^{\sigma}(o^j_i) &= \val^{\sigma}(c_{1-j}) - \frac{T}{2} + \frac{H_{d(C)} + 
L_{d(C)}}{2} \\
&= \val^{\sigma}(c_{j}) + \frac{T}{2} + \frac{H_{d(C)} + L_{d(C)}}{2} \\
\end{align*}
On the other hand, we have:
\begin{align*}
\val^{\sigma}(r^j_i) = \val^{\sigma}(c_j) + L_0
\end{align*}
Thus, the appeal of switching $o^j_i$ to $r^j_i$ is:
\begin{equation*}
L_0 - (\frac{T}{2} + \frac{H_{d(C)} + L_{d(C)}}{2}) < 0.
\end{equation*}
\end{proof}

Part 6a of Lemma~\ref{lem:stage3appeal} follows from Lemma~\ref{lem:six}. The
following lemma proves part 6b of Lemma~\ref{lem:stage3appeal}.

\begin{lemma} 
Suppose that we are in phase $j$, and let $\sigma$ be a policy with
$\sigma(r^j_i) = c_j$, $\sigma(l^j_i) = c_{1-j}$, and $\sigma(o^j_i) = r^j_i$
for some input bit $i$. The appeal of switching $o^j_i$ to $l^j_i$ lies in the
range  $[\dnc, \dncu]$.
\end{lemma}
\begin{proof}
We have:
\begin{equation*}
\val^{\sigma}(o^j_i) = \val^{\sigma}(c_j) + L_0.
\end{equation*}
On the other hand, we have:
\begin{align*}
\val^{\sigma}(l^j_i) &= \val^{\sigma}(c_{1-j}) -\frac{T}{2} + \frac{H_{d(C)} + L_{d(C)}}{2}. \\
&= \val^{\sigma}(c_j) + \frac{T}{2} + \frac{H_{d(C)} + L_{d(C)}}{2}. \\
\end{align*}
Thus, we can apply Lemma~\ref{lem:argappeal} to argue that the appeal of
switching $o^j_i$ to $l^j_i$ is:
\begin{align}
\nonumber
&p_{3} \cdot (\frac{T}{2} + \frac{H_{d(C)} + L_{d(C)}}{2} - L_0) \\
\nonumber
&= p_{3} \cdot (\frac{T}{2} + \frac{H_{d(C)} + L_{d(C)}}{2} )  & [L_0 = 0] \\
\label{eqn:bitchange}
&= \frac{\bl \cdot (\frac{T}{2} + \frac{H_{d(C)} + L_{d(C)}}{2} )}{\frac{3T}{2}
+ H_0}.
\end{align}
For our lower bound, note that Equation~\eqref{eqn:bitchange} is greater than:
\begin{align*}
&\frac{\bl \cdot \frac{T}{2}}{\frac{3T}{2} + H_0}. \\ 
&=\frac{\bl \cdot \frac{1}{2} \cdot 3^{d(C) + 6}}{\frac{3}{2} \cdot 3^{d(C) + 6} + 3^{d(C) + 2}}. \\ 
&=\frac{\bl \cdot \frac{1}{2} \cdot 3^{4}}{\frac{3}{2} \cdot 3^{4} + 1}. \\ 
&=\frac{125.55}{122.5} \\
&> \dnc
\end{align*}
For our upper bound, note that Equation~\eqref{eqn:bitchange} is less than:
\begin{align*}
&= \frac{\bl \cdot (\frac{T}{2} + 2 \cdot 3^{d(C) + 2})}{\frac{3T}{2}
+ H_0}. \\
&= \frac{\bl \cdot (\frac{1}{2} \cdot 3^{d(C) + 6} + 2 \cdot 3^{d(C) +
2})}{\frac{3}{2} \cdot 3^{d(C) + 6} + 3^{d(C) + 2}}. \\
&= \frac{\bl \cdot (\frac{1}{2} \cdot 3^{4} + 2)}{\frac{3}{2} \cdot 3^{4} + 1}. \\
&= \frac{128.65}{122.5} \\
&\le \dncu
\end{align*}
\qed
\end{proof}

The following Lemma proves part 7 of Lemma~\ref{lem:stage3appeal}.

\begin{lemma}
\label{lem:seven-three}
Suppose that we are in phase $j$, and let $\sigma$ be a policy with
$\sigma(l^{1-j}_i) = \sigma(r^{1-j}_i) = c_{1-j}$ and with $\sigma(o^{1-j}_i) =
l^{1-j}_i$ for some input bit $i$. The appeal of switching $o^{1-j}_i$ to
$r^{1-j}_i$ is strictly less than~$0$.
\end{lemma}
\begin{proof}
We have:
\begin{equation*}
\val^{\sigma}(o^{1-j}_i) = \val^{\sigma}(c_{1-j}) + H_0.
\end{equation*}
On the other hand, we have:
\begin{equation*}
\val^{\sigma}(r^{1-j}_i) = \val^{\sigma}(c_{1-j}) + L_0.
\end{equation*}
Thus, the appeal of switching $o^{1-j}_i$ to $r^{1-j}_i$ is $L_0 - H_0 < 0.$
\qed
\end{proof}


The following Lemma proves Part 8 of Lemma~\ref{lem:stage3appeal}.

\begin{lemma}
\label{lem:eight-two}
Suppose that we are in phase $j$, and let $\sigma$ be a policy with
$\sigma(l^{1-j}_{i'}) = \sigma(r^{1-j}_{i'}) = c_{1-j}$ for every input bit
$i'$. If $\sigma(o^{1-j}_i) = r^{1-j}_i$ for some input bit $i$, then the appeal
of switching $o^{1-j}_i$ to $r^{1-j}_i$ is strictly less than $\magic$.
\end{lemma}
\begin{proof}
We have:
\begin{equation*}
\val^{\sigma}(o^{1-j}_i) = \val^{\sigma}(c_{1-j}) + L_0.
\end{equation*}
On the other hand, we have:
\begin{equation*}
\val^{\sigma}(l^{1-j}_i) = \val^{\sigma}(c_{1-j}) + H_0.
\end{equation*}
Thus we can apply Lemma~\ref{lem:argappeal} to argue that the appeal of
switching $o^{1-j}_i$ to $r^{1-j}_i$ is at most:
\begin{align*}
&p_3 \cdot (H_0 - L_0) \\
&= \frac{\bl \cdot (H_0 - L_0)}{\frac{3T}{2} + H_0} \\
&\le \frac{\bl \cdot H_0}{\frac{3T}{2}} \\
&\le \frac{\bl \cdot 2 \cdot 3^{d(C)+2}}{\frac{3}{2} \cdot 3^{d(C)+6}} \\
&= \frac{\bl \cdot 2}{121.5} \\
&<\magic.
\end{align*}
\end{proof}

\section{Proof of Lemma~\ref{lem:stage4appeal}}
\label{app:stage4appeal}

A number of the claims made in this lemma follow from results that we have
already shown. In particular, we have:
\begin{itemize}
\item Part 1 of Lemma~\ref{lem:stage4appeal} follows from Lemma~\ref{lem:first-two}.
\item Part 3 of Lemma~\ref{lem:stage4appeal} follows from Lemma~\ref{lem:third-two}.
\item Part 6 of Lemma~\ref{lem:stage4appeal} follows from Lemma~\ref{lem:seven-three}.
\item Part 7 of Lemma~\ref{lem:stage4appeal} follows from Lemma~\ref{lem:eight-two}.
\end{itemize}

Part 2a of Lemma~\ref{lem:stage4appeal} follows from
Lemma~\ref{lem:second-three}. The following lemma proves part 2b of
Lemma~\ref{lem:stage4appeal}.
\begin{lemma}
Suppose that we are in phase $j$, and let $\sigma$ be a policy with 
$\sigma(r^j_i) = c_j$ for some gate $i$, and for every gate $i$':
\begin{itemize}
\item If $i'$ is an \org gate then we have $\sigma(x^{1-j}_{i'}) = c_{1-j}$.
\item If $i'$ is a \notg gate then we have $\sigma(a^{1-j}_{i'}) = c_{1-j}$.
\end{itemize} 
The appeal of switching $r^j_i$ to $o_{\inp(i)}$ is at least \rol.
\end{lemma}
\begin{proof}
We have:
\begin{equation*}
\val^{\sigma}(r^j_i) = \val^{\sigma}(c_j) + L_0.
\end{equation*}
On the other hand, we can apply Lemma~\ref{lem:gatelowerfinal} to argue that:
\begin{equation*}
\val^{\sigma}(o^{1-j}_{\inp(i)}) \ge \val^{\sigma}(c_{1-j})
\end{equation*}
Thus, we can apply Lemma~\ref{lem:argappeal} to argue that the appeal of
switching $r^j_i$ to $o_{\inp(i)}$ is at least:
\begin{align*}
&p_{7} \cdot (T - \frac{T}{2} + L_0) \\
&= \frac{\ro(\frac{T}{2} + L_0)}{\frac{T}{2} + H_{d(C)} - L_0} \\
&= \frac{\ro(\frac{T}{2})}{\frac{T}{2} + H_{d(C)}} & [L_0 = 0] \\
&\le \frac{\frac{\ro}{2} \cdot 3^{d(C) + 6}}{\frac{1}{2} \cdot 3^{d(C) + 6} + 2
\cdot 3^{d(C) + 2}}\\ 
&= \frac{\frac{\ro}{2} \cdot 3^{4}}{\frac{1}{2} \cdot 3^{4} + 2
}\\ 
&= \ro \cdot \frac{40.5}{42.5} \\
&> \rol.
\end{align*}
\end{proof}

The following Lemma proves part 2c of Lemma~\ref{lem:stage4appeal}.
\begin{lemma}
Suppose that we are in phase $j$, and let $\sigma$ be a policy with 
$\sigma(r^j_i) = o^{1-j}_{\inp(i)}$ for some gate $i$, and for every gate $i$':
\begin{itemize}
\item If $i'$ is an \org gate then we have $\sigma(x^{1-j}_{i'}) = c_{1-j}$.
\item If $i'$ is a \notg gate then we have $\sigma(a^{1-j}_{i'}) = c_{1-j}$.
\end{itemize} 
The appeal of switching $r^j_i$ to $c_{j}$ is strictly less than $0$.
\end{lemma}
\begin{proof}
We can apply Lemma~\ref{lem:gatelowerfinal} to argue that
\begin{equation*}
\val^{\sigma}(r^j_i) \ge \val^{\sigma}(c_{1-j}) - \frac{T}{2} =
\val^{\sigma}(c_j) + \frac{T}{2}.
\end{equation*}
Thus, we can apply Lemma~\ref{lem:argappeal} to argue that the appeal of
switching $r^j_i$ to $c_j$ is at most:
\begin{equation*}
p_{6} \cdot (-\frac{T}{2} + L_0) < 0. 
\end{equation*}
\qed
\end{proof}

Part 4 of Lemma~\ref{lem:stage4appeal} follows from
Lemma~\ref{lem:fourth-three}. Note that the precondition on the value of
$o^j_{\inp(i)}$ is satisfied due to Lemma~\ref{lem:gateupperfinal}.

In the following lemma, we show part 5 of Lemma~\ref{lem:stage4appeal}.
\begin{lemma}
Suppose that we are in phase $j$, and let $\sigma$ be a policy with
$\sigma(l^{1-j}_{i'}) = \sigma(r^{1-j}_{i'}) = c_{1-j}$ for every input bit
$i'$. If $\sigma(l^j_i) = c_{1-j}$ and $\sigma(o^j_i) = l^j_i$ for some input
bit $i$, then the appeal of switching $o^j_i$ to $l^j_i$ is strictly less than
$0$.
\end{lemma}
\begin{proof}
We have:
\begin{equation*}
\val^{\sigma}(o^j_i) = \val^{\sigma}(c_{1-j}) + H_0.
\end{equation*}
If $\sigma(r^j_i) = c_j$, then we have:
\begin{equation*}
\val^{\sigma}(r^j_i) = \val^{\sigma}(c_j) + L_0 < \val^{\sigma}(c_{1-j}).
\end{equation*}
On the other hand, if $\sigma(r^j_i) = o^{1-j}_{\inp(i)}$, then we can apply
Lemma~\ref{lem:gateupperfinal} to argue that:
\begin{equation*}
\val^{\sigma}(r^j_i) \le \val^{\sigma}(c_{1-j}) + H_{d(c)} - \frac{T}{2} <
\val^{\sigma}(c_{1-j}).
\end{equation*}
Thus, in either case, we have $\val^{\sigma}(r^j_i) < \val^{\sigma}(c_j)$. So,
the appeal of switching $o^j_i$ to $l^j_i$ is $0 - H_0 < 0$.
\qed
\end{proof}

Part 9a of Lemma~\ref{lem:stage4appeal} follows from Lemma~\ref{lem:nine}. In
the following Lemma, we show part 9b of Lemma~\ref{lem:stage4appeal}.
\begin{lemma}
Suppose that we are in phase $j$, and let $\sigma$ be a policy with
$\sigma^{x^l_i} = c_{1-j}$ for some \org gate $i$ and some $l \in \{0,1\}$. The
appeal of switching $x^l_i$ to $c_j$ is strictly less than $0$.
\end{lemma}
\begin{proof}
We have:
\begin{equation*}
\val^{\sigma}(x^l_i) = \val^{\sigma}(c_{1-j}) = \val^{\sigma}(c_j) + T.
\end{equation*}
So, by Lemma~\ref{lem:argappeal}, the appeal of switching $x^l_i$ to $c_j$ is:
\begin{equation*}
\frac{0.8}{T} \cdot -T < 0.
\end{equation*}
\qed
\end{proof}

Part 10a of Lemma~\ref{lem:stage4appeal} follows from Lemma~\ref{lem:ten}. In
the following lemma, we show part 10b of Lemma~\ref{lem:stage4appeal}. 
\begin{lemma}
Suppose that we are in phase $j$, and let $\sigma$ be a policy with
$\sigma(a^{1-j}_i) = c_{1-j}$ for some \notg gate $i$. The appeal of switching
$a^{1-j}_i$ to $c_j$ is strictly less than $0$.
\end{lemma}
\begin{proof}
We have:
\begin{equation*}
\val^{\sigma}(a^{1-j}_i) = \val^{\sigma}(c_{1-j}).
\end{equation*}
Thus, we can apply Lemma~\ref{lem:argappeal} to argue that the appeal of
switching $a^{1-j}_i$ to $c_j$ is:
\begin{equation*}
p_1 \cdot (-2T + H_{d(i) - 1}) < 0. 
\end{equation*}
\qed
\end{proof}

\section{Proof of Theorem~\ref{thm:mdpend}}
\label{app:mdpend}

\begin{proof}
Let $\sigma'_{\text{init}}$ be a policy that agrees with $\sigma_{\text{init}}$
for every state in $\const(C)$, and that also has $\sigma'_{\text{init}}(b_1) =
\sink$. This will be the starting policy for $\const_2(C)$.

We first show that $b_2$ cannot be switched until after clock phase $2^n$. This
holds because, by Lemma~\ref{lem:argappeal}, the appeal of switching $b_2$ to
$b_1$ is $0.2$, but in the proof of Theorem~\ref{thm:actionswitch} all actions
switched before the end of phase $2^n$ have appeal strictly greater than $0.2$.
Note that, in the initial policy for the construction $\sigma'_{\text{init}}$, we have
$\val^{\sigma'_{\text{init}}}(l^0_i) \ge 0$ and 
$\val^{\sigma'_{\text{init}}}(r^0_i) \ge 0$ for all input bits $i$, and all $j
\in \{0, 1\}$. Thus, for
all policies $\sigma$ with $\sigma(b_2) = \sink$, we have that the appeal of
switching either $l^0_i$ or $r^0_i$ to $b_2$ is less than or equal to $0$. These
facts combine to show that policy iteration on $\const_2(C)$ will proceed in the
same way as policy iteration on $\const(C)$ until after the $2^n$-th clock
phase. In particular, policy iteration will pass through to the end of stage 4
phase transition at the end of the $2^n$-th clock phase.

Note that at the end of stage 4 of the phase transition, we have that the choice
made at $o^0_z$ determines the $z$-th bit of $F^{2^n}(B^I)$. Once we arrive at
the end of stage 4, we can then apply Lemma~\ref{lem:stage4appeal} to argue that
Dantzig's rule will switch $b_2$ to $b_1$, before it switches $o^0_z$. This is
because all relevant appeals mentioned in Lemma~\ref{lem:stage4appeal} are
strictly less than $0.2$, and because there are no switchable actions in the
clock. Once $b_2$ has switched to $b_1$, Dantzig's rule will immediately switch
$l^0_z$ and $r^0_z$ to $b_1$. This holds by construction, because the value of
$l^0_z$ and $r^0_z$ can be at most $W$, and since the value of $b_1$ is now $2
\cdot W$, the appeal of switching these actions must be at least $W$.

Let $\sigma$ be a policy in which $\sigma(b_2) = b_1$ and $\sigma(l^0_z) =
\sigma(r^0_z) = b_1$. Let $a$ be the action not chosen by $\sigma$ at $o^0_z$.
By construction, we have that $\appeal(a) = 0$. Thus, policy iteration can never
switch $a$. Therefore, once $l^0_z$ and $r^0_z$ have been switched to $b_2$, the
policy at $o^0_z$ will not be changed, even if policy iteration makes switches
elsewhere in the construction. So, in the optimal policy we will have that
the strategy decision made at $o^0_z$ will determine the $z$-th bit of
$F^{2^n}(B^I)$.
\qed
\end{proof}

\end{document}